\documentclass[twoside,11pt]{article}

\usepackage{amsmath}
\usepackage{graphicx}
\usepackage{enumerate}
\usepackage{natbib}
\usepackage{url} 

\usepackage{jmlr2e}
\usepackage{color}
\usepackage{natbib}
\usepackage{amsmath,amssymb,amsfonts,bbm}
\usepackage{epic,eepic,epsfig,longtable}
\usepackage{multirow,verbatim}
\usepackage{esvect}
\usepackage{array}
\usepackage{algorithm}
\usepackage{multirow}
\usepackage{footnote}
\usepackage{threeparttable} 
\usepackage{epsfig}
\usepackage{setspace}
\usepackage{enumitem}
\usepackage{color}
\usepackage{fancyhdr}
\usepackage{appendix}
\usepackage{listings}
\usepackage{longtable}
\usepackage{url}
\usepackage{lineno,hyperref}
\usepackage{subfigure}
\usepackage{graphicx}
\usepackage{booktabs}
\usepackage{longtable}
\usepackage{siunitx}
\usepackage{mathbbol,bbm}
\usepackage{float}
\usepackage{multirow}
\usepackage{epsfig}
\usepackage{setspace}
\usepackage{color}
\usepackage{fancyhdr}
\usepackage{appendix}
\usepackage{listings}
\usepackage{longtable}
\usepackage{url}
\usepackage{lineno,hyperref}
\usepackage{subfigure}
\usepackage{mathrsfs}
\usepackage{stmaryrd}
\usepackage{appendix}
\usepackage{algorithm}
\usepackage{enumitem}
\usepackage{algpseudocode}
\usepackage{caption}
\usepackage{amsmath}
\usepackage{graphics}
\usepackage{epsfig}
\usepackage{footnote}
\usepackage{color}
\usepackage{threeparttable}

\definecolor{darkred}{RGB}{150,50,50}
\definecolor{brown}{RGB}{250,100,100}
\definecolor{green}{RGB}{000,150,100}
\definecolor{purple}{RGB}{250,000,180}

\def\blue{\color{black}}

\def\hbar{\bar{h}}

\definecolor{red2}{rgb}{0.7, 0, 0.1}

\def\biggiven{\,\big{|}\,}

\def\bv{\mathbf{v}}

\def\argmindum{\mathop{\mbox{argmin}}}
\def\argmin#1{\argmindum_{#1}}

\definecolor{purple}{rgb}{0.84, 0.17, 0.89}
\definecolor{darkblue}{rgb}{0,0,0.6}
\definecolor{darkblue1}{rgb}{0,0,0.7}
\definecolor{darkred}{rgb}{0.6,0,0}
\definecolor{grey}{rgb}{0.6,0.6,0.5}

\def\bx{\mathbf{x}}

\def\Uscr{\mathscr{U}}
\def\Fscr{\mathscr{F}}
\def\Lscr{\mathscr{L}}

\def\btheta{\boldsymbol{\theta}}
\def\suponetrans{^{\scriptscriptstyle \sf (1) T}}

\def\supmtrans{^{\scriptscriptstyle \sf {\tiny (}m{\tiny)} T}}

\def\supone{^{\scriptscriptstyle \sf (1)}}

\def\supm{^{\scriptscriptstyle \sf \text{\tiny(}m\text{\tiny)}}}

\def\independenT#1#2{\mathrel{\rlap{$#1#2$}\mkern2mu{#1#2}}}
\newcommand{\indp}{\protect\mathpalette{\protect\independenT}{\perp}}

\def\bX{\mathbf{X}}

\def\subsash{_{\scriptscriptstyle \sf SASH}}
\def\subipd{_{\scriptscriptstyle \sf IPD}}

\def\bone{\mathbf{1}}
\def\bzero{\mathbf{0}}

\newtheorem{assumption}{Assumption}

\newcommand{\beq}{\begin{equation}}
\newcommand{\eeq}{\end{equation}}
\newcommand{\beas}{\begin{eqnarray*}}
\newcommand{\eeas}{\end{eqnarray*}}
\newcommand{\bea}{\begin{eqnarray}}
\newcommand{\eea}{\end{eqnarray}}

\newcommand{\bet}{\begin{theorem}}
\newcommand{\eet}{\end{theorem}} 
\newcommand{\bel}{\begin{lemma}}
\newcommand{\eel}{\end{lemma}}
\newcommand{\bep}{\begin{proposition}}
\newcommand{\eep}{\end{proposition}}
\newcommand{\bed}{\begin{definition}}
\newcommand{\eed}{\end{definition}}
\newcommand{\bec}{\begin{corollary}}
\newcommand{\eec}{\end{corollary}}
\newcommand{\bex}{\begin{example}}
\newcommand{\eex}{\end{example}}

\usepackage{color}

\newcommand{\bei}{\begin{itemize}}
\newcommand{\eei}{\end{itemize}}
\newcommand{\ben}{\begin{enumerate}}
\newcommand{\een}{\end{enumerate}}
\def\0{\boldsymbol{0}}

\def\bX{\mathbf{x}}

\def\bX{\boldsymbol{X}}
\def\R{\mathbb{R}}

\def\bX{\boldsymbol{x}}

\def\Bcal{\mathcal{B}}

\def\d{\mathrm{d}}

\newcommand{\Rbb}{\mathbb{R}}

\newbox\TempBox \newbox\TempBoxA

\def\trans{^{\scriptscriptstyle \sf T}}

\def\text#1{\mbox{\sf #1}}

\def\bX{\mathbf{X}}

\def\bbeta{\boldsymbol{\beta}}
\def\boeta{\boldsymbol{\boeta}}

\def\bbetahat{\widehat{\bbeta}}
\def\Lscr{\mathscr{L}}
\def\betahat{\widehat{\beta}}

\def\bep{\boldsymbol{\epsilon}}

\def\bgamma{\boldsymbol{\gamma}}
\def\bgammahat{\widehat{\bgamma}}
\def\bgammatilde{\widetilde{\bgamma}}

\def\Xihat{\widehat{\Xi}}
\def\bv{\boldsymbol{v}}

\def\bthetahat{\widehat{\btheta}}

\def\bthetahat{\widehat{\btheta}}
\def\subsup{_{\scriptscriptstyle \sf sup}}

\def\bOmegahat{\widehat{\boldsymbol{\Omega}}}

\def\fhat{\widehat{f}}
\def\sumiNm{\sum_{i=1}^{N_m}}

\def\betahat_1{\widehat{\beta}_1}

\def\alphahat{\widehat{\alpha}}
\def\Pbb{\mathbb{P}}
\def\Ebb{\mathbb{E}}

\def\E{\mathbb{E}}

\usepackage{mathtools}

\def\df{{\rm df}}
\def\bW{\mathbb{W}}
\def\Pbbhat{\widehat{\Pbb}}
\def\x{\boldsymbol{x}}

\newcommand{\citeg}[1]{\citep[e.g.,][]{#1}}

\jmlrheading{27}{2026}{1-50}{9/23; Revised
7/26}{7/26}{23-1235}{Yue Liu, Molei Liu, Zijian Guo and Tianxi Cai}
\ShortHeadings{SASH}{Liu et al.}

\firstpageno{1}

\begin{document}

\title{Efficient Modeling of Surrogates to Improve Multi-source High-dimensional Integrative Regression}

\author{\name Yue Liu$^*$ \email yueliu@fas.harvard.edu \\
       \addr Department of Statistics\\
       Harvard University\\
       Cambridge, MA 02138, USA
       \AND
       \name Molei Liu$^*$ \email moleiliu@bjmu.edu.cn \\
       \addr Peking University Health Science Center;\\
       Beijing International Center for Mathematical Research;\\
       Center for Statistical Science\\
       Peking University \\
       Haidian District, Beijing 100084, China
       \AND
       \name Zijian Guo \email zijguo@zju.edu.cn\\
       \addr Center for Data Science\\
       Zhejiang University\\
       Xixi District, Hangzhou 310058, China
       \AND
       \name Tianxi Cai \email tcai.hsph@gmail.com \\
       \addr Department of Biostatistics\\
       Harvard Chan School of Public Health\\
       Boston, MA 02115, USA
       }

\editor{ }

  \maketitle

\footnotetext[1]{Yue Liu and Molei Liu contributed equally. Correspondence to: Molei Liu (moleiliu@bjmu.edu.cn).}

\begin{abstract}

\noindent Surrogate variables play an important role in various fields due to the scarcity or absence of gold-standard labels. We develop a novel approach named SASH for {\bf S}urrogate-{\bf A}ssisted and data-{\bf S}hielding {\bf H}igh-dimensional integrative regression. It is a semi-supervised approach that efficiently leverages sizable unlabeled samples with error-prone surrogate outcomes from multiple local sites to improve model estimation using the small gold-labeled sample. To facilitate stable and efficient knowledge extraction from the surrogates, our method first obtains a preliminary supervised estimator, and then uses it to assist in training a regularized single-index model (SIM) for the surrogates. Interestingly, through a chain of convex and properly penalized sparse regressions that approximate the SIM loss using bias correction, our method avoids the problem of local minima in the SIM, and fully eliminates the impact of the preliminary estimator's excessive error. In addition, it protects individual-level information through the aggregation of summary statistics from local sites, leveraging a similar idea of bias-corrected approximation. Through simulation studies, we demonstrate that our method outperforms existing approaches. Finally, we apply our method to develop a genetic risk model for type 2 diabetes using large-scale data sets from UK and Mass General Brigham biobanks, where only a small fraction of subjects in one site are labeled through manual chart review.

\end{abstract}


\section{Introduction}

\subsection{Background}

Electronic health records (EHRs) have become a crucial resource for a growing number of data-driven biomedical studies including disease diagnosis \citep{beaulieu2020examining,zhang2020maximum}, clinical knowledge extraction \citep{harnoune2021bert,hong2021clinical}, treatment response evaluation \citep{franklin2021real}, and genotype-phenotype association characterization when linked with biobank data sets \citep{kohane2011using,cai2018association,maiorino2023phenomics}. However, realizing the potential of EHRs in data-driven research is challenging. This is because surrogate variables, such as diagnostic codes, are widely available but often noisy, failing to accurately reflect the true disease status and introducing bias \citep{zhang2019high,geva2021high}. Manual chart review, while accurate, is labor-intensive and yields only a small number of gold-standard labels, exacerbating challenges in analyzing high-dimensional data. This creates a trade-off: surrogates reduce variance due to their large sample sizes but introduce bias, whereas gold labels are unbiased but inflate variance due to their limited availability. To balance this trade-off, robust semi-supervised methods are needed to integrate large, noisy surrogate data sets with small samples with gold-standard labels efficiently.

{ 
In addition to EHR studies, surrogate variables also play an important role in diverse application fields. For example, in time-series forecasting, surrogates often represent intermediate data points collected between the observed present and the predicted future. They are useful for improving predictions in social and economic contexts by bridging the gap between initial and final observations \citep{box2015time}. In addition, various types of surrogates have been introduced for long-term user experience characterization in recommendation systems \citep{wang2022surrogate}, automated vehicle safety modeling \citep{wang2021review}, and purchasing behavior prediction on e-commerce platforms \citep{wu2018turning}.}

When combining data across institutions, privacy constraints further complicate matters, as individual-level data cannot be directly shared. Frameworks like DataSHIELD \citep{wolfson2010datashield} allow only the transfer of summary statistics, which, along with high-dimensional data, intensifies the issue of noisiness and label scarcity \citeg{jones2012datashield,doiron2013data,gaye2014datashield}. Our real-world study aims to address the combined challenges of label scarcity, high-dimensionality, and privacy constraints.


\subsection{Problem Setup}\label{sec:problem}

Suppose individual-level data are stored at $M$ local sites indexed by $1,2,\ldots, M$. In each site $m\in\{1,2,\ldots, M\}$, there are $N_m$ independent and identically distributed subjects with underlying full data $\Fscr\supm = \{(Y_i\supm, S_i\supm, \bX_i\supmtrans)\trans:i = 1,2,\ldots, N_m\}$, where $\bX_i\supm$ is the $p$-dimensional covariate vector of subject $i$ in site $m$, $Y_i\supm$ denotes the binary true outcome of interest, and $S_i\supm$ is some imperfect surrogate or proxy of $Y_i\supm$. 

We assume that the true label $Y_i\supm$ is only observed for $m=1$ and $i=1,2, ...,n$ with $n \ll \min(N_1,p)$ and denote the labeled subset of subjects in site $1$ by $\Lscr\supone = \{(Y_i\supone, S_i\supone, \bX_i\suponetrans)\trans:i = 1,2,\ldots, n\}$, while for the remaining subjects in site $1$ and all subjects from sites $2,\ldots,M$, we only observe their covariates and surrogates, and let $\Uscr\supm = \{(S_i\supm, \bX_i\supmtrans)\trans: i = 1,2,\ldots, N_m\}$ denote the observed unlabeled data for $m=1,2,\ldots,M$. For the true outcome $Y_i\supm$, we are interested in estimating the coefficients in the logistic risk model
\begin{equation}
\Pbb(Y_i\supm=1 \mid \bX_{i}\supm)=g({ \alpha_{0}\supm}+\bbeta_0\trans\bX_i\supm) \ \mbox{with $g(a)=e^a/(1+e^a)$},\quad \mbox{for}\  1\leq m\leq M,
\label{equ:asu:1}
\end{equation}
where $\alpha_0\supm$ is the intercept for site $m$ and $\bbeta_0$ represents the high-dimensional sparse model coefficients shared by the local sites. One could simply use $\ell_1$-regularized logistic regression on the labeled data $\Lscr\supone$ to obtain a supervised estimator of $\bbeta_0$. However, this may not perform well due to the small size of $\Lscr\supone$ and the high-dimensionality.

To estimate and infer $\bbeta_0$ more efficiently, we leverage the large unlabeled samples from all sites with the surrogate $S_i\supm$ that is predictive of $Y_i\supm$ but imperfect and error-prone. { Our main model assumption for the surrogate outcome $S\supm_i$ is that it relates to $\bX_i\supm$ through a single index model (SIM):
\begin{equation}
\Ebb(S\supm_{i}\mid \bX\supm_{i} )=\breve f_m(\bbeta_0\trans \bX\supm_i)= f_{m}(\bgamma_0\trans \bX\supm_i ) \quad \text{for} \quad 1\leq m\leq M,
\label{eqn:sim}
\end{equation} 
where $\bgamma_0$ is a scaled vector with the {\em same direction} as $\bbeta_0$ and $\breve f_m$ and $f_m$ are some unknown and nuisance link functions varying across different sites.

The SIM assumption (\ref{eqn:sim}) holds in various practical scenarios such as the post hoc surrogate illustrated in diagram (i) on the right panel of Figure \ref{fig:dag}, which is frequently used in practice \citeg{hong2019semi}. In this case, $S_i\supm$ is a proxy coming directly from the disease outcome $Y_i\supm$ appearing between $S_i\supm$ and the baseline risk factors $\bX_i\supm$. For example, in our real-world study, $\bX_i\supm$ is genetic variants and $S_i\supm$ is the total count of diagnostic codes for the disease $Y_i\supm$. The diagnostic counts $S_i\supm$ are related to the disease-causing genetic variants only through the disease status, leading to the conditional independence:
\begin{equation}
S\supm_i\indp \bX_i\supm\mid Y\supm_i\quad \mbox{for} \quad m=1,2,\ldots,M.
\label{eqn:asu:indp}
\end{equation}
In Proposition \ref{prop:1} (justified in Appendix \ref{sec:proof:prop}), we show that our SIM assumption (\ref{eqn:sim}) holds under the logistic model assumption (\ref{equ:asu:1}) and the conditional independence condition (\ref{eqn:asu:indp}) that is reasonable for post hoc surrogates.
\begin{proposition}
The surrogate SIM assumption (\ref{eqn:sim}) holds under conditions (\ref{equ:asu:1}) and (\ref{eqn:asu:indp}).
\label{prop:1}
\end{proposition}}

Assumption (\ref{eqn:sim}) also holds in other contexts such as the early-endpoint surrogate described by diagram (ii) in Figure \ref{fig:dag} and the example (iii). Early-endpoint surrogates in (ii) are frequently encountered in contemporary clinical and biomedical studies \citep{prentice1989surrogate,vanderweele2013surrogate}. They encompass factors that can be rapidly assessed and used as an indicator of the long-term outcome $Y$. For instance, the tumor response rate could be used as a surrogate for overall survival in clinical studies. Similar to Proposition \ref{prop:1}, it is not hard to justify that such types of surrogates in (ii) and (iii) also follow the SIM (\ref{eqn:sim}) under the outcome model (\ref{equ:asu:1}).

Note that $\bbeta_0$ can be identified from \eqref{eqn:sim} only up to scalar multiples. To ensure identifiability, we assume that $\beta_{01}\neq 0$ and take $\bgamma_0=\bbeta_0/\beta_{01}$, consequently $\gamma_{01}=1$. For training stability, one may set the first covariate in $\bX$ as the risk factor believed to have a strong association with the outcome so that $\beta_{01}$ is away from zero. This requires moderate prior knowledge of the model. For example, in our T2D genetic risk study, we set $X_1$ as the ethnicity variable since it has been found and confirmed that the risk of T2D is significantly associated with ethnicity \citep{abate2003impact}.

Under the framework introduced above, we aim to derive an efficient estimator for $\bbeta_0$ through semi-supervised learning assisted by the surrogate while overcoming the DataSHIELD constraint \citep{wolfson2010datashield} that only summary statistics like mean vectors and covariance matrices are allowed to be transferred across the local sites. In our method, we will only transfer summary-level data between site $1$ (with a few samples of $Y$) and the other sites. We illustrate this data-sharing scheme as well as the structure of our observed data sets in the left panel of Figure \ref{fig:dag}.

\begin{figure}[htb!]
\centering
\includegraphics[width=0.52\textwidth]{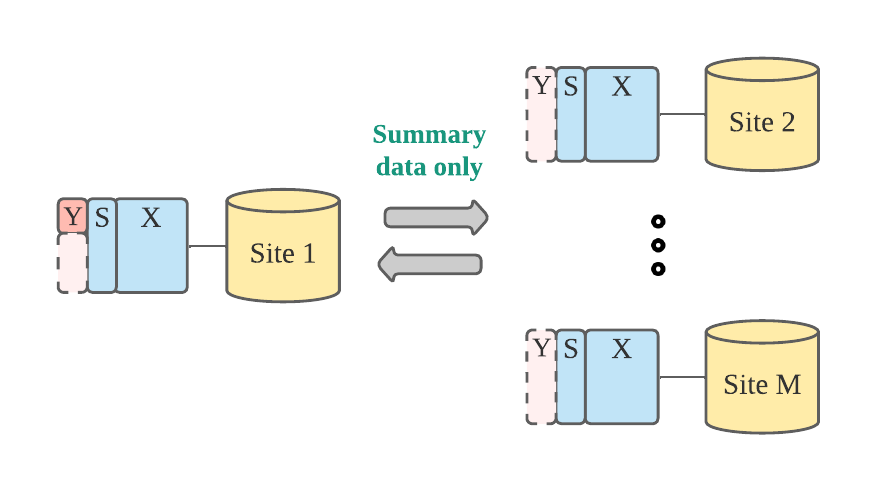}
\includegraphics[width=0.42\textwidth]{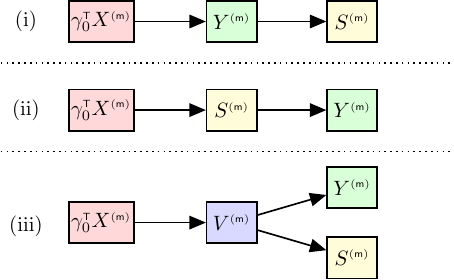}
\caption{\label{fig:dag} {Left panel: an illustration of the structure of the observed data and the scheme of communication under DataSHIELD. Right panel: directed acyclic graphs (DAGs) of data generation examples. (i) Post hoc surrogates; (ii) Early-endpoint surrogates; (iii) An illustrative example in which the surrogate is neither post hoc nor early-endpoint but our SIM assumption (\ref{eqn:sim}) still holds. Here $V\supm$ is some key latent factor mediating the effects of $X\supm$ on both $S\supm$ and $Y\supm$.}}
\end{figure}

\subsection{Related Literature}\label{sec:intro:lit}

{
Surrogates play an important role in various application fields where the collection of the primary or true outcome of interest is expensive, e.g., requiring intensive human efforts or long-term follow-up. With a partially observed true outcome $Y$ and fully observed surrogate $S$, surrogate-assisted semi-supervised (SAS) learning methods have been proposed to leverage unlabeled data to improve the learning accuracy of the true outcome model in the paucity or even absence of the gold-standard label $Y$. The key is to model and incorporate $S$ in a way that could reduce the variance without introducing bias. For example, \cite{huang2018pie}, \cite{zhang2020maximum}, and \cite{hong2019semi} used a maximum likelihood framework to address error-prone or anchor-positive labels. \cite{kallus2020role} studied the semiparametric efficiency bounds for estimating the average treatment effects (ATE) in the presence of surrogates, to demonstrate the usefulness of surrogates in causal inference. \cite{athey2019surrogate} and \cite{athey2020combining} proposed methods that leverage observational data and secondary surrogates to mitigate the missingness of the long-term primary outcome $Y$ in experimental data. \cite{lin2021learning} addressed the noisiness of surrogate outcomes in deep learning through a novel algorithm curbing memorization of the noisy instances in neural networks.

For high-dimensional data, \cite{zhang2022prior} proposed an adaptive SAS sparse regression approach that effectively leverages valid and informative surrogates while maintaining robustness to the surrogates of poor quality. They directly modeled $S$ against a linear function of $\bX$, which is valid for the SIM only when $\bX$ satisfies a restrictive Gaussian condition. They also aim to model $Y$ against both $S$ and $\bX$ for the purpose of EHR phenotyping. In contrast, we are interested in $Y\sim \bX$ only, which is for risk prediction with the baseline factors $\bX$, as well as identifying key risk factors. Recently, \cite{hou2021efficient} introduced a working imputation model for $Y$ including both $\bX$ and $S$ as predictors and deployed it on unlabeled data to estimate the high-dimensional sparse $Y\sim \bX$. Their method is guaranteed to improve statistical accuracy when the imputation model is sparser than $Y\sim \bX$. To the best of our knowledge, there is no existing SAS method that is effective under a flexible SIM in (\ref{eqn:sim}) with a non-Gaussian design $\bX$, let alone one that integrates multi-source SIMs under the DataSHIELD constraint. As will be discussed later, this methodological gap is due to the statistical and computational challenges of sparse SIM under non-Gaussian designs.

}

With the promise of more generalizable research findings through integrative analysis of multi-institutional data, much research in recent years has focused on federated learning that achieves information integration without sharing individual-level data. 
For example, \cite{he2016sparse} utilized this strategy for sparse meta-analysis assuming that the model coefficients from local sites have similar supports. Under high-dimensional settings, \cite{lee2017communication} and \cite{battey2018distributed} conducted communication-efficient distributed learning for the high-dimensional model by averaging and truncating the local debiased Lasso estimators \citeg{van2014asymptotically}. \cite{cai2021individual} proposed a federated learning method that accommodates model heterogeneity across local sites and complies with the DataSHIELD constraint by aggregating bias-corrected summary statistics. \cite{jordan2018communication} proposed an approximate likelihood
approach for communication-efficient aggregation of homogeneous data under the DataSHIELD constraint and \cite{duan2022heterogeneity} extended their framework to address distributional heterogeneity. 
Nevertheless, all these methods require observing the outcome $Y$ in all local sites. In our current set-up with only a few gold labels available in one site and abundant samples with the error-prone surrogate $S$ in all sites, directly implementing existing methods for $Y$ on the error-prone $S$ can result in severe bias due to the error of $S$ and the heterogeneity of $Y\mid S$ across the sites.


Technically, our work is also relevant to the estimation and inference of high-dimensional sparse SIM. Due to the unknown link function, it is much more challenging to study sparse SIM under high-dimensional settings. A common strategy is to employ surrogate loss functions. Specifically, \cite{neykov2016l1} and \cite{eftekhari2021inference} used $\ell_1$-regularized linear regression to learn the sparse coefficients in SIM. The validity of such a surrogate-loss approach heavily relies on the Gaussian (elliptical) distribution assumption of the design, which is often violated in practice. Under SIM with a non-Gaussian design, \cite{radchenko2015high} used splines to model the unknown link function and gradient-based procedures to estimate coefficients. To the best of our knowledge, no existing work ensures convergence of the optimization algorithm for sparse SIM with a non-Gaussian design matrix. 

{ Semiparametric efficient score functions have been frequently used to correct regularization bias in various inference problems such as those of high-dimensional generalized linear models (GLM) \citep{zhang2014confidence,van2014asymptotically,javanmard2014confidence}, treatment and structural effects \citeg{chernozhukov2016double}, and deep neural networks \citeg{farrell2021deep}. Recently, \cite{wu2021model} developed a debiasing method to infer the coefficient of the high-dimensional sparse SIM, which is motivated by the semiparametric efficient score for SIM coefficients derived in \cite{liang2010estimation}. 

More complicated than earlier work, the debiasing form of \cite{wu2021model} can simultaneously remove the first-order regularization bias of the parametric part and the smoothing bias of the nonparametric link function. This is because the unknown link in SIM is a nuisance model for the target coefficients. Our work essentially leverages this property to construct the one-step approximation for the non-convex SIM loss, which enables both the iterative updates towards the global solution and efficient data aggregation under the privacy constraints. A similar intrinsic relationship between debiased inference and efficient one-step approximation has been discussed in the context of high-dimensional GLMs \citep{cai2021individual}, but not studied for the SIM with a larger technical complication.

}

\subsection{Our Contribution}
To address limitations of existing approaches for sparse SIM and effectively leverage the surrogates to improve the learning efficiency on the outcome model in a multi-source setting, we propose a novel {\bf S}urrogate-{\bf A}ssisted and data-{\bf S}hielding {\bf H}igh-dimensional integrative regression (SASH) method. {First, it extracts a preliminary estimator with the small labeled samples only, to initiate the training of high-dimensional SIM (\ref{eqn:sim}) at each local site. Interestingly, our use of the preliminary supervised estimator plays an important role in ensuring the desirable convergence of the sparse SIM regression while the impact of its error can be fully removed through proper iteration and regularization. Then SASH derives proper bias-corrected summary statistics to approximate the SIM loss in local sites, and aggregates them to obtain an accurate estimator for the direction coefficients of $\bbeta$. Finally, it refits to the labeled data with the estimated direction to obtain the estimator for $\bbeta$.

We show that under mild and reasonable assumptions, our SASH estimator attains the optimal convergence rate, outperforms the supervised and local analyses, and incurs no first-order convergence-rate loss relative to the ideal IPD estimator obtained by pooling the individual data. We also demonstrate that our method achieves better performance than recent SAS methods in simulation and real-world studies. This is because the above-reviewed SAS learning methods \citeg{hou2021efficient,zhang2022prior} incorporate the surrogate $S$ through fully parametric regression or unstable SIM-training procedures that fail to characterize the semiparametric SIM (\ref{eqn:sim}) effectively and consequently cause efficiency loss. Technically, this limitation is mainly due to the convergence problem caused by local minima in single-index regression under non-Gaussian designs, which is even more pronounced in high-dimensional settings.

{ To overcome the convergence issue of the sparse SIM, we introduce a preliminary supervised estimator to guide the training, followed by two key innovative components: an efficient (bias-corrected) one-step approximation for the loss of SIM and an iterative updating algorithm with carefully tuned parameters. The resulting semi-supervised estimator achieves an error rate independent of the small labeled sample size $n$, instead leveraging the large unlabeled data size. However, the preliminary estimator driven by the small $n$ plays a central role in guiding our iterative algorithm and protecting it from falling into local minima at the very beginning. Importantly, we justify that an $\ell_2$-consistent preliminary estimator, guaranteed by the mild sparsity and regularity assumptions, is sufficient for our method to achieve the same convergence rate as the global optimum without reliance on the small $n$.} At first glance, this seems to be just a simple improvement of the Newton-type optimization. However, our strategy of addressing the nonparametric estimator in SIM via orthogonality, as well as adapting the regularization to control the excessive approximate errors, is actually more sophisticated than approaches in the existing literature and essential for eliminating the impact of the high-error preliminary estimator and ensuring the optimal convergence rate. Our work is also the first one to achieve an efficient multi-site aggregation for high-dimensional sparse SIM under the DataSHIELD constraint. This advancement is also based on our novel bias-corrected approximation procedure for SIM.

}


\section{Method}\label{sec:method}
\subsection{Outline of SASH}\label{sec:method:out}

Let $\Pbbhat\supm$ denote the empirical probability measure with observed data in $\Uscr\supm$ and $\Pbbhat_n\supone$ denote the empirical probability measure on the labeled data set $\Lscr\supone$. For a vector $\bv=(v_{1},v_2, \ldots, v_{d})\trans \in \mathbb{R}^{d}$ and $q>0$, we denote the $\ell_q$ norm of $\bv$ as $\|\bv\|_{q}=\big(\sum_{i=1}^{d}|v_{i}|^{q}\big)^{1 / q}$. For ease of notation, we drop the superscript in $\alpha\supone$ and denote it as $\alpha$ throughout the remainder of the paper since the intercept terms in sites $2,3,\ldots,M$ do not appear in our formulation. We write $a_n=O(b_n)$ or $a_n \lesssim b_n$ if $a_n \leq Cb_n$ for some constant $C>0$, 
write $a_n=o(b_n)$ if
$\lim_{n \rightarrow \infty} {a_n}/{b_n}=0$, 
and write $a_{n} \asymp b_{n}$ if $C \leq a_{n} / b_{n} \leq C^{\prime}$ for some constants $C, C^{\prime}>0$. Also, we write $a_n=O_p(b_n)$ and $a_n=o_p(b_n)$ respectively for $a_n=O(b_n)$ and $a_n=o(b_n)$ with probability approaching $1$.

\begin{algorithm}[htb!]
\caption{\label{alg:1} Outline of SASH.}

    {{\bf Input:} $\Lscr\supone$ and $\{\Uscr\supm, m=1,...,M\}$.}

\vspace{0.15cm}

    \textbf{Step I}: Obtain an initial estimator of the direction vector $\bgamma_0$ using $\Lscr\supone$ through supervised learning \eqref{eqn:step1}, and broadcast the direction estimator $\bgammatilde\subsup$ to all local sites. 
    
\vspace{0.15cm}

    \textbf{Step II}: In each local site, fit the SIM to estimate the direction $\bgamma_0$ using the preliminary direction estimator from site $1$ as shown in Algorithm \ref{alg:2}. Then derive summary statistics to approximate the individual-level loss function, and transfer them to site $1$.
    
    \vspace{0.15cm}
    
    \textbf{Step III}: Aggregate in site $1$ the summary statistics from the local sites to obtain the final estimator of $\bbeta_0$.
\end{algorithm}

In Algorithm \ref{alg:1}, we outline SASH, a {\em three-step} learning approach that effectively leverages the unlabeled data $\{\Uscr\supm , m = 1, ..., M\}$ and the labeled subset $\Lscr\supone$. Our identification strategy is to decompose the target $\bbeta_0$ into two parts, its direction $\bgamma_0$ and magnitude $\beta_{01}$. We first estimate $\bgamma_0$ precisely through fitting model (\ref{eqn:sim}), the SIM of $S\supm_i\sim\bX\supm_i$, with the large samples $\{\Uscr\supm , m = 1, ..., M\}$ via federated SIM learning. Crucially, this procedure is supported by a preliminary supervised estimator derived with the labeled data $\Lscr\supone$. Then for the one-dimensional magnitude $\beta_{01}$, we estimate it using $\Lscr\supone$ and the estimate of $\bgamma$. The details of the main steps are described in Section \ref{sec:method:main}.

\subsection{Details of the main steps}\label{sec:method:main}

\paragraph{Step I.} We first obtain an initial estimator by fitting a supervised penalized logistic model using the small labeled samples in site $1$:
\begin{equation}\label{eqn:step1}
\{\widetilde{\alpha}\subsup,\widetilde{\bbeta}\subsup\} = \argmin{\alpha\in \R,\bbeta \in \R^{p}} 
\Pbbhat_n\supone \ell\{Y,g(\alpha+\bbeta\trans \bX)\}
+ \lambda\subsup \|\bbeta_{\text{-}1}\|_1, 
\end{equation}
where $\ell(y,x)=-\{y\log x + (1-y)\log(1-x)\}$ and $\lambda\subsup>0$ is a tuning parameter. Since we assume that $\beta_{01}\neq 0$, we impose the Lasso penalty only on $\bbeta_{\text{-}1}=(\beta_2,\ldots,\beta_p)\trans$ in (\ref{eqn:step1}), thereby avoiding direct shrinkage of the normalizing coefficient toward zero. We then broadcast $\widetilde{\bgamma}\subsup=\widetilde{\bbeta}\subsup/\widetilde{\beta}_{{\scriptscriptstyle \sf sup},1}$ to all sites for local training of the sparse SIM in the next step. Here, the normalization of $\widetilde{\bbeta}\subsup$ with $\widetilde{\beta}_{{\scriptscriptstyle \sf sup},1}$ is made corresponding to the identification condition $\gamma_1=1$.


\setcounter{theorem}{0}

\begin{remark}
{We assume that at least one of the risk factors is {\em known} to be strongly associated with $Y$ and specify it as $X_1$. Without such prior knowledge, it is plausible to modify our proposed algorithm by introducing a screening procedure to select a strong risk factor of $Y$ at the beginning. Existing variable selection methods such as Lasso \citep{tibshirani1996regression} and sure independence screening \citep{fan2008sure} could be applied to the labeled sample $\Lscr\supone$. Alternatively, surrogate-based screening could be applied to the combined $(S,\bX)$ data from the labeled sample and the unlabeled samples $\{\Uscr\supm, m=1,\ldots,M\}$. 
} 
\label{rem:id}
\end{remark}

\paragraph{Step II.} 

Based on Step I, we next model $S\supm_i\sim\bX\supm_i$ in each local site, to obtain a more accurate estimator for $\bgamma$ and use it to derive appropriate summary statistics for integrative analyses. For this purpose, we first define the loss of the SIM (\ref{eqn:sim}). Given any $\bgamma$ and subject $\bX\supm_i$, the unknown link function $f_{m}(\cdot)$ can be naturally estimated through a kernel smoothing estimator:
\begin{equation}
\fhat_m(\bX_i\supm;\bgamma)=\frac{\sum_{j=1,j\neq i}^{N_m} K_{h_m}(\bgamma\trans\bX_j\supm-\bgamma\trans\bX_i \supm)S\supm_j}{\sum_{j=1,j\neq i}^{N_m} K_{h_m}(\bgamma\trans\bX_j\supm-\bgamma\trans\bX_i\supm)}
\label{equ:kern}
\end{equation}
with $K_{h_m}(u)=K(u/h_m)/h_m$, where $K(u)$ represents a kernel function and $h_m>0$ is a bandwidth parameter specific for each $m$. Then the individual-level loss function for $S\supm$ can be formulated as
\begin{equation}
L\supm(\bgamma)=\Pbbhat\supm\widehat\phi_m(\bX;\bgamma)\{S-\fhat_m(\bX;\bgamma)\}^2,
\label{equ:idl}
\end{equation}
where $\widehat\phi_m(\bX;\bgamma)$ is a weighting function introduced to improve estimation efficiency and it may be data-driven. Typically, improper specification of $\phi_m$ has no impact on the convergence rate of the estimator but the finite-sample performance may be affected. In Appendix \ref{sec:app:method}, we provide practical suggestions on $\phi_m$ for binary, count, and continuous surrogates.

The objective (\ref{equ:idl}) is challenging to optimize since the presence of $\fhat_m(\bX;\bgamma)$ makes it non-convex and, thus, prone to the problem of local minima. Our solution is to initiate at the supervised estimator $\widetilde{\bgamma}\subsup$ and approximate (\ref{equ:idl}) with some quadratic function of $\bgamma\supm$. 
For the quadratic approximation with sufficiently small bias, the key procedure is to linearize $\fhat_m(\bX_i\supm;\bgamma)$ with respect to $\bgamma$. Heuristically, with some preliminary $\widetilde{\bgamma}$, e.g., $\widetilde{\bgamma}=\widetilde{\bgamma}\subsup$, we approximate $\widehat{f}_{m}(\bX_i\supm;\bgamma)$ by
\begin{equation}
\widehat{f}_{m}(\bX_i\supm;\bgamma)\approx\fhat_m(\bX_i\supm;\widetilde{\bgamma}) + (\bgamma-\widetilde{\bgamma})\trans\partial_{\bgamma}\fhat_m(\bX_i\supm;\widetilde{\bgamma})+O(\|\bgamma-\widetilde{\bgamma}\|_2^2),
\label{equ:linear:f}
\end{equation}
where $\partial_{\bgamma}\fhat_m(\bX;\bgamma)$ represents the partial derivative function of $\fhat_m(\bX;\bgamma)$ on $\bgamma$ with its explicit form given by equation (\ref{equ:form:pf}) in Appendix. To make the explicit form of $\partial_{\bgamma}\fhat_m(\bX;\bgamma)$ extractable, we further require $K(u)$ to be differentiable with derivative $K'(u)$ and $K'_h(u)=h^{-2}K'(u/h)$, as satisfied, for example, by a suitably normalized triweight kernel $K(u)\propto(1-u^2)^3\mathbf 1(|u|\le1)$. Based on (\ref{equ:linear:f}), we define the one-step approximation of the SIM loss $L\supm(\bgamma)$ around $\widetilde\bgamma$ as
\begin{equation}
Q\supm(\bgamma;\widetilde\bgamma)=\Pbbhat\supm\left\{S-\fhat_m(\bX;\widetilde\bgamma)-(\bgamma-\widetilde\bgamma)\trans\partial_{\bgamma}\fhat_m(\bX;\widetilde\bgamma)\right\}^2.
\label{equ:quad:app}
\end{equation}

As will be justified in Section \ref{sec:theory}, impact of the preliminary $\widetilde{\bgamma}$ on the solution of (\ref{equ:quad:app}) is up to the second order, in a similar spirit to the classic one-step approximation theory \citeg{bickel1975one}. However, when $n$ is much smaller than $N_m$, the second-order impact of $\widetilde{\bgamma}\subsup$ can still exceed the error of the desirable global solution to $L\supm(\bgamma)$ and, thus, dominate the error rate. Thus, instead of one-shot updating, we carry out multiple rounds of updates on $\bgamma$ with $Q\supm(\bgamma;\widetilde\bgamma)$ in our proposed Algorithm \ref{alg:2}.

\begin{algorithm}[htb!]

\caption{\label{alg:2} Sparse SIM regression with multi-round convex approximation.}
{{\bf Input:} The preliminary estimator $\widetilde{\bgamma}\subsup$, samples $\{\Uscr\supm, m =1, ...,M\}$ with surrogates, and the number of iterations $T_m$. }
\vspace{0.15cm}
  {{\bf For} site $m=1,2,\ldots,M$, }initiate with $\widehat\bgamma_{[0]}\supm=\widetilde{\bgamma}\subsup$,
    
    {\hspace{0.3in}{\bf For} $t=1,2,\ldots,T_m$,}
    \begin{equation}\label{eqn:step2}
\widehat\bgamma_{[t]}\supm=\argmin{\bgamma\in\mathbb{R}^p}Q\supm(\bgamma;\widehat\bgamma_{[t-1]}\supm)+\lambda\supm_t \|\bgamma\|_1,\quad{\rm s.t.}\quad\gamma_1=1. 
 \end{equation}
    \hspace{0.2in} Return in site $m$: $\widehat\bgamma\supm=\widehat\bgamma_{[T_m]}\supm$.
\end{algorithm}

Algorithm \ref{alg:2} is an iterative algorithm that starts at $\widehat\bgamma_{[0]}\supm=\widetilde{\bgamma}\subsup$, and updates from $\widehat\bgamma_{[t-1]}\supm$ to $\widehat\bgamma_{[t]}\supm$ with (\ref{eqn:step2}), an $\ell_1$-regularized quadratic objective with loss $Q\supm(\cdot)$. In (\ref{eqn:step2}), we force the first coefficient of $\widehat\bgamma_{[t]}\supm$ to be $1$ corresponding to our identification assumption. The penalty parameter $\lambda\supm_t$ in (\ref{eqn:step2}) is set to gradually decrease with $t$ so that the excessive error of the preliminary estimator can be properly controlled in early iterations, and the regularization error caused by $\lambda\supm_t$ can be controlled at its optimal rate around the final rounds; see Theorem \ref{thm:local}. This indicates an essential difference between our strategy and a standard Newton-type procedure for the $\ell_1$-penalized SIM with a fixed tuning parameter.

{
\begin{remark}
Due to the non-convexity of the original SIM loss in (\ref{equ:idl}), it is necessary to initiate with a sufficiently good $\widetilde{\bgamma}\subsup$ (i.e., being close enough to the true parameters) to ensure the convergence of Algorithm \ref{alg:2}. Intuitively, this can protect our iteratively updated optimizer from falling into some local minima at the very beginning. In Section \ref{sec:theory}, we prove that $\|\bgamma_0-\widetilde{\bgamma}\subsup\|_2=o_p(1)$ is a sufficient condition for convergence. This $\ell_2$-consistency property is justified in Lemma \ref{lem:1} under mild assumptions, in which the labeled sample size $n$ could be substantially smaller than the unlabeled size $N_m$. In addition, we show that the impact of the initial supervised estimator $\widetilde{\bgamma}\subsup$ will decay exponentially fast along the iteration of (\ref{eqn:step2}) and vanish in the end of Algorithm \ref{alg:2}. As a consequence, the number of iterations $T_m$ does not need to be larger than $\log\log N_m$, indicating that Algorithm \ref{alg:2} is not costly in computation; see Theorem \ref{thm:local}.  
\label{rem:imp:gamma}
\end{remark}
}


\begin{remark}
In (\ref{equ:linear:f}) and (\ref{equ:quad:app}), we use $\partial_{\bgamma}\fhat_m(\bX_i\supm;\bgamma)$, the derivative of the kernel estimator $\fhat_m(\bX_i\supm;\bgamma)$ on $\bgamma$ as the gradient to construct the linear expansion of $\widehat{f}_{m}(\bX_i\supm;\bgamma)$. Unlike the GLM, this term is not an estimator of $\bX f'_{m}(\bgamma\trans \bX\supm_i )$ where $f'_{m}$ is the derivative of the true link $f_{m}$. Actually, there is a subtle but essential difference between using $\partial_{\bgamma}\fhat_m(\bX_i\supm;\bgamma)$ and an estimate of $\bX f'_{m}(\bgamma\trans \bX\supm_i )$ in that the former is taken on the kernel estimator $\fhat_m$ while the latter treats $f_m$ as known. Consequently, using the latter would cause excessive bias from the kernel part. The term $\partial_{\bgamma}\fhat_m(\bX_i\supm;\bgamma)$ serves a similar role as the efficient score function of SIM $(\bX-\Ebb[\bX\mid\bgamma\trans \bX\supm_i]) f'_{m}(\bgamma\trans \bX\supm_i )$ used for bias correction with respect to the nuisance kernel estimator in the inference of SIM \citep{liang2010estimation,wu2021model}.
\label{rem:1}
\end{remark}


Based on $\widehat\bgamma\supm$, we derive appropriate summary statistics to approximate the individual data loss function $L\supm(\bgamma)$ in each site $m$. Again, inspired by the bias-corrected approximation used in 
(\ref{equ:linear:f}) and (\ref{equ:quad:app}), we consider the quadratic expansion
\begin{equation}
L\supm(\bgamma) \approx Q\supm(\bgamma;\widehat\bgamma\supm)= C\supm-2\bgamma\trans \bOmegahat_{x,s}\supm+\bgamma\trans \bOmegahat_{x,x}\supm \bgamma,
\label{equ:app:summary}
\end{equation}
where the second-order summary statistics $\bOmegahat\supm_{x,x}\in \R^{p\times p}$ and the first-order summary statistics $\bOmegahat\supm_{x,s}\in \R^{p}$ are formulated and computed as:
\begin{equation}\label{eqn:hess}
\begin{split}
\bOmegahat_{x,x}\supm &=  \Pbbhat\supm \widehat\phi_m(\bX;\widehat\bgamma\supm)\partial_{\bgamma}\fhat_m(\bX;\widehat\bgamma\supm)\Big\{\partial_{\bgamma}\fhat_m(\bX;\widehat\bgamma\supm)\Big\}\trans,\\
    \bOmegahat_{x,s}\supm &=\Pbbhat\supm\widehat\phi_m(\bX;\widehat\bgamma\supm)\Big\{S-\fhat_m(\bX;\widehat\bgamma\supm)+\widehat\bgamma\supmtrans\partial_{\bgamma}\fhat_m(\bX;\widehat\bgamma\supm)\Big\}\partial_{\bgamma}\fhat_m(\bX;\widehat\bgamma\supm),    
\end{split}
\end{equation}
and $C\supm=\Pbbhat\supm\widehat\phi_m(\bX;\widehat\bgamma\supm)\big\{S-\fhat_m(\bX;\widehat\bgamma\supm)+\widehat\bgamma\supmtrans\partial_{\bgamma}\fhat_m(\bX;\widehat\bgamma\supm)\big\}^2$ is a constant not depending on $\bgamma$. Motivated by the above discussion, $\bOmegahat_{x,x}\supm$ and $\bOmegahat_{x,s}\supm$ are sufficient for an effective quadratic approximation to the local individual loss $L\supm(\bgamma)$, while being free of individual-level information ruled out by DataSHIELD \citep{wolfson2010datashield}. These two statistics are then transferred to site 1 for an integrative sparse SIM regression.

\paragraph{Step III.}
Finally, we aggregate summary statistics $\bOmegahat_{x,s}\supm,\bOmegahat_{x,x}\supm$ from local sites and the small labeled samples at site $1$, to derive the final estimator for $\bbeta$. For simplicity, we drop the useless constant $C\supm(\widehat\bgamma\supm)$ in (\ref{equ:app:summary}) and still denote by $Q\supm(\bgamma;\widehat\bgamma\supm)= -2\bgamma\trans \bOmegahat_{x,s}\supm+\bgamma\trans \bOmegahat_{x,x}\supm \bgamma$. We first derive a temporary estimator for the direction coefficients $\bgamma$ through:
\begin{equation}\label{eqn:step3-1}
\widehat{\bgamma}=\argmin{\bgamma\in \R^{p}}\,
\frac{1}{N} \sum_{m=1}^{M}N_m Q\supm(\bgamma;\widehat\bgamma\supm)+ \lambda \|\bgamma\|_1\quad{\rm s.t.}\quad \gamma_1=1,
\end{equation}
where $N=\sum_{m=1}^{M}N_m$. Then, we plug $\widehat{\bgamma}$ in a low-dimensional logistic regression on the labeled data, and estimate the intercept $\alpha$ and the scalar $\beta_{1}$:
\begin{equation}\label{eqn:ahat}
(\widehat \alpha,\widehat \beta_1) = \argmin{\alpha, \beta_1} L_Y(\alpha, \beta_1,\widehat\bgamma),\quad\mbox{where}\quad L_Y(\alpha, \beta_1,\bgamma)=\Pbbhat_n\supone\ell\big\{Y,g(\alpha+\beta_1{\bgamma}\trans \bX)\big\}.
\end{equation}
Finally, we aggregate the summary data and the labeled samples based on $(\widehat \alpha,\widehat \beta_1)$ to derive:
\begin{equation}\label{eqn:step3-3}
\widehat{\bgamma}^{\dagger}={\argmin{\bgamma\in \R^{p}}}
\frac{1}{N^+}\Big[\sum_{m=1}^{M}N_m Q\supm(\bgamma;\widehat\bgamma\supm)+ n L_Y(\widehat\alpha, \widehat\beta_1,\bgamma)\Big]+
\lambda^\dagger\|\bgamma\|_1\quad{\rm s.t.}\quad \gamma_1=1,
\end{equation}
where $N^+=N+n$. {\blue The coefficient $n$ multiplying $L_Y$ is used consistently in the SASH objective, the corresponding BIC, and the IPD objective below.} Note that $\widehat{\bgamma}^{\dagger}$ solved from \eqref{eqn:step3-3} could be slightly more accurate than $\widehat{\bgamma}$ solved from \eqref{eqn:step3-1} due to the use of the small labeled sample. Our final SASH estimator is taken as $ \bbetahat\subsash=\widehat \beta_1\widehat\bgamma^{\dagger}$. 



\subsection{Tuning strategy}\label{sec:tuning}


Theoretically optimal rates for all tuning parameters will be given in Section \ref{sec:theory}. { In this section, we outline our tuning strategy in the numerical implementation that may be slightly different from the theory. Details of our tuning procedures can be found in Appendix \ref{sec:app:num:detail}.} For the penalty parameter $\lambda\subsup$ used in Step I, we select it using cross-validation (CV). For the bandwidth $h_m$, we fix $h_m= \{{s \log (p \vee N_m)}/{N_m} \}^{1/5}$ as suggested by Theorem \ref{thm:local}. Through simulation, we demonstrate that a moderate variation of $h_m$ does not substantially change the overall results in Section \ref{sec:sim:sens}. 
Thus, for simplicity and computation efficiency, we suggest using the fixed value for $h_m$ introduced above. For the unknown sparsity level $s$, since $h_m$ is proportional to $s^{1/5}$ that is typically small in practice, the value of $h_m$ is not sensitive to the choice of this unknown parameter. 

By Theorem \ref{thm:local}, setting the number of iterations in Algorithm \ref{alg:2} as $\log\log N_m$ is sufficient for desirable convergence. Since $\log\log N_m$ is typically small in practice, we fix $T_m=4$ in our numerical studies. Considering the unavailability of individual-level samples from sites $2,3,\ldots,M$ in Step III, we select the tuning parameters $\lambda$ and $\lambda^\dagger$ using the Bayesian information criterion (BIC) inspired by \cite{cai2021individual}. To this end, we additionally calculate $\{\widehat\sigma\supm\}^2=L\supm(\widehat\bgamma\supm)$ at site $m$ and transfer it to site $1$ together with $\bOmegahat_{x,x}\supm$ and $\bOmegahat_{x,s}\supm$. Then in Step III, we propose to choose the $\lambda$ that minimizes:
\begin{equation}
{\blue {\rm BIC}_1(\lambda)=\frac{1}{N}\left[\sum_{m=1}^{M}\frac{N_mQ\supm(\widehat\bgamma(\lambda);\widehat\bgamma\supm)}{\{\widehat{\sigma}\supm\}^2} +\df\{\widehat\bgamma(\lambda)\}\log(N)\right],}
\label{equ:bic}
\end{equation}
where $\widehat\bgamma(\lambda)$ is the solution of \eqref{eqn:step3-1} using $\lambda$ as the penalty parameter, and $\df\{\widehat\bgamma(\lambda)\}$ is the degrees of freedom, i.e., the number of nonzero coefficients in $\widehat\bgamma(\lambda)$. ${\rm BIC}_1(\lambda)$ is a commonly used tuning criterion and it depends on the data only through the summary statistics $\{\widehat\sigma\supm\}^2$, $\bOmegahat_{x,x}\supm$, and $\bOmegahat_{x,s}\supm$. Similarly, for $\lambda^\dagger$ in \eqref{eqn:step3-3} with its corresponding solution $\widehat\bgamma^\dagger(\lambda^\dagger)$, we choose the one minimizing:
\[
{\blue {\rm BIC}_2(\lambda^\dagger)=\frac{1}{N^+}\left[\sum_{m=1}^{M}\frac{N_mQ\supm(\widehat\bgamma^\dagger(\lambda^\dagger);\widehat\bgamma\supm)}{\{\widehat{\sigma}\supm\}^2}+ n L_Y(\widehat\alpha, \widehat\beta_1,\widehat\bgamma^\dagger(\lambda^\dagger))+\df\{\widehat\bgamma^\dagger(\lambda^\dagger)\}\log(N^+)\right].}
\]
Finally, the penalty parameter $\lambda\supm_t$ used in Algorithm \ref{alg:2} is also selected using the BIC similar to (\ref{equ:bic}), with the loss function calculated based on $\widehat\bgamma_{[t-1]}\supm$ derived in the previous iteration; see details in Appendix \ref{sec:app:num:detail}.

\subsection{Interval estimation}\label{sec:method:ci}

{\blue We introduce a method to construct a confidence interval (CI) for $\x\trans\bbeta_0$ with any given $\x$ based on the SASH estimator and the labeled samples at site $1$. Because $(\widehat\alpha,\widehat\beta_1)$ in \eqref{eqn:ahat} is fitted using $\widehat\bgamma$, the low-dimensional likelihood expansion additionally requires
\[
\sqrt n\,\|\widehat\bgamma-\bgamma_0\|_2=o_p(1).
\]
To replace $\bgamma_0$ by $\widehat\bgamma^\dagger$ in the final target, we also require
\[
\sqrt n\,\x\trans(\widehat\bgamma^{\dagger}-\bgamma_0)=o_p(1).
\]
For loadings with bounded $\ell_2$ norm, the rate condition $ns\log p/N\to0$ is sufficient for both requirements by Theorem \ref{thm:aggregate}; $N/n\to\infty$ alone is not sufficient. Under these conditions, $n^{1/2}\widehat\sigma_1^{-1}(\widehat\beta_1-\beta_{01})$ converges weakly to the standard normal distribution, where $\widehat\sigma_1^2$ is the $(2,2)$ entry of
\[
n\left\{(\bone_n,\bX_{1:n}\supone\widehat\bgamma)\trans
\widehat\bW
(\bone_n,\bX_{1:n}\supone\widehat\bgamma)\right\}^{-1},
\]
the second column corresponding to $\beta_1$, $\bX_{1:n}\supone=(\bX_1\supone,\ldots,\bX_n\supone)\trans$, and
\[
\widehat\bW = \operatorname{diag}\left\{\dot g\big(\widehat\alpha+\widehat\beta_1\bX_1\suponetrans\widehat\bgamma\big),\ldots,
\dot g\big(\widehat\alpha+\widehat\beta_1\bX_n\suponetrans\widehat\bgamma\big)\right\}.
\]
Thus, the $100(1-\delta)\%$ CI for $\x\trans\bbeta_0$ is
\[
\left[\widehat\beta_1\x\trans\widehat\bgamma^{\dagger}
-\frac{\widehat\sigma_1|\x\trans\widehat\bgamma^{\dagger}|}{\sqrt n}\Phi^{-1}(1-\delta/2),\quad
\widehat\beta_1\x\trans\widehat\bgamma^{\dagger}
+\frac{\widehat\sigma_1|\x\trans\widehat\bgamma^{\dagger}|}{\sqrt n}\Phi^{-1}(1-\delta/2)\right],
\]
where $\Phi$ is the standard normal distribution function.}

{\blue When either direction-estimation condition above fails, debiasing is needed to obtain a valid interval for $\x\trans\bbeta_0$.} Compared with supervised debiased inference based solely on the labeled sample \citep{cai2021optimal,guo2021inference}, the unlabeled covariates and the more precise direction estimator $\widehat\bgamma^{\dagger}$ may reduce the cost of debiasing. When $\x$ is dense, the supervised debiased estimator of \cite{cai2021optimal} can have asymptotic variance increasing with $\|\x\|_2^2$. If $\widehat\bgamma^{\dagger}$ accurately identifies the nonzero set $S_\beta$ of $\bbeta_0$, the variance may potentially be reduced from order $\|\x\|_2^2$ to $\|\x_{S_\beta}\|_2^2$ when $|S_\beta|\ll p$. A robust implementation of this extension is left for future research.

\section{Theoretical Results}\label{sec:theory}
\subsection{Notations and Assumptions}

Let $\|\bv\|_0$ be the total number of nonzero elements in $\bv$ and 
 $\|\bv\|_{\infty}=\max_{1 \leq i \leq d}|v_{i}|$. For a matrix $A=[A_{ij}]$, 
let $\ell_1$-norm $\|A\|_{1}=\max _{j} \sum_{i}|A_{ij}|$, $\ell_\infty$-norm $\|A\|_{\infty}=\max _{i} \sum_{j}|A_{i j}|$, and $\|A\|_{2}=\lambda_{\max }(A)$
where $\lambda_{\max }(A)$ represents the largest singular value of matrix $A$. We say that a random variable $X$ is sub-Gaussian if 
$\Pbb (|X|>t)\leq Ce^{-ct^{2}}$ for some constants $C, c$, and any $t > 0$, and a random vector $\bX$ is  sub-Gaussian if $\bv\trans \bX$ is sub-Gaussian for any $\bv \in \Rbb^d$ such that $\|\bv\|_2 = 1$. Recall that we assume the labeled sample size $n\lesssim N_m$ for $m\in\{1,2,\ldots,M\}$. 
For simplicity, we set $N=\sum_{m=1}^M N_m$. 
Define the set of nonzero coefficients as $\Bcal = \{j: \bbeta_{0j} \neq 0 \} = \{j: \bgamma_{0j} \neq 0 \}$ and its cardinality, i.e., the sparsity level as  $s = |\Bcal|$. {\blue We focus on the scenario that the weighting function $\widehat\phi_m(\cdot;\cdot)=1$ in the theoretical analysis.} {\blue Let $\widetilde\bX_i\supone=(1,\bX_i\suponetrans)\trans$ and denote the joint Fisher information for the unpenalized intercept and the regression coefficients by
\[
\mathbf{H}_{\rm joint}
=
\mathbb{E}\left[
\dot g(\alpha_0+\bbeta_0\trans\bX_i\supone)
\widetilde\bX_i\supone\widetilde\bX_i\suponetrans
\right].
\]}
To establish the theoretical properties of our SASH estimator, we introduce regularity, sparsity, and smoothness assumptions as below.

{
\begin{assumption}\label{ass:1}
{\blue It holds that
\[
\kappa^{-1}<\lambda_{\min}(\mathbf{H}_{\rm joint})
\leq\lambda_{\max}(\mathbf{H}_{\rm joint})<\kappa,
\qquad
\|\bbeta_0\|_2\leq\kappa,
\qquad
\|\bgamma_0\|_2\leq\kappa_\gamma,
\]
for fixed constants $\kappa,\kappa_\gamma>0$. The joint-information condition accounts for the unpenalized intercept and the unpenalized normalizing coefficient $\beta_1$, and excludes collinearity involving these coordinates.} The sparsity level satisfies $s=|\mathcal B|=o(n\beta_{01}^2/\log p)$ and $|\beta_{01}|=O(1)$.
\end{assumption}
}

\begin{assumption}\label{ass:fm}
{\blue
Let
\[
r_n=C_\Gamma\sqrt{\frac{s\log p}{n\beta_{01}^2}},
\qquad
R_n=4C_\Gamma s\sqrt{\frac{\log p}{n\beta_{01}^2}},
\]
for a sufficiently large fixed constant $C_\Gamma$, and define
\[
\Gamma=
\left\{\bgamma\in\mathbb R^p:
\gamma_1=1,\ 
\|\bgamma-\bgamma_0\|_2\le r_n,\ 
\|\bgamma-\bgamma_0\|_1\le R_n
\right\}.
\]
The following conditions hold for every $m=1,\ldots,M$.
\begin{enumerate}
\item[(a)] The link $f_m$ has bounded first three derivatives on $\mathbb R$, and
\[
\mathbb E\!\left[\{f'_m(\bgamma_0\trans\bX_i\supm)\}^2\right]=a_m^2>0,
\qquad
\max_i|f'_m(\bgamma_0\trans\bX_i\supm)|\le b,
\]
for fixed constants $a_m,b>0$.

\item[(b)] Let $g_{m,\bgamma}$ be the density of $\bgamma\trans\bX_i\supm$. Uniformly over $\bgamma\in\Gamma$, $g_{m,\bgamma}$ is twice differentiable, $g_{m,\bgamma}$ and its first two derivatives are bounded, and $g_{m,\bgamma}$ is bounded away from zero.

\item[(c)] For $r,q\in\{0,1,2\}$, let
\[
\mathcal K_{m,h}^{(r,q)}
=
\left\{
(\bx,\bx')\mapsto
K_h^{(r)}\{\bgamma\trans(\bx'-\bx)\}
(\bx'-\bx)^{\otimes q}:\bgamma\in\Gamma
\right\}.
\]
Let $\mathfrak F_{m,h}$ contain these classes and the corresponding coordinatewise score classes. There exist constants $A,C>0$ and envelopes $F_{\mathcal F}$ such that, uniformly over probability measures $Q$, $\mathcal F\in\mathfrak F_{m,h}$, and $0<\varepsilon\le1$,
\[
\log N\!\left(
\varepsilon\|F_{\mathcal F}\|_{Q,2},
\mathcal F,L_2(Q)
\right)
\le Cs\log\!\left(\frac{Ap}{\varepsilon h}\right),
\]
with the envelope moments required in the concentration bounds.
\end{enumerate}
}
\end{assumption}

\begin{assumption}\label{ass:k}
{\blue The kernel function $K(\cdot)$ is nonnegative, bounded, symmetric around zero, and twice differentiable.} $K(\cdot)$ and its derivatives $K^{\prime}(\cdot), K^{\prime \prime}(\cdot)$ are all Lipschitz on $\mathbb{R}$ and satisfy that
\[
\lim _{|\nu| \rightarrow \infty} K(\nu)=0,\quad\int_{-\infty}^{\infty} K(\nu) d \nu=1,\quad\mbox{and}\quad\int_{-\infty}^{\infty} \nu K^{\prime}(\nu) d \nu=-1.
\]
In addition, for $i=0,1,\ldots,4$, $\int\left|\nu^{i} K(\nu)\right| d \nu<\infty$, $\int\left|\nu^{i} K^{\prime}(\nu)\right| d \nu<\infty$; 
and for $i=0,1,2$, $\int\left|\nu^{i} K''(\nu)\right| d \nu<\infty$, $\int\left|\nu^{i} K'^2(\nu)\right| d \nu<\infty$.

\end{assumption}

\begin{assumption}\label{ass:x}
For each sample $i$ from site $m$, $\bX_i\supm$ and $\epsilon_i\supm=S\supm_i-f_{m}(\bgamma_0\trans \bX\supm_i )$ are both sub-Gaussian.
\end{assumption}

\begin{remark}
Assumption \ref{ass:1} is a standard joint-curvature condition for logistic regression with an unpenalized intercept and an unpenalized normalizing coefficient. Combined with Assumption \ref{ass:x}, it yields the $\ell_1$- and $\ell_2$-rates for the preliminary estimator in Lemma \ref{lem:1}. Assumption \ref{ass:fm}(a) imposes standard smoothness and informativeness conditions on $f_m$, while Assumption \ref{ass:fm}(b)--(c) imposes regularity conditions on the projected-index density and the associated kernel classes. Assumption \ref{ass:k} is satisfied by standard smooth kernels such as the Gaussian kernel. Assumption \ref{ass:x} controls the tail behavior of the covariates and residuals. The present theory assumes a continuously distributed index and sub-Gaussian residuals; the count-surrogate and discrete-genetic-index applications are evaluated empirically and extensions to sub-exponential residuals or discrete indices are left for future work.
\end{remark}


{
\begin{assumption}\label{ass:4}
Define the normalized tangent space
$
\mathbb{V}_0=\big\{\boldsymbol{v}\in\mathbb{R}^p:v_1=0,\ \|\boldsymbol{v}\|_2=1\big\}.
$
There exist positive constants $c$ and $C$ such that the following conditions hold for each sample $i$ from site $m$:
{\blue
\begin{align*}
&c \leq \inf_{\bv \in \mathbb{V}_0} \bv\trans \Ebb\left[\left\{f'_m(\bgamma_0\trans\bX_i\supm)\right\}^2\operatorname{Cov}\big(\bX_i\supm \mid \bgamma_0\trans\bX_i\supm \big)\right] \bv \\
&\hspace{1cm}\leq \sup_{\|\bv\|_2=1}\bv\trans \Ebb\left[\left\{f'_m(\bgamma_0\trans\bX_i\supm)\right\}^2\operatorname{Cov}\big(\bX_i\supm \mid \bgamma_0\trans\bX_i\supm \big)\right] \bv \leq C,\\
&\sup _{\bgamma \in \Gamma} \frac{1}{N_m} \sum_{i=1}^{N_m}\Big[\lambda_{\max }\big(\Ebb(\bX_i\supm \bX_i\supmtrans \mid \bX_i\supmtrans \bgamma)\big)\Big]^2 \leq  C,\\ 
&\max_{1 \leq i \leq N_m} \sup _{\bgamma \in \Gamma} \lambda_{\max }\big(\Ebb(\bX_i\supm \bX_i\supmtrans \mid \bX_i\supmtrans \bgamma)\big) \leq C \log (p \vee N_m).
\end{align*}
}
\end{assumption}
}

\begin{assumption}\label{ass:5}
For each site $m$, every $\bgamma\in\Gamma$, and $t\in\mathbb R$, the conditional mean $\mathbb E(\bX_i\supm\mid\bgamma\trans\bX_i\supm=t)$ is twice differentiable in $t$. Let $\mathbb E^{(1)}$ and $\mathbb E^{(2)}$ denote its first two derivatives. For every $\bv\in\mathbb R^p$,
{\blue
\begin{align*}
\max_i\sup_{\bgamma\in\Gamma}
\left|\mathbb E(\bv\trans\bX_i\supm\mid\bgamma\trans\bX_i\supm)\right|
&\le C\|\bv\|_2,\\
\max_i\sup_{\bgamma\in\Gamma}
\left|\mathbb E^{(1)}(\bv\trans\bX_i\supm\mid\bgamma\trans\bX_i\supm)\right|
&\le C\|\bv\|_2,\\
\sup_{|t|\le cs\sqrt{\log(p\vee N_m)}}\sup_{\bgamma\in\Gamma}
\left|\mathbb E^{(2)}(\bv\trans\bX_i\supm\mid\bgamma\trans\bX_i\supm=t)\right|
&\le C\|\bv\|_2,
\end{align*}
and
\[
\max_i\sup_{\bgamma\in\Gamma}
\left|\mathbb E^{(1)}\{(\bv\trans\bX_i\supm)^2\mid\bgamma\trans\bX_i\supm\}\right|
\le C\|\bv\|_2^2\sqrt{\log(p\vee N_m)}.
\]}
\end{assumption}



\begin{remark}
{{\blue Assumption \ref{ass:4} imposes restricted curvature on the derivative-weighted conditional covariance, which is the relevant population Hessian of the SIM risk.} The restriction $v_1=0$ reflects the normalization $\gamma_1=1$ and removes the intrinsically non-identifiable direction of the single-index model. Assumption \ref{ass:5} imposes smoothness conditions on the conditional moments of $\bX_i\supm$ given $\bgamma\trans\bX_i\supm$, which are used to control the kernel-based approximation uniformly over local neighborhoods of $\bgamma_0$. We provide a justification of Assumption \ref{ass:4} under a Gaussian design in Appendix \ref{sec:ass}; analogous conditions for more general sub-Gaussian designs may be justified using results such as \cite{zhou2009restricted} and \cite{rudelson2012reconstruction}. Similar assumptions have also been used to analyze high-dimensional sparse SIM by \cite{wu2021model}. Unlike that work, we do not impose sparsity assumptions on the inverse Hessian matrix of the SIM because SASH does not rely on local debiased estimators for aggregation.}
\end{remark}

\subsection{Convergence Rates}
In this section, we derive the estimation error rate for the SASH estimator. We first present the estimation error of the local direction estimators obtained via regularized SIM regression in Step II, with its proof provided in Appendix \ref{sec:pf:local}.

\setcounter{theorem}{0}

\renewcommand{\theHtheorem}{convergence.\arabic{theorem}}

{

\begin{lemma}[Consistency of the initial supervised estimator]
\label{lem:1} 
Under the logistic model assumption (\ref{equ:asu:1}) and Assumptions \ref{ass:1} and \ref{ass:x}, there exists $\lambda\subsup\asymp(\log p/n)^{1/2}$ such that, with probability approaching $1$,
{\blue
\[
|\widetilde\beta_{{\scriptscriptstyle\sf sup},1}|\ge \frac12|\beta_{01}|,
\qquad
\|\widetilde{\bgamma}\subsup-\bgamma_0\|_2
\lesssim \sqrt{\frac{s\log p}{n\beta_{01}^2}},
\qquad
\|\widetilde{\bgamma}\subsup-\bgamma_0\|_1
\lesssim s\sqrt{\frac{\log p}{n\beta_{01}^2}}.
\]
Consequently, $\widetilde\bgamma\subsup\in\Gamma$ for a sufficiently large $C_\Gamma$.}
\end{lemma}

Based on the consistency of $\widetilde{\bgamma}\subsup$ shown in Lemma \ref{lem:1}, we can justify the convergence of our iterative Algorithm \ref{alg:2} initialized with $\widetilde{\bgamma}\subsup$ by characterizing and controlling the impact of $\widetilde{\bgamma}\subsup$ along the iterations, and then establish the convergence of the local-site estimators $\widehat\bgamma\supm$ in Theorem \ref{thm:local}.} { It is important to note the complication of justifying Theorem \ref{thm:local} due to the excessive bias from the kernel estimator of $f_m$. As a key step in the proof, we showed that this bias can be reduced leveraging the orthogonality between $\partial_{\bgamma}\fhat_m(\bX\supm_i;\bgamma)$ and $\bX_i\supmtrans\bgamma$, i.e.,  $\psi(\bX_i\supmtrans\bgamma)\partial_{\bgamma}\fhat_m(\bX\supm_i;\bgamma)$ concentrate around zero for arbitrary $\bgamma$ and $\psi(\cdot)$. Remark \ref{rem:1} is intrinsically related to this technical property.}

\setcounter{theorem}{0}

\renewcommand{\theHtheorem}{local.\arabic{theorem}}

\begin{theorem}[Convergence rate for the local estimators] 
{ Under the logistic model assumption (\ref{equ:asu:1}), the SIM assumption (\ref{eqn:sim})}, Assumptions \ref{ass:1} -- \ref{ass:5} and the sparsity assumption $s=O({N_m^{1/6}}/{\log (p \vee N_m)})$ for all $m \in [M]$, there exist $\lambda\subsup\asymp(\log p/n)^{1/2}$, $h_m$ satisfying $
\{{s \log (p \vee N_m)}/{N_m} \}^{1/5}\lesssim h_m \lesssim N_m^{-1/6}$, 
$T_m \asymp \log(-\log h_m) - \log\{\log (n\beta_{01}^2/(s\log p)\}$, $ \lambda_{T_m}\supm \asymp(\log p/{N_m})^{1/2}$, and $
\{\lambda_t\supm:t<T_m\}$ as specified in \eqref{eqn:lambda} of Appendix, 
such that 
$$ 
\|\widehat\bgamma\supm - \bgamma_0\|_2 \lesssim  \sqrt{\frac{s\log p}{N_m}},\quad \|\widehat\bgamma\supm - \bgamma_0\|_1 \lesssim  s\sqrt{\frac{\log p}{N_m}}
$$ 
with probability approaching $1$.
 \label{thm:local}

\end{theorem}

\setcounter{theorem}{5}
\begin{remark}
In terms of the sparsity conditions in Theorem \ref{thm:local}, we impose them regarding both the small labeled sample size $n$ and the large $N_m$. For $n$, we require $s=o(n\beta^2_{01}/\log p)$ which just warrants $\ell_2$-consistency of the GLM Lasso estimator when $\beta_{01}\asymp 1$ and, thus, is regarded to be mild in existing literature \citeg{buhlmann2011statistics,negahban2012unified}. This condition also implies a rate constraint on the strong signal $\beta_{01}$ that $|\beta_{01}|$ needs to be larger than $(s\log p/n)^{1/2}$ in rate. For $N_m$, we assume $s=O({N_m^{1/6}}/{\log (p \vee N_m)})$, which is more restrictive than the GLM setting with the true outcome $Y$ of all the $N_m$ samples. Nevertheless, this assumption is still reasonable and mild in our EHR applications where $N_m$ is typically much larger than $n$ and $p$. 
\end{remark}

\begin{remark}
It is well-known that under Assumption \ref{ass:fm} and a low-dimensional setting, the optimal bandwidth for the kernel estimator of $f_m(\cdot)$ is $N_m^{-1/5}$ for a twice differentiable function \citeg{wand1994kernel}. The rate condition $\{{s \log (p \vee N_m)}/{N_m} \}^{1/5}\lesssim h_m \lesssim N_m^{-1/6}$ in Theorem \ref{thm:local} is inconsistent with this choice due to the presence of high-dimensional sparse $\bgamma$. In addition, the number of iterations $T_m$ is around the order $\log\{\log (N_m)\}$, indicating that Algorithm \ref{alg:2} can converge fast and be computationally efficient in practice. 
\end{remark}

The error rate of $\widehat\bgamma\supm$ provided in Theorem \ref{thm:local} is nearly minimax optimal in the sense that even the sparse GLM estimator extracted under known link functions cannot attain any better convergence rates than this up to a $\log N_m$ scale \citep{raskutti2011minimax}. To the best of our knowledge, except for Gaussian or elliptical design under which one could reduce the SIM fitting to a linear regression \citeg{neykov2016l1,eftekhari2021inference}, no existing sparse SIM approach has a similar optimal convergence guarantee for the non-Gaussian (general) design due to the non-convexity issue. From the methodological perspective, we leverage the supervised estimator and the one-step approximation of SIM introduced in Section \ref{sec:method} to fill this gap. In theory, justifying this convergence property is more challenging than analyzing the standard parametric sparse regression, because first, the impact of the nuisance kernel estimator $\widehat f_m$ needs to be properly reduced through orthogonality; second, an essential effort is needed to analyze our iterative Algorithm \ref{alg:2} with the error rates changing successively.

Based on Theorem \ref{thm:local}, we derive the convergence rate of our SASH estimator in Theorem \ref{thm:aggregate} that is proved in Appendix \ref{sec:pf:aggregate}. 

\setcounter{theorem}{1}
\begin{theorem}[Convergence rate of SASH]\label{thm:aggregate}
Under all assumptions stated in Theorem \ref{thm:local}, there exists $\lambda,\lambda^\dagger \asymp \sqrt{{\log p}/N}$ where $N=\sum_{m=1}^M N_m$, such that the estimators introduced in Step III satisfy $$\|\widehat\bgamma^{\dagger}-\bgamma_0\|_2 \lesssim  \sqrt{\frac{s\log p}{ N}},\quad\|\widehat\bgamma^{\dagger}-\bgamma_0\|_1 \lesssim  s\sqrt{\frac{\log p}{ N}},\quad
|\betahat_1-\beta_{01}| \lesssim 
\sqrt{\frac{1}{n}} + \sqrt{\frac{s\log p}{N}},$$
and, thus,
{\blue $$\|\bbetahat\subsash-\bbeta_0\|_2 \lesssim 
\sqrt{\frac{1}{n}} + \sqrt{\frac{s\log p}{N}},\quad\|\bbetahat\subsash-\bbeta_0\|_1 \lesssim 
\sqrt{\frac{s}{n}} + s\sqrt{\frac{\log p}{N}}.$$}
with probability approaching $1$.

\end{theorem}
{\blue One can see from Theorem \ref{thm:aggregate} that the $\ell_2$ estimation error of SASH has a parametric component $n^{-1/2}$ from estimating $(\alpha,\beta_1)$ and a direction-estimation component $\sqrt{s\log p/N}$. For the $\ell_1$ error, multiplication of the scalar error $|\widehat\beta_1-\beta_{01}|$ by $\|\bgamma_0\|_1\lesssim\sqrt{s}\|\bgamma_0\|_2$ yields the term $\sqrt{s/n}$; the direction component is $s\sqrt{\log p/N}$.} Notably, although Algorithm \ref{alg:2} initiates at the supervised estimator $\widetilde\bgamma
\subsup$ with its $\ell_2$-error being $\{s\log p/(n\beta_{01}^2)\}^{1/2}$, the error rate of our final estimator $\widehat\bgamma^{\dagger}$, as well as that of $\widehat\bgamma\supm$, turns out to be free of the small labeled sample size $n$ and dominated by the large $N$ and $N_m$. This is highly desirable but not readily achieved in the current literature of semi-supervised learning \citeg{zhang2020prior,hou2021efficient} and sparse SIM \citeg{radchenko2015high,wu2021model}. To achieve this result in Theorems \ref{thm:local} and \ref{thm:aggregate}, our main technical novelty can be summarized as (1) developing the one-step approximation of the SIM loss introduced in Section \ref{sec:method:main}; (2) designing and justifying Algorithm \ref{alg:2} with a proper number of iteration $T_m$ and sequence of tuning parameters $\{\lambda\supm_t:t=1,2,\ldots,T_m\}$ to ensure the optimal convergence rate of our SIM estimator. 


\subsection{Comparison with the IPD Estimator}\label{sec:thm:equi}

In this section, we compare SASH with the IPD pooling estimator obtained using the same samples $\Lscr\supone$ and $\{\Uscr\supm:m=1,\ldots,M\}$ but without the DataSHIELD constraint. The complete IPD construction and proof are given in Appendix~\ref{sec:pf:equivalence}. In brief, the local training in Algorithm~\ref{alg:2} is replaced by an integrative sparse SIM regression that pools individual-level samples from all sites, followed by the labeled-data refitting steps corresponding to \eqref{eqn:ahat} and \eqref{eqn:step3-3}. Theorem~\ref{thm:equivalence} establishes equality of the first-order convergence rates and a one-sided no-rate-loss comparison; it does not assert that $\|\bbetahat\subsash-\bbetahat\subipd\|_q$ is of smaller order than the common $\ell_q$ estimation rate.

\begin{theorem}[First-order rate comparison with IPD]\label{thm:equivalence}
Under all assumptions in Theorem~\ref{thm:local}, there exist temporary and final IPD penalty parameters
$\lambda_{{\scriptscriptstyle\sf IPD},0},\lambda_{\subipd}\asymp\sqrt{\log p/N}$ such that the IPD estimator satisfies
\[
\|\bbetahat\subipd-\bbeta_0\|_2
\lesssim
\sqrt{\frac1n}+\sqrt{\frac{s\log p}{N}},
\qquad
\|\bbetahat\subipd-\bbeta_0\|_1
\lesssim
\sqrt{\frac{s}{n}}+s\sqrt{\frac{\log p}{N}}
\]
with probability approaching one. Further, there exist positive sequences
$\lambda_{\delta,0}=o\{\lambda_{{\scriptscriptstyle\sf IPD},0}\}$ and
$\lambda_\delta=o(\lambda_{\subipd})$ such that, if SASH uses
\[
\lambda=\lambda_{{\scriptscriptstyle\sf IPD},0}+\lambda_{\delta,0},
\qquad
\lambda^\dagger=\lambda_{\subipd}+\lambda_\delta
\]
in Step III, then it incurs no first-order rate loss relative to IPD in the following one-sided sense:
\begin{equation}\label{eqn:ipd-comparison-main}
\begin{aligned}
\|\bbetahat\subsash-\bbeta_0\|_2
&\le
\|\bbetahat\subipd-\bbeta_0\|_2
+o_p\!\left(\sqrt{\frac1n}+\sqrt{\frac{s\log p}{N}}\right),\\
\|\bbetahat\subsash-\bbeta_0\|_1
&\le
\|\bbetahat\subipd-\bbeta_0\|_1
+o_p\!\left(\sqrt{\frac{s}{n}}+s\sqrt{\frac{\log p}{N}}\right).
\end{aligned}
\end{equation}
\end{theorem}

Theorem~\ref{thm:equivalence} shows that SASH matches the first-order $\ell_1$ and $\ell_2$ convergence rates of IPD and that its estimation error is no larger than the IPD error plus a smaller-order remainder. This is a rate comparison, not a direct estimator-difference assertion. As in \cite{cai2021individual}, the sparsity condition $s=O\{N_m^{1/6}/\log(p\vee N_m)\}$ makes the higher-order error from the summary-level approximation negligible.

\section{Simulation}\label{sec:sim}

\subsection{Data Generation Setup}

We conduct simulation studies to evaluate our proposed SASH estimator and compare it with existing methods. {\bf R} code for implementation can be downloaded from \url{https://github.com/moleibobliu/SASH}. Throughout, we let the number of sites $M=4$, the labeled sample (in site $1$) size $n=200$, dimension $p=300$, and the size of unlabeled sample $N_m=8000$ for each site $1 \leq m \leq M$. To mimic our real example of EHR data which consists of count variables, we generated $\widetilde\bX\supm_i$ from a multivariate Poisson distribution with mean $5$ and equal correlation being $0.25$, and we take covariates $\bX\supm_i$ as $\log(\widetilde\bX\supm_i + 1 )$. The binary outcome $Y\supm \in \R^{N_{m}}$ is only observed in site $1$ and it follows the logistic model: $\Pbb(Y_i\supm=1 \mid \bX_{i}\supm)=g(\bbeta_0\trans\bX_i\supm)$, with a sparse regression vector $\bbeta_0=(1, -1, 0.5, -0.5,0.25,-0.25,0.125,-0.125,\bzero_{p-8})\trans$. Since the EHR study consists of surrogates of different forms, we generate two categories of surrogate $S\supm$ as follows:
	\begin{equation*}
	\begin{aligned}
	& S_i\supm\mid Y_i\supm=1 \sim {\rm Pois}(\mu_1\supm) \quad {\rm and} \quad S_i\supm\mid Y_i\supm=0 \sim {\rm Pois}(\mu_0\supm) \quad {\rm for} \quad 1 \leq m \leq [M/2], \\\
	& S_i\supm\mid Y_i\supm=1 \sim {\rm Bern}(v_1\supm) \quad {\rm and}  \quad S_i\supm\mid Y_i\supm=0 \sim {\rm Bern}(v_0\supm) \quad {\rm for} \quad [M/2] < m \leq M. 
	\end{aligned}
	\end{equation*}
We consider two hyperparameter settings for the surrogate: (i) weak surrogates with $\mu_1\supm=3$, $\mu_0\supm=1$, 
$v_1\supm=0.75$, and $v_0\supm=0.25$; and (ii) strong surrogates with $\mu_1\supm=5$, $\mu_0\supm=1$, 
$v_1\supm=0.85$, and $v_0\supm=0.15$. {\blue Because Poisson residuals are sub-exponential rather than sub-Gaussian, the Poisson component of this simulation evaluates empirical robustness beyond the scope of Assumption \ref{ass:x}; the Bernoulli component is covered by the bounded-residual case.} Throughout, we summarize results based on $200$ replications in each configuration.

\subsection{Benchmark}\label{sec:sim:bench}

For each simulated data set, we obtain the estimator for $\bbeta_0$ using SASH as well as the following benchmark methods: 
{
\paragraph{IPD} The ideal IPD version of our framework introduced in Section \ref{sec:thm:equi}, with its complete details included in Appendix \ref{sec:pf:equivalence}. Although it is not realistic to pool all individual data together under the privacy constraints, IPD is still an important benchmark method helping evaluate how SASH loses efficiency due to aggregation based on summary statistics, which has been theoretically studied in Section \ref{sec:thm:equi}.

\paragraph{\bf SL} The supervised learning estimator obtained by \eqref{eqn:step1} only using labeled data $\Lscr\supone$.

\paragraph{\bf SS-uLasso} The semi-supervised (SS) version of unsupervised Lasso ({\bf uLasso}) proposed by \cite{chakrabortty2017surrogate}. They developed an unsupervised learning method for the SIM coefficient that applies Lasso to the subset of unlabeled data restricted to some extreme quantile of the surrogate. To adapt their method, we first select subjects with the top $2\%$ (or $5\%$) large or small $S\supm$ at each site $m$, and define $\widetilde Y\supm_i=1$ if $S\supm_i$ is extremely large, and $\widetilde Y\supm_i=0$ if $S\supm_i$ is extremely small. Then we perform distributed sparse logistic regression \citep{cai2021individual} for $\widetilde Y_i\supm\sim \bX\supm_i$ on the selected samples, to obtain an integrative uLasso estimator for $\bgamma_0$. At last, we combine the estimated $\bgamma$ with the labeled samples in the same way as Step III of SASH to derive the final SS-uLasso estimator of $\bbeta_0$.

\paragraph{SS-RMRCE} The SS version of the regularized maximum rank correlation estimator ({\bf RMRCE}) \citep{han2017provable} that maximizes the rank correlation between $S\supm_i$ and $\bX\supmtrans_i\bgamma$. We aggregate the local RMRCE by simple averaging, and, again, use Step III of SASH to obtain the SS version. Due to the prohibitively high computational cost of RMRCE under large $p$, we consider two versions: (i) the ``oracle'' version only including $\{X\supm_{ij}:j=1,\ldots,8\}$ to fit the SIM, as if we knew the true set of active predictors; (ii) the ``screening'' version first selecting top $8$ predictors of $S\supm$ at each local site according to the absolute correlation between each $X\supm_{ij}$ and $S\supm_i$, then combining the $M$ selected sets by majority voting, and finally fitting RMRCE on the selected predictors.

\paragraph{pLasso} The prior Lasso method \citep{jiang2016variable} that leverages some external knowledge of $\bbeta$ by adding to the logistic loss on the labeled data $\Lscr\supone$ a penalty of the discrepancy between the external knowledge and the target model. In our case, this external knowledge is naturally taken as the SS estimator obtained by calibrating some SIM estimator of $\bgamma$ with Step III of SASH. The SIM estimator is obtained using the adaptive Lasso \citep{zou2006adaptive} version of $\ell_1$-regularized least squares regression ({\bf L1LS}) \citep{neykov2016l1} prone to bias under non-Gaussian designs. Thanks to its shrinkage strategy, pLasso achieves some adaptivity to this bias but this strategy is still not optimal in terms of the SIM estimation. Note that one can directly implement L1LS under the privacy constraints since summary data $\Pbbhat\supm\bX\bX\trans$ and $\Pbbhat\supm S\bX$ are sufficient for L1LS. 

\paragraph{PASS} The prior adaptive semi-supervised learning method \citep{zhang2022prior} incorporating the SIM estimator by adaptively shrinking $\bbeta$ to its direction with the $\ell_1$-penalty when fitting logistic regression of $\Lscr\supone$. It uses the adaptive Lasso version of {\bf L1LS} to estimate $\bgamma$.

\vspace{0.25cm}

The RMRCE method is subject to the non-convexity issue of SIM and both uLasso and L1LS encounter bias for non-Gaussian designs. One could note for our descriptions that to ensure fair comparisons, all methods are adapted to the SS setting and utilize unlabeled data from all sources. Additional implementation details (e.g., tuning) are included in Appendix \ref{sec:app:num:detail}. Since the direction parameter $\bgamma$ is an important by-product, we also study SASH's estimation performance of $\bgamma$ and compare it with {\bf IPD}, {\bf L1LS}, {\bf RMRCE}, and {\bf uLasso}. As an ablation study, we include {\bf SASH (w/o Step I)}, the SASH estimator starting by guessing $\bgamma_1 = 1, \bgamma_{-1}=0$ in Algorithm \ref{alg:2} instead of using the supervised $\widetilde\bgamma\subsup$ obtained in Step I. This procedure does not take advantage of the labeled samples to fit the regularized SIM. By Remark \ref{rem:1} and Section \ref{sec:theory}, using $\widetilde\bgamma\subsup$ to initialize the SIM training is essential for the stability and efficiency of SASH. We compare SASH with SASH (w/o Step I) to further demonstrate this point. 
}

\subsection{Main Results}

\subsubsection{SASH shows desirable estimation performance}

In Table \ref{tab:1} and Figure \ref{fig:error:beta}, we present the average and boxplot of the $\ell_2$ estimation error $\|\widehat\bbeta-\bbeta_0\|_2$ of SASH and other competing methods. Both IPD and SASH show better performance under strong surrogacy than weak surrogacy since the former has more informative surrogates to $Y$. Among all the privacy-preserving methods, SASH achieves the closest accuracy to the ideal IPD, with around $20\%$ larger average estimation error than IPD under the weak surrogate setting and $10\%$ larger error under the strong surrogacy. Compared to other semi-supervised methods, SASH attains around $10\%$--$15\%$ smaller error than PASS and { SS-uLasso (two versions using $2\%$ or $5\%$ extreme $S\supm$), $19\%$--$25\%$ smaller error than SS-RMRCE (the oracle and screening versions)}, and $40\%$ smaller error than pLasso under weak surrogacy. Meanwhile, SASH shows even better performance than these existing methods under strong surrogacy. Lastly, since SL does not utilize the unlabeled data with surrogates, it performs the worst among all methods, with around $120\%$ and $200\%$ higher average estimation errors than SASH under the weak and strong surrogate settings, respectively. 


\begin{table}[htb!]
\centering
\begin{small}
\begin{tabular}{lcc} 
\hline
          & Weak surrogate & Strong surrogate  \\ 
\hline
IPD       & 0.735          & 0.600             \\
PASS      & 1.106          & 1.058             \\
pLasso    & 1.560          & 1.640             \\
SASH      & 0.949          & 0.661             \\
SL        & 1.834          & 1.834             \\
SS-RMRCE (screening)  & 1.299          & 1.276             \\
SS-RMRCE (oracle)  &    1.168  &  1.122    \\ 
SS-uLasso (2\%) & 1.096          & 1.060            \\
SS-uLasso (5\%) & 1.038    &   0.841      \\
\hline
\end{tabular}
\end{small}
\caption{\label{tab:1} Average $\ell_2$ estimation errors of $\bbeta$ under the weak surrogate and strong surrogate settings with $N=8000$, $M=4$, $n=200$, and $p=300$. The results are based on 200 simulation replications.}
\end{table}

\begin{figure}[htb!]
	\begin{center}
		\begin{tabular}{ccc}
	{Weak Surrogates} & {Strong Surrogates}\\
		[0pt]
\includegraphics[width=.49\textwidth,angle=0]{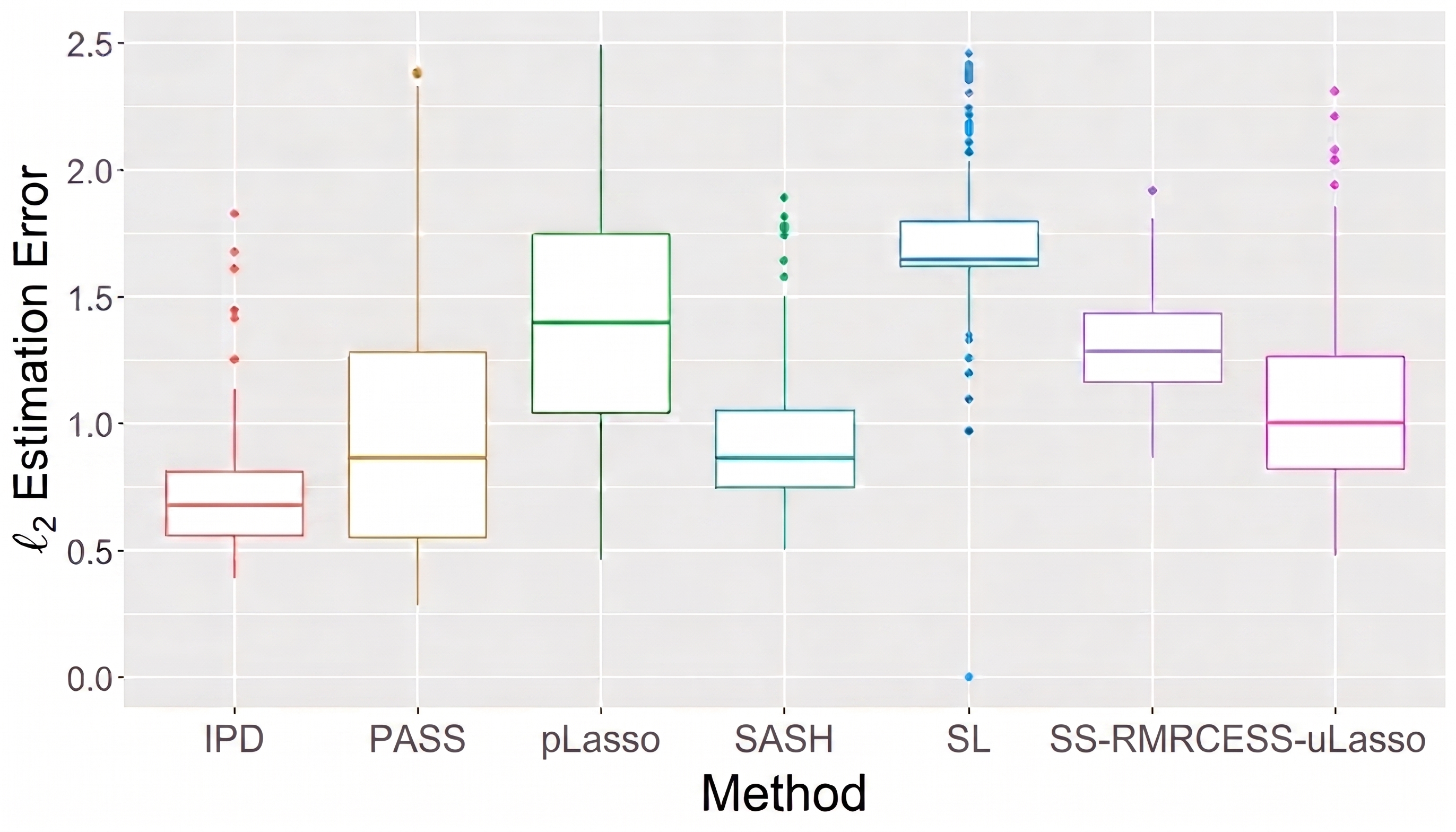} &
	\includegraphics[width=.49\textwidth,angle=0]{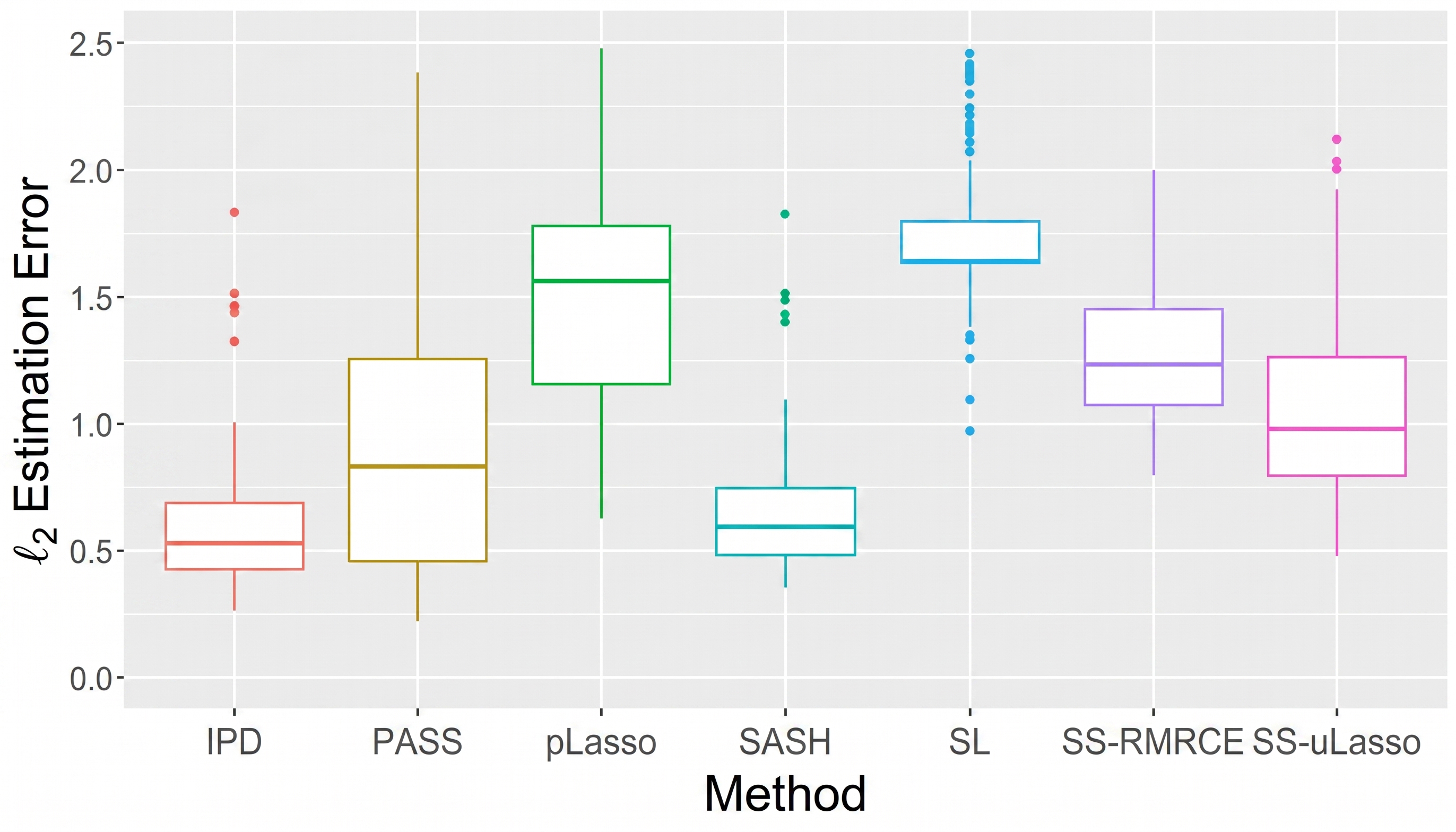} 
		\end{tabular}
	\end{center}
    \vspace{-15pt}
	\caption{Boxplots of the $\ell_2$ estimation errors of $\bbeta$ under the strong surrogate and weak surrogate settings with $N=8000$, $M=4$, $n=200$, and $p=300$. For simplicity, we present the version of SS-uLasso with $2\%$ extreme $S\supm$ and SS-RMRCE with empirically screened predictors, with the other versions of them presented in Table \ref{tab:1}. The methods under comparison are introduced in Section \ref{sec:sim:bench}.  The results are based on 200 simulation replications. SASH achieves the closest accuracy to the ideal IPD, with around $20\%$ larger average estimation error than IPD under the weak surrogate setting and $10\%$ larger error under the strong surrogate settings.} 
	\label{fig:error:beta}
\end{figure}

SASH outperforms existing semi-supervised methods in terms of estimating $\bbeta$ because our approach to learning $\bgamma$ utilizes the information from SIM of $S$ more effectively than existing methods. To further demonstrate this point, we evaluate and compare the estimation performance on $\bgamma$ of the SIM methods introduced in Section \ref{sec:sim:bench}; see Figure \ref{fig:error:gamma}. SASH attains a substantially smaller estimation error than L1LS, uLasso, and RMRCE, as well as a fairly close performance to the ideal IPD estimator. For example, in the strong surrogacy scenario, the estimator of $\bgamma$ by SASH has around $30\%$ smaller median error than uLasso and $40\%$ smaller than RMRCE. This confirms that our way of fitting the sparse SIM is more effective than the existing methods. In addition, compared with SASH utilizing $\widetilde\bgamma\subsup$, SASH w/o Step I has much higher estimation errors on $\bgamma$, i.e., around $55\%$ higher under weak surrogacy and $125\%$ higher under strong surrogacy, implying that the supervised estimator obtained in Step I does play an important role in improving the performance of SASH.

\begin{figure}[htb!]
	\begin{center}
		\begin{tabular}{ccc}
	{Weak Surrogates} & {Strong Surrogates} \\
		[0pt]
\includegraphics[width=.49\textwidth,angle=0]{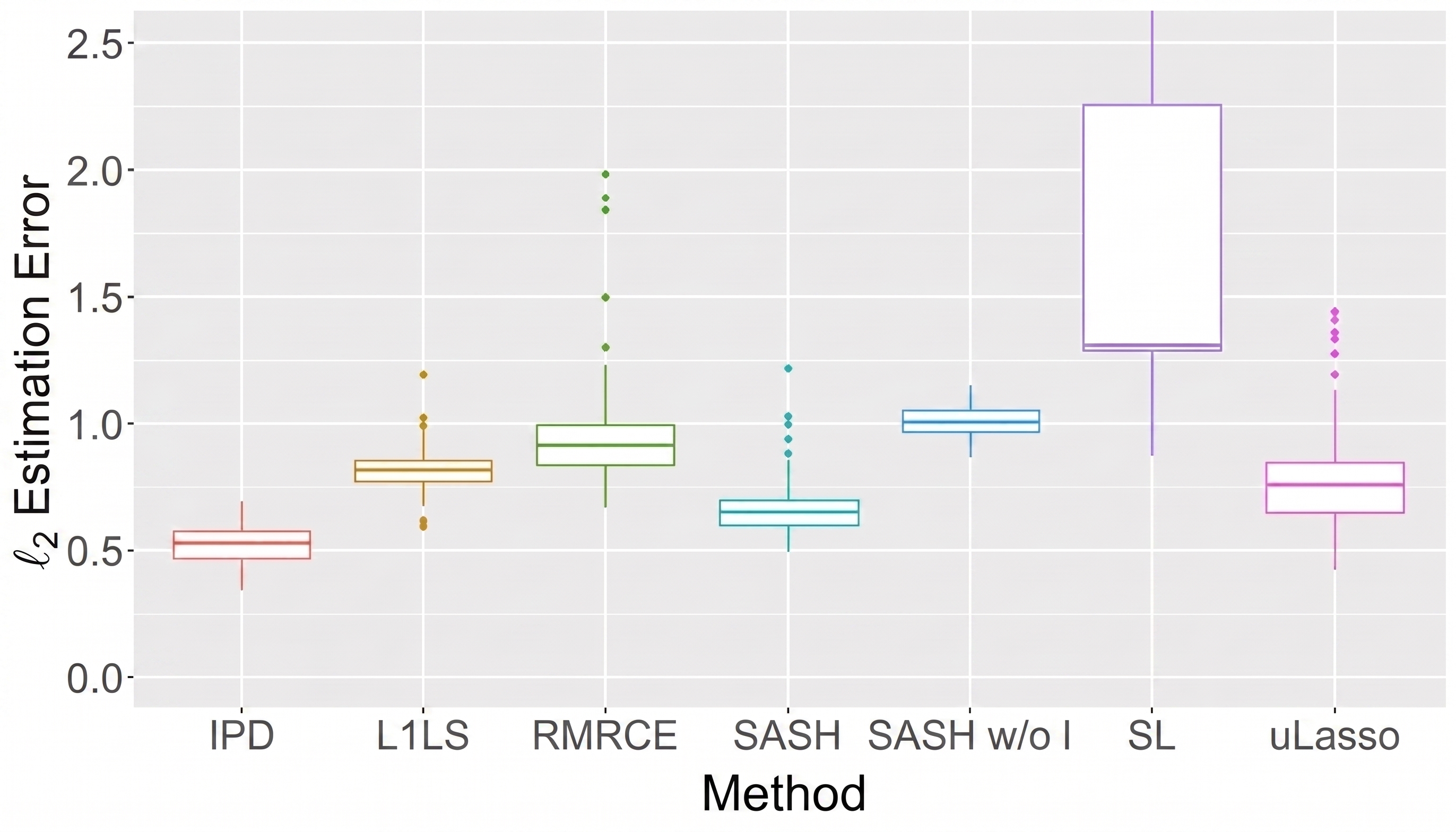} &
	\includegraphics[width=.49\textwidth,angle=0]{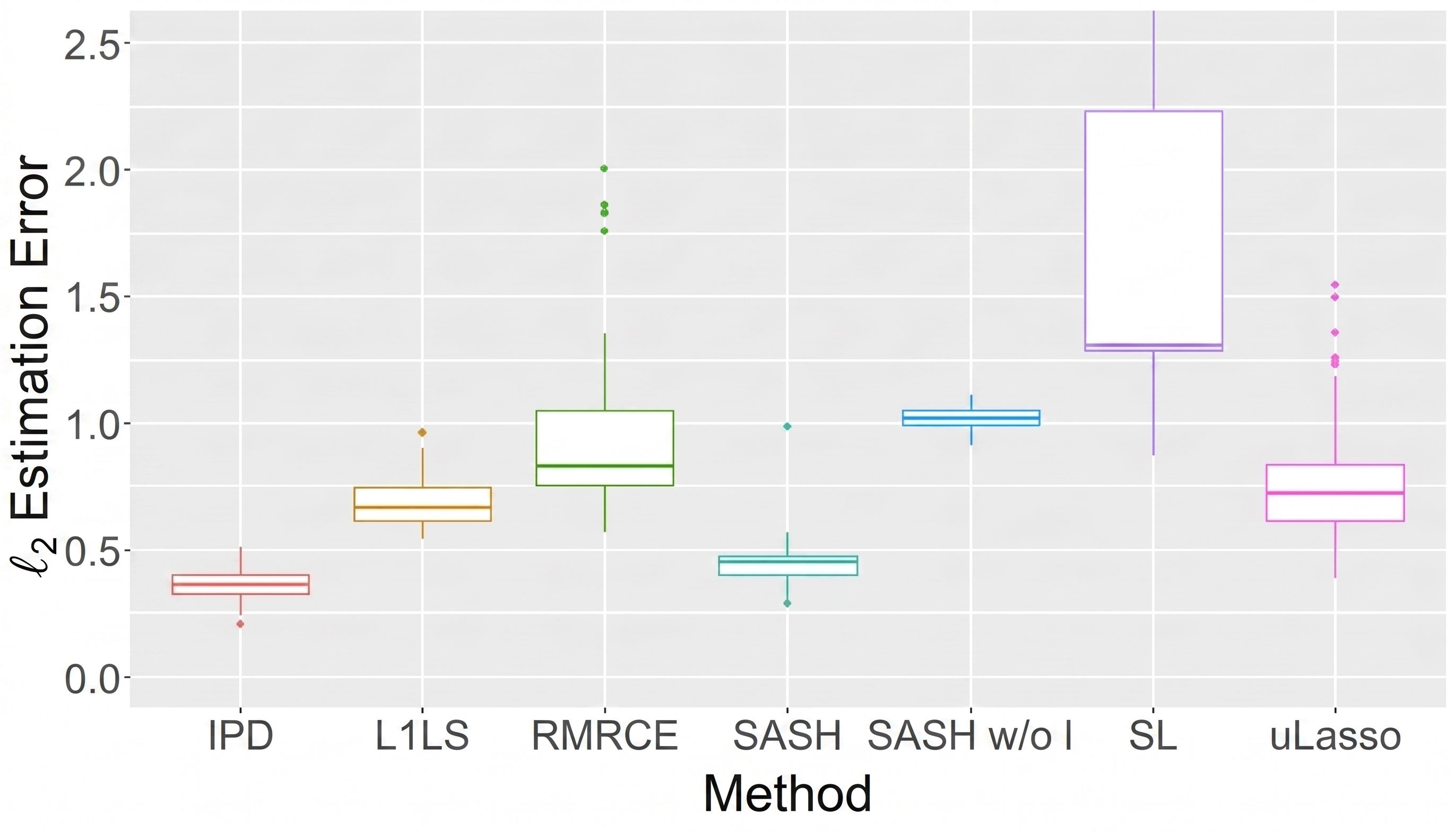} 
		\end{tabular}
	\end{center}
    \vspace{-15pt}
	\caption{Boxplots of the $\ell_2$ estimation errors of $\bgamma$ under the strong surrogate and weak surrogate settings with $N=8000$, $M=4$, $n=200$, and $p=300$. We present the version of uLasso with $2\%$ extreme $S\supm$ and RMRCE with empirically screened predictors. The methods under comparison are introduced in Section \ref{sec:sim:bench}. The results are based on 200 simulation replications. SASH attains  a fairly close performance to the ideal IPD estimator.  In addition, compared with SASH, SASH w/o Step I has much higher estimation errors on $\bgamma$, implying that the supervised estimator obtained in Step I does play an important role in improving the performance of SASH.} 
	\label{fig:error:gamma}
\end{figure}

\subsubsection{Additional communication can improve efficiency}

SASH is moderately outperformed by the ideal IPD estimator under strong surrogacy and the gap between their performance is even more pronounced with a weaker surrogate. This is due to the error caused by approximating the IPD loss with summary statistics and analyzed in Theorems \ref{thm:aggregate} and \ref{thm:equivalence}. To further mitigate this error in finite-sample studies, we investigate a modified SASH estimator, denoted as SASH+, obtained by performing one more round of communication between site $1$ and other sites. Specifically, we transfer the current SASH estimator of $\bgamma$ obtained in site 1 to all sites, use it to derive the summary statistics $\bOmegahat_{x,x}\supm$ and $\bOmegahat_{x,s}\supm$ as introduced in Step II, and transfer the updated statistics back to site 1 for the integrative regression in Step III. This strategy has been studied in existing distributed and federated learning methods \citep{jordan2018communication,duan2022heterogeneity} to reduce the error occurring in data aggregation. 

In Figure \ref{fig:error:sash} of Appendix \ref{sec:app:sashplus}, we present the simulated $\ell_2$ errors on $\bbeta$ and $\bgamma$ of IPD, SASH, and SASH+, under both the strong and weak surrogate settings. Interestingly, we find that benefiting from the one additional round of communication, SASH+ reduces the error gap on $\bbeta$ between SASH and IPD from $20\%$ to $10\%$ under weak surrogacy, and from $10\%$ to $7\%$ under strong surrogacy. These results demonstrate that more rounds of communication are promising to further improve SASH and make it closer to the ideal IPD.

\subsubsection{CI estimation could be further improved}

{

We also present CI coverage probabilities (CPs) for all nine truly nonzero coefficients produced based on IPD, SASH and SL in Table \ref{tab:cover}, constructed using the procedure introduced in Section \ref{sec:method:ci}. SASH attains a close coverage performance to IPD and is significantly better than plugging the SL estimator into the method introduced in Section \ref{sec:method:ci}. Specifically, both SASH and IPD maintain desirable coverage rate (around $95\%$) on more than half of the parameters in both scenarios while SL only shows good performance on two or three of them. This is due to the excessive regularization error of SL. 

However, SASH and IPD attain lower coverage rates than the nominal level for the coefficients indexed by $\{1,3,4\}$. This is because the error of the SIM estimator $\widehat\bgamma$ brings about some non-ignorable bias in finite-sample studies. Our proposed inference approach in Section \ref{sec:method:ci} is easy to implement but overlooks the regularization bias in $\widehat\bgamma$, subject to the condition that the unlabeled sample size $N$ is much larger than $n$. More accurately, for valid asymptotic inference, our method requires the error rate of the sparse SIM, i.e., $\{{s\log p}/{N}\}^{1/2}$ by Theorem \ref{thm:aggregate}, to be much smaller than $n^{-1/2}$. To overcome this issue and enhance our inference procedure in both the asymptotic regime and finite-sample study, one could incorporate the debiasing method of sparse SIMs \citep{wu2021model} to adjust for the regularization bias of $\widehat\bgamma$. This warrants future studies.

}



\begin{table}[htbp]
\centering
\begin{tabular}{llccccccccc}
\toprule
\multirow{2}{*}{Setting} & \multirow{2}{*}{Method} & \multicolumn{9}{c}{$\beta_j$} \\
\cmidrule(lr){3-11}
 & & 1 & 2 & 3 & 4 & 5 & 6 & 7 & 8 & 9\\
\midrule
\multirow{3}{*}{Weak Surrogates} 
 & IPD  & 0.85 & 0.95 & 0.89 & 0.91 & 0.94 & 0.94 & 0.95 & 0.96 & 0.92\\
 & SASH & 0.84 & 0.97 & 0.87 & 0.91 & 0.94 & 0.94 & 0.95 & 0.96 & 0.92\\
 & SL   & 0.66 & 0.84 & 0.79 & 0.84 & 0.89 & 0.94 & 0.92 & 0.97 & 0.91\\
\midrule
\multirow{3}{*}{Strong Surrogates} 
 & IPD  & 0.86 & 0.95 & 0.88 & 0.91 & 0.94 & 0.94 & 0.95 & 0.96 & 0.94\\
 & SASH & 0.86 & 0.97 & 0.87 & 0.93 & 0.96 & 0.95 & 0.95 & 0.96 & 0.93\\
 & SL   & 0.66 & 0.84 & 0.78 & 0.84 & 0.88 & 0.95 & 0.92 & 0.97 & 0.91\\
\bottomrule
\end{tabular}
\caption{Coverage probability of the CIs for $\bbeta$ under the weak surrogate and strong surrogate settings with $N=8000$, $M=4$, $n=200$, and $p=300$. CIs are constructed using the procedure introduced in Section \ref{sec:method:ci}. SASH and IPD maintain desirable coverage rate (around $95\%$) on more than half of the parameters in both scenarios, and are significantly better than plugging the SL estimator into the method introduced in Section \ref{sec:method:ci}.}
\label{tab:cover}
\end{table}

\subsubsection{Performance of SASH is not sensitive to choices of $h$}\label{sec:sim:sens}
We also conduct an additional simulation study to investigate the sensitivity of SASH to different choices of the kernel bandwidth parameter $h$ fixed as $h_0=\{{s \log (p \vee N_m)}/{N_m} \}^{1/5}$ in our main simulation studies. The descriptions and results of this study are presented in Appendix \ref{sec:analysis:h}. Figure \ref{fig:h} demonstrates that different versions of SASH with $h\in\{0.5h_0,h_0,2h_0\}$ have similar estimation performance. For example, the mean estimation error of SASH with $h=2h_0$ is $4.1\%$ smaller than that of $h_0$. Compared to the relative efficiency between SASH and other benchmark methods, this discrepancy caused by the changes of $h$ is substantially smaller, even with $h\in\{0.5h_0,h_0,2h_0\}$ varying in a wide enough range. This means that our method is not sensitive to the specification of $h$.

\section{Real example}\label{sec:realdata}

Due to the increasing linkage of EHR with bio-repositories, large-scale biobank data have emerged as an important and rich resource for risk prediction modeling \citep{kho2011electronic}. For example, both UK Biobank (UKB) and Massachusetts General and Brigham (MGB) Healthcare Biobank contain not only a wealth of medical and clinical records but also their linked genomic data, which enables us to derive genetic risk prediction models useful for personalized medicine \citep{lewis2020polygenic}. However, there exist statistical challenges impeding the effective use of EHR-linked biobank data in this application. Among them, the most prominent obstacle is the lack of accurate or gold-standard labels on disease status in EHR. For example, patients may receive diagnostic codes for type 2 diabetes (T2D) even if they do not actually have the disease. Precise information on T2D status requires human curation via manual chart review, which is not scalable and often infeasible due to the lack of access to clinical notes. Together with high-dimensional genetic features, it is challenging if not infeasible to derive reliable genetic risk prediction models based on labeled data sets. On the other hand, diagnostic codes and other patient-reported outcomes are readily available yet inaccurate. With data from multiple institutions and the availability of gold-standard labels for a small subset of patients, surrogate-assisted and federated approaches like SASH have the potential to address these challenges and produce robust and efficient genetic risk models using biobank data.

To illustrate the potential of these surrogate-assisted federated methods, we develop a T2D genetic risk prediction  model using four sets of biobank data from two healthcare systems, the UK Biobank (UKB) and the Mass General Brigham (MGB). Since the prevalence of T2D differs greatly by age and we only have gold-standard labels on T2D status on older patients, we further partitioned UKB and MGB data into ages $\geq 60$ and $< 60$. This results in $N_1=10{,}685$ and $N_2=7{,}653$ for MGB$_{\geq 60}$ and MGB$_{< 60}$; and $N_3=211{,}061$ and $N_4=275{,}730$ for  UKB$_{\geq 60}$ and UKB$_{< 60}$. Our target model includes gender, ethnicity, and $336$ T2D-associated single nucleotide polymorphisms (SNPs) reported in previous studies \citep{kozak2009ucp1,rodrigues2013genetic} in covariates $\bX$. For the surrogate $S$, we took it as the log-transformed total count of the ICD codes associated with T2D for the MGB subjects, and the binary response to the T2D screening question for those in UKB. The surrogate's distribution and its informativeness to $Y$ could differ greatly between MGB and UKB due to the difference in their definition. Within MGB$_{\geq 60}$, we have an additional $n=277$ patients with gold labels $Y$ on the true T2D status obtained via manual chart review. 

The area under the receiver operating characteristic curve (AUC) of the surrogate $S$ for classifying $Y$ on MGB$_{\geq 60}$ is around $0.85$. This indicates that the surrogate $S$ is error-prone and needs to be properly modeled although $P(Y=0 \mid S=0)$ is almost one, indicating that $S$ is a reasonable negative predictor. Although the sample sizes of the unlabeled data differ greatly among those four data sets, the number of potential cases with $S>0$ does not differ as substantially, with $3{,}081$, $875$, $14{,}992$, and $9{,}393$ in MGB$_{\geq 60}$, MGB$_{< 60}$, UKB$_{\geq 60}$ and UKB$_{< 60}$ respectively. Similar to Section \ref{sec:sim}, we implement SL, SASH, pLasso, PASS, SS-uLasso, and SS-RMRCE to construct the T2D genetic risk model (\ref{equ:asu:1}) while the IPD estimator is not feasible to construct due to the DataSHIELD constraint. For model evaluation, we consider several commonly used accuracy measures including the AUC, the Brier score (BS) defined as the mean squared prediction error of $Y$, and the logistic deviance of the estimator. These accuracy quantities are computed through $5$-fold CV on the labeled samples and presented in Table \ref{tab:EHR1}. As a result of effectively using the unlabeled surrogate samples, SASH performs better than all benchmark methods on all of our evaluation metrics. For example, compared with pLasso, our method attains a $20\%$ larger AUC, a $12\%$ smaller BS, as well as a $13\%$ smaller deviance. 


\begin{table}[htb!]
	\centering	
	\begin{tabular}{ccccccc} 
	\hline
	& PASS & pLasso & SASH  & SL & SS-RMRCE & SS-uLasso \\ 
	\hline
	AUC & 0.533 & 0.563 & {\bf 0.673} & 0.556 & 0.592 & 0.638 \\ 
	BS  & 0.146 & 0.142 & {\bf 0.125} & 0.130 & 0.133 & 0.128 \\ 
	deviance & 0.982 & 0.940 & {\bf 0.822} & 0.852 & 0.870 & 1.054 \\ 
	\hline
	\end{tabular}
	\caption{AUC, Brier score (BS), and deviance of the approaches under comparison in our real-world application evaluated via 5-fold CV.}
	\label{tab:EHR1}
\end{table}

We also present features assigned nonzero estimated coefficients by at least one of the compared methods in Table \ref{tab:EHR2} of Appendix \ref{sec:table}. Interestingly, SASH returns $46$ nonzero coefficients while the other methods select at most $16$ of them. Meanwhile, SASH shows less shrinkage to zero compared to the benchmark methods. Both the larger detection set and the less shrinkage of SASH are also benefits of using the unlabeled surrogate samples to learn the direction of $\bbeta$ more effectively and accurately. 

\section{Conclusion and Discussion}\label{sec:discussion}

{We developed SASH, a novel surrogate-assisted federated learning approach that leverages large unlabeled data with surrogates from multiple local sites to improve the efficiency of statistical learning with labeled data.} Our method is efficient in terms of communication and protects individual-level information by aggregating appropriate summary statistics. It is shown to match the first-order convergence rates of, and to be numerically close to, the IPD estimator, which is ideal but only available without the DataSHIELD constraints. In particular, our newly proposed Algorithm \ref{alg:2} ensures the convergence of the local sparse SIM estimator by using a sufficiently accurate supervised estimator and varying the penalty parameter adaptively during the iteration. In the existing literature on SIM, such a convergence property is not readily achieved for non-Gaussian design due to the non-convexity of the SIM loss function. This novel algorithm, as well as the private data aggregation (\ref{eqn:hess}), is based on our efficient one-step approximation theory of high-dimensional sparse SIM. As we remarked, this development is intrinsically relevant to the efficient score and debiased inference of SIM \citep{liang2010estimation,wu2021model}.

{ More generally speaking, any preliminary estimator consistent with $\bgamma_0$, e.g., prior knowledge of $\bgamma$ derived from some external source data, can be used as a warm start for Algorithm \ref{alg:2} to enable the desirable convergence. With the availability of such (decent) prior knowledge, our strategy of iteratively solving the (debiased) one-step approximation could be generalized to overcome the non-convex optimization issue of other semiparametric regression problems such as maximum rank correlation \citep{han1987non}, and sliced inverse regression \citep{li1991sliced}. Inspired by Remark \ref{rem:1}, this extension requires one to derive the forms of the bias-corrected quadratic approximation for those more complicated regression problems. 

}


By Remark \ref{rem:id}, our method requires prior knowledge of at least one strong enough risk factor with coefficient $\beta_{01}$ to ensure that $\bgamma_0=\bbeta_0/\beta_{01}$ can be well-estimated. Specifically, $|\beta_{01}|$ needs to be larger than $(s\log p/n)^{1/2}$ to guarantee the convergence of Algorithm \ref{alg:2}. Although data-driven feature screening methods \citeg{tibshirani1996regression,fan2008sure} could be used to address the absence of such prior knowledge, they may not guarantee a correct solution without additional regularity and signal strength assumptions. We also note that some previous work uses $\|\bbeta_0\|_2$ instead of $\beta_{01}$ to standardize $\bbeta_0$, i.e., taking $\bgamma_0=\bbeta_0/\|\bbeta_{0}\|_2$. This strategy can work without our assumption on $\beta_{01}$ but its corresponding constraint $\|\bgamma\|_2=1$ will make the optimization of SIM non-convex. Thus, we could modify our identifying condition on $\bgamma$ to remove the assumption on $\beta_{01}$ but we need to be more careful about the quality of SIM training. {More instance-dependent theoretical analysis could refine the convergence rate of SASH and relax some of our technical assumptions in certain cases. For example, we now assume in Assumption \ref{ass:fm} that the unknown link function $f_m$ and the probability density function of $\bgamma\trans\bX\supm_i$ are two or three times differentiable. When those functions have higher-order smoothness properties, the kernel approximation bias will become smaller and our other assumptions like the sparsity $s=O({N_m^{1/6}}/{\log (p \vee N_m)})$ could be potentially refined and relaxed. Also, the approximate sparsity (i.e., $\ell_r$-sparse for some $r\in(0,1)$) regime and the bounded design scenario can be studied to further extend or refine our convergence rate analyses.}

Our method accommodates heterogeneous link functions $f_m$ for $m=1,2,\ldots,M$, which allows the aggregation of different types of surrogates across local sites. This is in a similar spirit to some recent work of angle-based federated or transfer learning \citeg{gu2022robust,tian2023learning}. However, we do not consider the scenario with heterogeneous model coefficients. This problem has been addressed in the federated learning of GLM by \cite{cai2021individual} and others. It would be interesting to explore surrogate-assisted federated learning in this direction. Also, it warrants future research to enhance robustness to the violation of the SIM assumption (\ref{eqn:sim}), or more generally speaking, misleading surrogates. Some recent work of robust SAS \citep{hou2021efficient,zhang2022prior} could be inspiring and useful. In EHR and e-commerce applications, there is often a need to handle multiple or even high-dimensional surrogates to ensure statistical efficiency \citeg{hou2021efficient,cai2022semi}. SASH can be extended to the multi-surrogate problem through the ensemble of surrogates' SIM objective functions. When the surrogates have dependence structures and different strengths, finding the optimal ensemble strategy is important and warrants future studies.

\section*{Acknowledgment}
The authors declare no funding.
{\setlength{\bibsep}{0.0pt}
\bibliography{library}
}
\newpage

\appendix

\setcounter{theorem}{0}
\setcounter{equation}{0}
\setcounter{figure}{0}
\setcounter{table}{0}

\renewcommand{\thetheorem}{\thesection.\arabic{theorem}}
\renewcommand{\theequation}{\thesection.\arabic{equation}}
\renewcommand{\thefigure}{\thesection.\arabic{figure}}
\renewcommand{\thetable}{\thesection.\arabic{table}}

\renewcommand{\theHtheorem}{\thesection.\arabic{theorem}}
\renewcommand{\theHequation}{\thesection.\arabic{equation}}
\renewcommand{\theHfigure}{\thesection.\arabic{figure}}
\renewcommand{\theHtable}{\thesection.\arabic{table}}



\section{Proof of the Main Theorems}\label{sec:A}

\subsection{Proof Outline}\label{sec:outline}
Before proving the main theorems, we first give an outline of the main steps in the proofs.


\begin{enumerate}
\item[S1] We use Lemmas \ref{lem:fhat} and \ref{lem:partial-fhat} to control kernel estimation error.
\item[S2] We use Lemmas \ref{lem:A3} and \ref{lem:A4} bound the additional higher-order error terms introduced by kernel approximation.

{\item[S3] We establish the restricted curvature condition in Lemma \ref{lem:re} on the normalized tangent space $\mathbb{V}_0$ specified in Assumption \ref{ass:4}.}

\item[S4]
To prove Theorem \ref{thm:local}, since $\widehat\bgamma_{[t]}\supm$ minimizes the loss function $Q(\cdot)$ defined in \eqref{eqn:Q1}, we start from the basic inequality  $Q(\widehat\bgamma_{[t]}\supm) \leq Q(\bgamma_0)$, and use the results of S1, S2 and S3.

\item[S5] To prove Theorem \ref{thm:aggregate}, we still utilize the basic inequality
$Q(\widehat\bgamma) \leq Q(\bgamma_0)$ to bound $\widehat\bgamma - \bgamma_0$, where the quadratic loss function $Q(\cdot)$ in aggregation step is defined in \eqref{eqn:Q2}.

\item[S6] To prove the rate comparison in Theorem \ref{thm:equivalence}, following a strategy similar to that used in \cite{cai2021individual}, we start from the basic inequality
$Q(\widehat\bgamma) \leq Q(\widehat\bgamma_{\subipd})$ to compare $\widehat\bgamma$ and $\widehat\bgamma_{\subipd}$, as $\widehat\bgamma$ minimizes quadratic loss function $Q(\cdot)$ in aggregation, and leverage the fact that $\widehat\bgamma_{\subipd}$ minimizes the individual-level loss function to simplify the basic inequality. {Since both $\widehat\bgamma$ and $\widehat\bgamma_{\subipd}$ satisfy the normalization $\gamma_1=1$, their difference has first coordinate zero, so the curvature condition in Assumption \ref{ass:4} applies directly to the comparison direction.}

\end{enumerate}

\newpage

\subsection{Proof of Theorem \ref{thm:local}}\label{sec:pf:local}
\begin{proof}
In Step II, we update the estimator  for
$\bgamma_0$ (direction of $\bbeta_0$) iteratively in Algorithm \ref{alg:2}. Without loss of generality, we consider the $m$-th site ($1 \leq m \leq M$) and the $t$-th iteration. {\blue For the theoretical analysis, we take $\widehat\phi_m(\bX_i\supm;\widehat\bgamma_{[t-1]}\supm)=1$. Lemma \ref{lem:1} gives the initial membership $\widehat\bgamma_{[0]}\supm\in\Gamma$, and the recursion below keeps all iterates in $\Gamma$ on the same high-probability event.}

Recall that we defined 
$\footnotesize Q\supm(\bgamma;\widetilde\bgamma)=\Pbbhat\supm\left\{S-\fhat_m(\bX;\widetilde\bgamma)-(\bgamma-\widetilde\bgamma)\trans\partial_{\bgamma}\fhat_m(\bX;\widetilde\bgamma)\right\}^2.$
Here we use $\small Q\supm_{\lambda\supm_{t}}(\bgamma;\widehat\bgamma_{[t-1]}\supm)$ to denote the loss function in Site $m$ at iteration $t$ with penalty that

\begin{scriptsize}
\begin{equation}
	\begin{aligned}\label{eqn:Q1}
    \small
	& Q\supm_{\lambda\supm_{t}}(\bgamma;\widehat\bgamma_{[t-1]}\supm) \\ =:& N_m^{-1}\sum_{i=1}^{N_m} \left\{S\supm_{i}-\fhat_m(\bX_i\supm;\widehat\bgamma_{[t-1]}\supm)-(\bgamma-\widehat\bgamma_{[t-1]}\supm)\trans\partial_{\bgamma}\fhat_m(\bX_i\supm;\widehat\bgamma_{[t-1]}\supm)\right\}^2+\lambda_t\supm\|\bgamma\|_1 \\
	=& N_m^{-1}\sum_{i=1}^{N_m} \left\{S\supm_{i} - \fhat_m(\bX_i\supm;\bgamma_0) + \fhat_m(\bX_i\supm;\bgamma_0) - \fhat_m(\bX_i\supm;\widehat\bgamma_{[t-1]}\supm) -(\bgamma-\widehat\bgamma_{[t-1]}\supm)\trans\partial_{\bgamma}\fhat_m(\bX_i\supm;\widehat\bgamma_{[t-1]}\supm)\right\}^2
 +\lambda_t\supm\|\bgamma\|_1\\
	 = &N_m^{-1}\sum_{i=1}^{N_m} \bigg\{S\supm_{i} - \fhat_m(\bX_i\supm;\bgamma_0) + (\bgamma_0-\widehat\bgamma_{[t-1]}\supm)\trans \int_{0}^{1} \partial_{\bgamma}\fhat_m \big(\bX_i\supm; \widehat\bgamma_{[t-1]}\supm + \tau (\bgamma_0-\widehat\bgamma_{[t-1]}\supm)\big)  \mathrm{d} \tau  \\ 
	& \ \ \ \  - (\bgamma-\widehat\bgamma_{[t-1]}\supm)\trans\partial_{\bgamma}\fhat_m(\bX_i\supm;\widehat\bgamma_{[t-1]}\supm)\bigg\}^2 +\lambda_t\supm \|\bgamma\|_1 \\
	 = &N_m^{-1}\sum_{i=1}^{N_m} \bigg\{S\supm_{i} - \fhat_m(\bX_i\supm;\bgamma_0) + (\bgamma_0-\widehat\bgamma_{[t-1]}\supm)\trans  \\ 
	& \ \ \cdot \int_{0}^{1} \Big(\partial_{\bgamma}\fhat_m \big(\bX_i\supm; \widehat\bgamma_{[t-1]}\supm + \tau (\bgamma_0-\widehat\bgamma_{[t-1]}\supm)\big) - \partial_{\bgamma}\fhat_m(\bX_i\supm;\widehat\bgamma_{[t-1]}\supm) \Big) \mathrm{d} \tau   - (\bgamma-\bgamma_{0})\trans\partial_{\bgamma}\fhat_m(\bX_i\supm;\widehat\bgamma_{[t-1]}\supm)\bigg\}^2
+ \lambda_t\supm \|\bgamma\|_1,
	\end{aligned}
	\end{equation}
    \end{scriptsize}
where the third equality comes from Taylor expansion that
$$\small \fhat_m(\bX_i\supm;\bgamma_0) - \fhat_m(\bX_i\supm;\widehat\bgamma_{[t-1]}\supm) = (\bgamma_0-\widehat\bgamma_{[t-1]}\supm)\trans \int_{0}^{1} \partial_{\bgamma}\fhat_m \big(\bX_i\supm; \widehat\bgamma_{[t-1]}\supm + \tau (\bgamma_0-\widehat\bgamma_{[t-1]}\supm)\big)  \mathrm{d} \tau. $$
From the optimization procedure in Step II at iteration $t$ that $\widehat\bgamma_{[t]}\supm$ minimize the loss function $Q\supm_{\lambda\supm_{t}}(\cdot;\widehat\bgamma_{[t-1]}\supm)$, we have the basic inequality 
$\small Q\supm_{\lambda\supm_{t}}(\widehat\bgamma_{[t]}\supm;\widehat\bgamma_{[t-1]}\supm) 
 \leq Q\supm_{\lambda\supm_{t}}(\bgamma_0;\widehat\bgamma_{[t-1]}\supm)$, i.e., 
 \begin{scriptsize}
	\begin{align*}
	& N_m^{-1}\sum_{i=1}^{N_m} \bigg\{S\supm_{i} - \fhat_m(\bX_i\supm;\bgamma_0) + (\bgamma_0-\widehat\bgamma_{[t-1]}\supm)\trans  \\ 
	& \ \ \cdot \int_{0}^{1} \Big(\partial_{\bgamma}\fhat_m \big(\bX_i\supm; \widehat\bgamma_{[t-1]}\supm + \tau (\bgamma_0-\widehat\bgamma_{[t-1]}\supm) \big) - \partial_{\bgamma}\fhat_m(\bX_i\supm;\widehat\bgamma_{[t-1]}\supm) \Big) \mathrm{d} \tau   - (\widehat\bgamma_{[t]}\supm-\bgamma_{0})\trans\partial_{\bgamma}\fhat_m(\bX_i\supm;\widehat\bgamma_{[t-1]}\supm)\bigg\}^2 +\lambda_t\supm\|\widehat\bgamma_{[t]}\supm\|_1 \\
	\leq & N_m^{-1}\sum_{i=1}^{N_m} \bigg\{S\supm_{i} - \fhat_m(\bX_i\supm;\bgamma_0) + (\bgamma_0-\widehat\bgamma_{[t-1]}\supm)\trans 
    \cdot \int_{0}^{1} \Big(\partial_{\bgamma}\fhat_m \big(\bX_i\supm; \widehat\bgamma_{[t-1]}\supm + \tau (\bgamma_0-\widehat\bgamma_{[t-1]}\supm)\big) - \partial_{\bgamma}\fhat_m(\bX_i\supm;\widehat\bgamma_{[t-1]}\supm) \Big) \mathrm{d} \tau  \bigg\}^2 +\lambda_t\supm\|\bgamma_{0}\|_1, 
	\end{align*}
    \end{scriptsize}
%
%
%
which further gives

%
%
\begin{scriptsize}
\begin{equation}\label{eqn:A1}
	\begin{aligned}
	& \frac{1}{N_m}\sum_{i=1}^{N_m} \Big\{ (\widehat\bgamma_{[t]}\supm-\bgamma_{0})\trans \partial_{\bgamma}\fhat_m(\bX_i\supm;\widehat\bgamma_{[t-1]}\supm) \Big\}^2 \\
	\leq & \frac{2}{N_m}\sum_{i=1}^{N_m} (\widehat\bgamma_{[t]}\supm-\bgamma_{0})\trans \partial_{\bgamma}\fhat_m(\bX_i\supm;\widehat\bgamma_{[t-1]}\supm)  
	\cdot \bigg\{S\supm_{i} - \fhat_m(\bX_i\supm;\bgamma_0) 
	\\ 
	& \ \  + (\bgamma_0-\widehat\bgamma_{[t-1]}\supm)\trans   
	\int_{0}^{1} \Big(\partial_{\bgamma}\fhat_m(\bX_i\supm; \widehat\bgamma_{[t-1]}\supm + \tau (\bgamma_0-\widehat\bgamma_{[t-1]}\supm)) - \partial_{\bgamma}\fhat_m(\bX_i\supm;\widehat\bgamma_{[t-1]}\supm) \Big) \mathrm{d} \tau  \bigg\} 
 +\lambda_t\supm \big(\|\bgamma_{0}\|_1 - \|\widehat\bgamma_{[t]}\supm\|_1 \big) \\ 
	 = &\frac{2}{N_m}\sum_{i=1}^{N_m} (\widehat\bgamma_{[t]}\supm-\bgamma_{0})\trans \partial_{\bgamma}\fhat_m(\bX_i\supm;\widehat\bgamma_{[t-1]}\supm)  
	\cdot \bigg\{S\supm_{i} -f_m(\bgamma_0\trans\bX_i\supm)+f_m(\bgamma_0\trans\bX_i\supm)- \fhat_m(\bX_i\supm;\bgamma_0) 
	\\ 
	& \ \  + (\bgamma_0-\widehat\bgamma_{[t-1]}\supm)\trans   
	\int_{0}^{1} \Big(\partial_{\bgamma}\fhat_m(\bX_i\supm; \widehat\bgamma_{[t-1]}\supm + \tau (\bgamma_0-\widehat\bgamma_{[t-1]}\supm)) - \partial_{\bgamma}\fhat_m(\bX_i\supm;\widehat\bgamma_{[t-1]}\supm) \Big) \mathrm{d} \tau  \bigg\}
 +\lambda_t\supm \big(\|\bgamma_{0}\|_1 - \|\widehat\bgamma_{[t]}\supm\|_1 \big) 
	\\
	= & 
	\frac{2}{N_m}\sum_{i=1}^{N_m} (\widehat\bgamma_{[t]}\supm-\bgamma_{0})\trans 
	f'_m(\bgamma_0\trans\bX\supm_i) \Big(\bX\supm_i - \Ebb\big(\bX \biggiven \bgamma_0\trans\bX_i\supm \big)\Big) \epsilon\supm_{i}  \\
	& \ \ 
	+ \frac{2}{N_m}\sum_{i=1}^{N_m} (\widehat\bgamma_{[t]}\supm-\bgamma_{0})\trans 
	\Big\{\partial_{\bgamma}\fhat_m(\bX_i\supm;\widehat\bgamma_{[t]}\supm)-f'_m(\bgamma_0\trans\bX\supm_i) \Big(\bX\supm_i - \Ebb\big(\bX \biggiven \bgamma_0\trans\bX_i\supm \big)\Big) \Big\} \epsilon\supm_{i}
	\\
	& \ \ 
	+ \frac{2}{N_m}\sum_{i=1}^{N_m} (\widehat\bgamma_{[t]}\supm-\bgamma_{0})\trans \partial_{\bgamma}\fhat_m(\bX_i\supm;\widehat\bgamma_{[t-1]}\supm)
\cdot \Big\{f_m(\bgamma_0\trans\bX_i\supm)- \fhat_m(\bX_i\supm;\bgamma_0) \Big\}
	\\
	& \ \ 
	+ \frac{2}{N_m}\sum_{i=1}^{N_m} (\widehat\bgamma_{[t]}\supm-\bgamma_{0})\trans \partial_{\bgamma}\fhat_m(\bX_i\supm;\widehat\bgamma_{[t-1]}\supm)   
	\cdot \bigg\{ (\bgamma_0-\widehat\bgamma_{[t-1]}\supm)\trans   
	\int_{0}^{1} \Big(\partial_{\bgamma}\fhat_m(\bX_i\supm; \widehat\bgamma_{[t-1]}\supm + \tau (\bgamma_0-\widehat\bgamma_{[t-1]}\supm) - \partial_{\bgamma}\fhat_m(\bX_i\supm;\widehat\bgamma_{[t-1]}\supm) \Big) \mathrm{d} \tau  \bigg\}
	\\ 
	& \ \ +\lambda_t\supm \big(\|\bgamma_{0}\|_1 - \|\widehat\bgamma_{[t]}\supm\|_1 \big) 
:= A_1 + A_2 + A_3 + A_4 +\lambda_t\supm \big(\|\bgamma_{0}\|_1 - \|\widehat\bgamma_{[t]}\supm\|_1 \big) ,
	\end{aligned}
\end{equation}
\end{scriptsize}
where $\epsilon\supm_{i}:=S\supm_{i}-f_m(\bgamma_0\trans\bX_i\supm)$.
Denote $\Delta_t\supm := \widehat\bgamma_{[t]}\supm-\bgamma_{0}$. By Lemmas \ref{lem:A1} - \ref{lem:A4}, we have 
\begin{small}
\begin{equation}\nonumber
\begin{aligned}
& |A_1| \lesssim  \|\Delta_{t}\supm\|_1\sqrt{\frac{\log p}{N_m}}, \\
& |A_2| \lesssim  \|\Delta_{t}\supm\|_1\sqrt{\frac{\log p}{N_m}}, \\
& |A_3| \lesssim   \|\Delta_{t}\supm\|_1\sqrt{\frac{1}{N_m}} + \|\Delta_{t}\supm\|_1 h_m^3 + \|\Delta_{t}\supm\|_2\|\Delta_{t-1}\supm\|_2 h_m^2,  \\
& |A_4| \lesssim 	\|\Delta_{t}\supm\|_2 
  \|\Delta_{t-1}\supm\|_2 h_m
	+ 
	\|\Delta_{t}\supm\|_2
	 \|\Delta_{t-1}\supm\|_2^2. \\
\end{aligned}    
\end{equation}
\end{small}

Combining the above bounds with Lemma \ref{lem:re}, we have 
\begin{scriptsize}
\begin{equation}\label{eqn:A2}
\begin{aligned}
\|\Delta_t\supm\|_2^2 
& \lesssim   \|\Delta_{t}\supm\|_1 \bigg( \sqrt{\frac{\log p}{N_m}} + \sqrt{\frac{1}{N_m}}  + h_m^3 \bigg) 
+ \|\Delta_{t}\supm\|_{2} 
   \Big(\|\Delta_{t-1}\supm\|_2 h_m^2 + \|\Delta_{t-1}\supm\|_2 h_m
+ \|\Delta_{t-1}\supm\|_2^2  \Big)
+  \lambda_t\supm \big(\|\bgamma_{0}\|_1 - \|\widehat\bgamma_{[t]}\supm\|_1 \big) \\
& 
  \lesssim \sqrt{s}\|\Delta_{t}\supm\|_2  \sqrt{\frac{\log p}{N_m}}  
+ \|\Delta_{t}\supm\|_{2} 
 \Big( \|\Delta_{t-1}\supm\|_2 h_m 
+ \|\Delta_{t-1}\supm\|_2^2  \Big)
+ \lambda_t\supm \sqrt{s} \|\Delta_{t}\supm\|_2 
\end{aligned}
\end{equation}
\end{scriptsize}
 by Lemmas \ref{cor:1} and \ref{cor:2}, which gives 
 \begin{small}
\begin{equation}\label{eqn:A3}
\begin{aligned}
\|\Delta_t\supm\|_2
\lesssim &    \sqrt{\frac{s\log p}{N_m}}   + \|\Delta_{t-1}\supm\|_2 h_m
+ \|\Delta_{t-1}\supm\|_2^2 
+ \lambda_t\supm \sqrt{s},
\end{aligned}  
\end{equation}
\end{small}
under the assumption that $ h_m \lesssim N_m^{-1/6}  $.




{ By Lemma \ref{lem:1}}, the initial estimator satisfies  $\|\Delta_{0}\supm\|_{2} = \|\widehat\bgamma_{[0]}\supm - \bgamma_0 \|_2 \lesssim \sqrt{{ s \log p }/{n\beta^2_{01}}}$.
We choose $
\Big(\frac{s \log (p \vee N_m)}{N_m} \Big)^{1/5}\lesssim h_m \lesssim N_m^{-1/6}$ assumed in Theorem \ref{thm:local}, and choose 
$\lambda_t\supm$ satisfying the assumption of Lemma \ref{lem:cone} that $\lambda_t\supm \asymp \sqrt{\frac{\log p}{N_m}}   + \|\Delta_{t-1}\supm\|_2 h_m 
+ \|\Delta_{t-1}\supm\|_2^2$. 
Then by \eqref{eqn:A3} we have
\begin{small}
\begin{equation}\label{eqn:A-A5}
  \|\Delta_{t}\supm\|_{2}  \lesssim \begin{cases}
   \Big(\sqrt{\frac{ s \log p}{n\beta^2_{01}}}\Big)^{2^{t}}, 
   & \text{if $t < t'$}.\\
    \sqrt{\frac{s\log p}{N_m}}  + \Big(\sqrt{\frac{ s \log p}{n\beta^2_{01}}}\Big)^{2^{t'}} h_m^{t-t'}, 
    & \text{otherwise},
  \end{cases}
\end{equation}
\end{small}
where $\small t' = \log_2 \Big(\log h_m/\log(\sqrt{\frac{ s \log p }{n\beta^2_{01}}}) \Big) + 1$. When $t < t'$, $\|\Delta_{t-1}\supm\|_2^2 $ dominates the error in \eqref{eqn:A3}. And when $t\geq t'$, $\|\Delta_{t-1}\supm\|_2 h_m $ dominates the error in \eqref{eqn:A3} and we have $\small \|\Delta_{t}\supm\|_2 \lesssim \sqrt{\frac{s\log p}{N_m}}  + \Big(\sqrt{\frac{ s \log p}{n\beta^2_{01}}}\Big)^{2^{t'}} h_m^{t-t'} \lesssim \sqrt{\frac{s\log p}{N_m}}  + h_m^{t-t'+1}$ by the definition of $t'$. By the assumption that $\small
\Big(\frac{s \log (p \vee N_m)}{N_m} \Big)^{1/5}\lesssim h_m \lesssim N_m^{-1/6}$, after the iteration $t'$, we only need to iterate constant times to achieve  $\|\Delta_{t}\supm\|_2 \lesssim \sqrt{\frac{s\log p}{N_m}}$.
So we have $T_m \asymp t' \asymp \log(-\log h_m) -\log\{\log (n/(s\log p)\}$
such that
$$\|\Delta_{T_m}\supm\|_{2}  \lesssim \sqrt{\frac{s \log p}{N_m}},$$ 
 and thus $\|\Delta_{T_m}\supm\|_{1}  \lesssim s\sqrt{\frac{ \log p}{N_m}}$
 from Lemma \ref{cor:1}.

Finally to be specific, by \eqref{eqn:A-A5}, we choose 
\begin{small}
\begin{equation}\label{eqn:lambda}
  \lambda_t\supm = \begin{cases}
    \Big(\sqrt{\frac{ s \log p }{n\beta_{01}^2}}\Big)^{2^{t}}, & \text{if $t < t'$}.\\
    \sqrt{\frac{\log p}{N_m}}  + \Big(\sqrt{\frac{ s \log p }{n\beta_{01}^2}}\Big)^{2^{t'}} h_m^{t-t'}, & \text{otherwise}.
  \end{cases}
\end{equation}
\end{small}

\end{proof}

\subsection{Proof of Theorem \ref{thm:aggregate}}\label{sec:pf:aggregate}
\begin{proof}
In Step III, we first aggregate summary statistics from local sites to derive a temporary estimator for the direction $\bgamma_0$. Let the loss function be 
\begin{small}
\begin{equation}\label{eqn:Q2}
 J(\bgamma) = 
\frac{1}{N}\Big\{-2\bgamma\trans\sum_{m=1}^{M}N_m\bOmegahat_{x,s}\supm+\bgamma\trans\sum_{m=1}^{M}N_m \bOmegahat_{x,x}\supm \bgamma \Big\}+\lambda \|\bgamma\|_1,
\end{equation}
\end{small}
where the summary statistics are defined as 	$$ \small
\bOmegahat_{x,x}\supm =  N_m^{-1}\sumiNm \partial_{\bgamma}\fhat_m(\bX_i\supm;\widehat\bgamma\supm)\Big\{\partial_{\bgamma}\fhat_m(\bX_i\supm;\widehat\bgamma\supm)\Big\}\trans,
$$
$$ \small
\bOmegahat_{x,s}\supm  = N_m^{-1}\sumiNm \Big\{S_i\supm-\fhat_m(\bX_i\supm;\widehat\bgamma\supm)+\widehat\bgamma\supmtrans\partial_{\bgamma}\fhat_m(\bX_i\supm;\widehat\bgamma\supm)\Big\}\partial_{\bgamma}\fhat_m(\bX_i\supm;\widehat\bgamma\supm).
$$
$\bOmegahat_{x,s}\supm $ can be decomposed as 
\begin{small}
\begin{equation}\label{eqn:X0}
\begin{aligned}
\bOmegahat_{x,s}\supm  = - \Xihat\supm  + \bOmegahat_{x,x}\supm  \widehat\bgamma\supm, \text{ where } \Xihat\supm  = - N_m^{-1}\sumiNm \Big\{S_i\supm-\fhat_m(\bX_i\supm;\widehat\bgamma\supm)\Big\} \partial_{\bgamma}\fhat_m(\bX_i\supm;\widehat\bgamma\supm).
\end{aligned}
\end{equation}
\end{small}

By the basic inequality
$J(\widehat\bgamma) \leq J(\bgamma_0)$, we have 
\begin{scriptsize}
	\begin{equation}\label{eqn:thm2-1}
	\begin{aligned}
	& \frac{1}{N} \bigg\{\sum_{m=1}^{M} N_{m}(\widehat\bgamma-\bgamma_0)\trans \bOmegahat_{x,x}\supm (\widehat\bgamma-\bgamma_0) \bigg\} + \lambda \|\widehat\bgamma\|_1  
		\leq   \frac{2}{N} \sum_{m=1}^{M} N_{m}(\widehat\bgamma-\bgamma_0)\trans \Big\{-\widehat\Xi\supm + \bOmegahat_{x,x}\supm (\widehat\bgamma\supm-\bgamma_0) \Big\}  + \lambda \|\bgamma_{0}\|_1. 
	\end{aligned}
	\end{equation}
    \end{scriptsize}
%

We first bound $\frac{1}{N} \sum_{m=1}^{M} N_{m}(\widehat\bgamma-\bgamma_0)\trans \Big\{-\widehat\Xi\supm + \bOmegahat_{x,x}\supm (\widehat\bgamma\supm-\bgamma_0) \Big\}$. We have
\begin{footnotesize}
\begin{equation}\nonumber
\begin{aligned}
& \frac{1}{N} \sum_{m=1}^{M} N_{m}(\widehat\bgamma-\bgamma_0)\trans \Big\{-\widehat\Xi\supm + \bOmegahat_{x,x}\supm (\widehat\bgamma\supm-\bgamma_0) \Big\}  \\
= & \frac{1}{N} \sum_{m=1}^{M} (\widehat\bgamma-\bgamma_0)\trans \sum_{i=1}^{N_m} \partial_{\bgamma}\fhat_m(\bX_i\supm;\widehat\bgamma\supm)  \bigg\{S\supm_{i} - \fhat_m(\bX_i\supm;\widehat\bgamma\supm) +   \partial_{\bgamma}\fhat_m(\bX_i\supm;\widehat\bgamma\supm)\trans (\widehat\bgamma\supm-\bgamma_0) \bigg\}  \\
= & \frac{1}{N} \sum_{m=1}^{M} (\widehat\bgamma-\bgamma_0)\trans \sum_{i=1}^{N_m} \partial_{\bgamma}\fhat_m(\bX_i\supm;\widehat\bgamma\supm)  \bigg\{S\supm_{i} 
- f_m(\bgamma_0\trans\bX_i\supm) 
+ f_m(\bgamma_0\trans\bX_i\supm) - \fhat_m(\bX_i\supm;\bgamma_0) \\
& \qquad\qquad\qquad\qquad + \fhat_m(\bX_i\supm;\bgamma_0) 
- \fhat_m(\bX_i\supm;\widehat\bgamma\supm) 
+ \partial_{\bgamma}\fhat_m(\bX_i\supm;\widehat\bgamma\supm)\trans (\widehat\bgamma\supm-\bgamma_0) \bigg\} \\
= & \frac{1}{N} \sum_{m=1}^{M} (\widehat\bgamma-\bgamma_0)\trans \sum_{i=1}^{N_m} \partial_{\bgamma}\fhat_m(\bX_i\supm;\widehat\bgamma\supm)  \bigg\{\epsilon_i\supm 
+ f_m(\bgamma_0\trans\bX_i\supm) - \fhat_m(\bX_i\supm;\bgamma_0) \\
& \qquad\qquad\qquad\qquad + (\bgamma_0-\widehat\bgamma\supm)\trans   
\int_{0}^{1} \Big(\partial_{\bgamma}\fhat_m(\bX_i\supm; \widehat\bgamma\supm + \tau (\bgamma_0-\widehat\bgamma\supm)) - \partial_{\bgamma}\fhat_m(\bX_i\supm;\widehat\bgamma\supm) \Big) \mathrm{d} \tau \bigg\}. 
\end{aligned}   
\end{equation}
\end{footnotesize}
By similar argument in the proof of Theorem \ref{thm:local} and the assumption $h_m \lesssim  N_m^{-1/6}$, we have

\begin{footnotesize}
\begin{equation}\nonumber
\begin{aligned}
& \Big|\frac{1}{N} \sum_{m=1}^{M}\sum_{i=1}^{N_m}(\widehat\bgamma-\bgamma_0)\trans  \partial_{\bgamma}\fhat_m(\bX_i\supm;\widehat\bgamma\supm)  \epsilon_i\supm \Big| 
\lesssim \|\widehat\bgamma-\bgamma_0\|_1\sqrt{\frac{\log p}{N}}, \\
& \Big| \frac{1}{N} \sum_{m=1}^{M}\sum_{i=1}^{N_m}(\widehat\bgamma-\bgamma_0)\trans  \partial_{\bgamma}\fhat_m(\bX_i\supm;\widehat\bgamma\supm) \Big\{ f_m(\bgamma_0\trans\bX_i\supm) - \fhat_m(\bX_i\supm;\bgamma_0) \Big\} \Big| \\
 \lesssim &  \|\widehat\bgamma-\bgamma_0\|_1\sqrt{\frac{1}{N}} 
+ \|\widehat\bgamma-\bgamma_0\|_2 \frac{1}{N} \sum_{m=1}^{M}N_m h_m^3 
+ \|\widehat\bgamma-\bgamma_0\|_2 \frac{1}{N} \sum_{m=1}^{M}N_m\|\widehat\bgamma\supm-\bgamma_0\|_2 h_m^2
 \lesssim  \|\widehat\bgamma-\bgamma_0\|_1\sqrt{\frac{1}{N}},  \\
\end{aligned}   
\end{equation} 
\end{footnotesize}
and
\begin{footnotesize}
\begin{equation}\nonumber
\begin{aligned}
& \bigg| \frac{1}{N_m}\sum_{i=1}^{N_m}(\widehat\bgamma-\bgamma_0)\trans  \partial_{\bgamma}\fhat_m(\bX_i\supm;\widehat\bgamma\supm) \bigg\{ (\bgamma_0-\widehat\bgamma\supm)\trans   
\int_{0}^{1} \Big(\partial_{\bgamma}\fhat_m(\bX_i\supm; \bgamma_0 + \tau (\bgamma_0-\widehat\bgamma\supm)) - \partial_{\bgamma}\fhat_m(\bX_i\supm;\widehat\bgamma\supm) \Big) \mathrm{d} \tau \bigg\} \bigg| \\
 \lesssim &  \|\widehat\bgamma-\bgamma_0\|_2 
  \frac{1}{N} \sum_{m=1}^{M}N_m\|\widehat\bgamma\supm-\bgamma_0\|_2 h_m
+ \|\widehat\bgamma-\bgamma_0\|_2 \frac{1}{N} \sum_{m=1}^{M}N_m\|\widehat\bgamma\supm-\bgamma_0\|_2^2 \\
\lesssim &  \|\widehat\bgamma-\bgamma_0\|_2 
  \frac{1}{N} \sum_{m=1}^{M}N_m \sqrt{\frac{s\log p}{N_m}} h_m
+ \|\widehat\bgamma-\bgamma_0\|_2 \frac{1}{N} \sum_{m=1}^{M}N_m \frac{s\log p}{N_m} \\
\lesssim &  \|\widehat\bgamma-\bgamma_0\|_2\sqrt{\frac{s\log p}{N}},
\end{aligned}   
\end{equation} 
\end{footnotesize}
where the first bound follows from Lemmas \ref{lem:A1} and \ref{lem:A2}, the second bound follows from Lemma \ref{lem:A3}, and the third bound follows from Lemma \ref{lem:A4}. 
{Let
\[
\mathcal{R}_N=\frac{1}{N}\sum_{m=1}^M N_m\Big\{-\widehat\Xi\supm+\bOmegahat_{x,x}\supm(\widehat\bgamma\supm-\bgamma_0)\Big\}.
\]
The same termwise bounds, before taking the inner product with $\widehat\bgamma-\bgamma_0$, imply that the leading stochastic part of $\mathcal{R}_N$ is $O_p\{\sqrt{\log p/N}\}$ in the $\ell_\infty$ norm, while the remaining higher-order terms are $o_p(\lambda)$ under the bandwidth and sparsity conditions of Theorem \ref{thm:local}. Taking the constant in $\lambda\asymp\sqrt{\log p/N}$ sufficiently large and applying the basic inequality \eqref{eqn:thm2-1} yield
\[
\|(\widehat\bgamma-\bgamma_0)_{S^c}\|_1\leq 3\|(\widehat\bgamma-\bgamma_0)_S\|_1.
\]
Because $\widehat\gamma_1=\gamma_{01}=1$, the aggregate error also has first coordinate zero. Consequently,
\[
\|\widehat\bgamma-\bgamma_0\|_1\leq 4\sqrt{s}\|\widehat\bgamma-\bgamma_0\|_2.
\]
}
Thus we have
\begin{small}
\begin{equation}\label{eqn:A16}
 \bigg| \frac{1}{N} \sum_{m=1}^{M} N_{m}(\widehat\bgamma-\bgamma_0)\trans \Big\{-\widehat\Xi\supm + \bOmegahat_{x,x}\supm (\widehat\bgamma\supm-\bgamma_0) \Big\} \bigg| 
\lesssim   \sqrt{s}
\|\widehat\bgamma-\bgamma_0\|_2\sqrt{\frac{\log p}{N}},
\end{equation}
\end{small}
which follows from similar argument in Lemma \ref{cor:1}.

{Since $(\widehat\bgamma-\bgamma_0)_1=0$, the normalized aggregate error belongs to $\mathbb{V}_0$ in Assumption \ref{ass:4}. Therefore, by an argument identical to that in Lemma \ref{lem:re}, we have}
\begin{equation}\label{eqn:A17}
\small
\frac{1}{N} \bigg\{\sum_{m=1}^{M} N_{m}(\widehat\bgamma-\bgamma_0)\trans \bOmegahat_{x,x}\supm (\widehat\bgamma-\bgamma_0) \bigg\} \geq \kappa \|\widehat\bgamma-\bgamma_0\|_2^2.
\end{equation}
Combining \eqref{eqn:thm2-1}, \eqref{eqn:A16} and \eqref{eqn:A17}, we have
\begin{equation}\nonumber
\small
\|\widehat\bgamma-\bgamma_0\|_2^2 
\lesssim  \sqrt{s}
\|\widehat\bgamma-\bgamma_0\|_2 \sqrt{\frac{\log p}{N}} + \lambda \sqrt{s} \|\widehat\bgamma-\bgamma_0\|_2
\end{equation} since $\|\bgamma_{0}\|_1-\|\widehat\bgamma\|_1 \lesssim \sqrt{s} \|\widehat\bgamma-\bgamma_0\|_2$ by similar argument in Lemma \ref{cor:2}.
So we finally have
$$\small \|\widehat\bgamma-\bgamma_0\|_2 
\lesssim 
\sqrt{\frac{s\log p}{N}}  $$
by taking $\lambda \asymp  \sqrt{{\log p}/{N}} $.


Then plug $\bgammahat$ in a low-dimensional logistic regression on labeled
data to estimate intercept $\alpha$ and scalar $\beta_1$. Denoting $\small \btheta=(\alpha,\beta_1)$ and $\small \ell_i(\btheta,\bgamma) = \ell\{Y_i\supone,g(\alpha+\beta_1{\bgamma}\trans \bX_i\supone)\}$,  we have
\begin{scriptsize}
\begin{equation}\label{eqn:3-2}
\begin{aligned}
&\frac{1}{n}
\sum_{i=1}^n (\bthetahat-\btheta_0)\trans
\nabla^2_{\theta}\ell_i(\btheta_0,\bgammahat)
(\bthetahat-\btheta_0) \\
\leq &-\frac{2}{n}
\sum_{i=1}^n(\bthetahat-\btheta_0)\trans
\nabla_{\btheta}\ell_i(\btheta_0,\bgammahat) \\
&{\blue -\frac{2}{n}\sum_{i=1}^n\int_0^1(1-t)
(\bthetahat-\btheta_0)\trans
\Big\{\nabla^2_{\theta}\ell_i\big(\btheta_0+t(\bthetahat-\btheta_0),\bgammahat\big)
-\nabla^2_{\theta}\ell_i(\btheta_0,\bgammahat)\Big\}
(\bthetahat-\btheta_0)\,\mathrm{d}t}.
\end{aligned}
\end{equation}
\end{scriptsize}
from basic inequality that $\frac{1}{n}  \sum_{i=1}^{n}\ell_i(\bthetahat, \widehat{\bgamma}) \leq \frac{1}{n}  \sum_{i=1}^{n}\ell_i(\btheta_0, \widehat{\bgamma})$ and the decomposition that
	\begin{footnotesize}
\begin{align*}
\ell_i(\bthetahat,\bgammahat)
&=\ell_i(\btheta_0,\bgammahat)
 +(\bthetahat-\btheta_0)\trans\nabla_{\btheta}\ell_i(\btheta_0,\bgammahat) \\
&\quad {\blue +\int_0^1(1-t)(\bthetahat-\btheta_0)\trans
\nabla^2_{\btheta}\ell_i\big(\btheta_0+t(\bthetahat-\btheta_0),\bgammahat\big)
(\bthetahat-\btheta_0)\,\mathrm{d}t} \\
&=\ell_i(\btheta_0,\bgammahat)
 +(\bthetahat-\btheta_0)\trans\nabla_{\btheta}\ell_i(\btheta_0,\bgammahat)
 +\frac12(\bthetahat-\btheta_0)\trans
\nabla^2_{\btheta}\ell_i(\btheta_0,\bgammahat)
(\bthetahat-\btheta_0) \\
&\quad {\blue +\int_0^1(1-t)(\bthetahat-\btheta_0)\trans
\Big\{\nabla^2_{\btheta}\ell_i\big(\btheta_0+t(\bthetahat-\btheta_0),\bgammahat\big)
-\nabla^2_{\btheta}\ell_i(\btheta_0,\bgammahat)\Big\}
(\bthetahat-\btheta_0)\,\mathrm{d}t}.
\end{align*}
\end{footnotesize}
where 
\begin{scriptsize}
\begin{equation}\label{eqn:ell}
\begin{aligned}
& \ell_i(\btheta, \bgamma) = \log(1 + \exp(\alpha+\beta_1\bgamma\trans \bX_i\supone)) - Y_i\supone (\alpha+\beta_1\bgamma\trans \bX_i\supone), \\
&   \nabla_{\btheta} \ell_i(\btheta,\bgamma) = \Big(-Y_i + \frac{\exp(\alpha+\beta_1\bgamma\trans \bX_i\supone)}{1+\exp(\alpha+\beta_1\bgamma\trans \bX_i\supone)} \Big)
\binom{1 }{\bgamma\trans\bX_i\supone}, \\
& \nabla^2_{\btheta} \ell_i(\btheta,\bgamma) = \frac{\exp(\alpha+\beta_1\bgamma\trans \bX_i\supone)}{\big(1+\exp(\alpha+\beta_1\bgamma\trans \bX_i\supone)\big)^2}
\binom{1 }{\bgamma\trans\bX_i\supone}
(1, \bgamma\trans\bX_i\supone), \\
&\nabla^2_{\btheta \bgamma} \ell_i(\btheta,\bgamma) = \frac{\exp(\alpha+\beta_1\bgamma\trans \bX_i\supone)}{\big(1+\exp(\alpha+\beta_1\bgamma\trans \bX_i\supone)\big)^2} \binom{1 }{\bgamma\trans\bX_i\supone}
\beta_1 {\bX_i\supone}\trans 
+ \Big(-Y_i + \frac{\exp(\alpha+\beta_1\bgamma\trans \bX_i\supone)}{1+\exp(\alpha+\beta_1\bgamma\trans \bX_i\supone)} \Big)
\binom{\boldsymbol{0}}{{\bX_i\supone}\trans}.
\end{aligned}    
\end{equation}
\end{scriptsize}

First, by Lemma J.3 in \cite{ning2017general}, we have RE condition holds for some constant $\kappa$ that 
\begin{small}
\begin{equation}\label{eqn:A19}
 \frac{1}{n} 
	\sum_{i=1}^n  (\bthetahat-\btheta_0)\trans \nabla^2_{\theta}\ell_i(\btheta_0,\bgammahat)(\bthetahat-\btheta_0) \geq \kappa   \|\bthetahat-\btheta_0\|_2^2.
\end{equation}  
\end{small}
Second, we bound $\frac{1}{n}
\sum_{i=1}^n (\bthetahat-\btheta_0)\trans \nabla_{\btheta}\ell_i(\btheta_0,\bgammahat)$. We have
\begin{small}
\begin{equation}\label{eqn:A20-1}
\Big| \frac{1}{n}
\sum_{i=1}^n (\bthetahat-\btheta_0)\trans \nabla_{\btheta}\ell_i(\btheta_0,\bgamma_0) \Big| 
\leq \|\bthetahat-\btheta_0\|_1  \Big\|\frac{1}{n}
\sum_{i=1}^n \nabla_{\btheta}\ell_i(\btheta_0,\bgamma_0)\Big\|_{\infty} 
\lesssim \sqrt{\frac{1}{n}}\|\bthetahat-\btheta_0\|_1 
\end{equation}
\end{small}
holds by Lemma E.2 in \cite{ning2017general},
and
\begin{small}
\begin{equation}\label{eqn:A20-2}
\begin{aligned}
& \Big| (\bthetahat-\btheta_0)\trans \frac{1}{n} \sum_{i=1}^n \big(\nabla_{\btheta}\ell_i(\btheta_0,\bgammahat) - \nabla_{\btheta}\ell_i(\btheta_0,\bgamma_0)\big)\Big| \\
= & \Big| (\bthetahat-\btheta_0)\trans \frac{1}{n}
\sum_{i=1}^n \int_0^1  \nabla^2_{\btheta \bgamma}\ell_i(\btheta_0,\bgamma_0 + t(\bgammahat-\bgamma_0) )\d t (\bgammahat - \bgamma_0) \Big| \\
\lesssim &
\|\bthetahat-\btheta_0\|_2
\|\bgammahat - \bgamma_0\|_2  
\end{aligned}    
\end{equation}
\end{small}
follows from the expression of $\nabla^2_{\btheta \bgamma}\ell_i(\cdot)$ in \eqref{eqn:ell}.
So
$$\small \Big| \frac{1}{n}
\sum_{i=1}^n (\bthetahat-\btheta_0)\trans \nabla_{\btheta}\ell_i(\btheta_0,\bgammahat) \Big| \lesssim 
\sqrt{\frac{1}{n}}\|\bthetahat-\btheta_0\|_1 
+ 
\|\bthetahat-\btheta_0\|_2
\|\bgammahat - \bgamma_0\|_2
$$

Third, we have
\begin{scriptsize}
\begin{equation}\label{eqn:A21}
\begin{aligned}
& \bigg| \frac{1}{n} \sum_{i=1}^n  \int_0^1 {\blue (1-t)}(\bthetahat-\btheta_0)\trans \Big\{\nabla^2_{\btheta}\ell_i\big(\btheta_0+t(\bthetahat-\btheta_0),\bgammahat \big) - \nabla^2_{\theta} \ell_i(\btheta_0,\bgammahat) \Big\}(\bthetahat-\btheta_0) \mathrm{d} t \bigg| 
\\
= & \bigg| \frac{1}{n} \sum_{i=1}^n  \int_0^1 {\blue (1-t)}(\bthetahat-\btheta_0)\trans
\bigg\{ 
\frac{\exp(a_t+b_t\bgammahat\trans \bX_i\supone)}{\big(1+\exp(a_t+b_t\bgammahat\trans \bX_i\supone)\big)^2}
-
\frac{\exp(a+b\bgammahat\trans \bX_i\supone)}{\big(1+\exp(a_0+b_0\bgammahat\trans \bX_i\supone)\big)^2}
\bigg\} \binom{1 }{\bgammahat\trans\bX_i\supone}
(1, \bgammahat\trans\bX_i\supone)
(\bthetahat-\btheta_0) \mathrm{d} t \bigg|
 \\
\lesssim &
\|\bthetahat-\btheta_0\|^3_2,
\end{aligned}
\end{equation}
\end{scriptsize}
where $(a_t,b_t)\trans:=\btheta_0+t(\bthetahat-\btheta_0)$. Combining \eqref{eqn:3-2}, \eqref{eqn:A19}, \eqref{eqn:A20-1}, \eqref{eqn:A20-2} and \eqref{eqn:A21}, we have
\begin{footnotesize}
\begin{equation}\nonumber
\begin{aligned}
\|\bthetahat-\btheta_0\|_2^2 
\lesssim & \|\bthetahat-\btheta_0\|_1 \sqrt{\frac{1}{n}} + \|\bthetahat-\btheta_0\|_2 \|\bgammahat - \bgamma_0\|_2 + \|\bthetahat-\btheta_0\|^3_2 
\lesssim  \|\bthetahat-\btheta_0\|_2 \Big(\sqrt{\frac{1}{n}} +  \sqrt{\frac{s\log p}{N}} \Big) + \|\bthetahat-\btheta_0\|^3_2,
\end{aligned}
\end{equation} 
\end{footnotesize}
which immediately gives
\begin{small}
\begin{equation}\label{eqn:a}
\|\bthetahat-\btheta_0\|_2 \lesssim  \sqrt{\frac{1}{n}} +  \sqrt{\frac{s\log p}{N}}.
\end{equation}
\end{small}

Finally, we aggregate the summary statistics from local sites and the labeled sample loss based on $(\widehat \alpha,\widehat \beta_1)$ to derive the final estimator for $\bgamma_0$. We have that the loss function $$\small \mathcal{L}(\bgamma)=\frac{1}{N+n}\Big\{-2\bgamma\trans\sum_{m=1}^{M}N_m\bOmegahat\supm_{x,s}+\bgamma\trans\sum_{m=1}^{M}N_m\bOmegahat\supm_{x,x}\bgamma+ \sum_{i=1}^{n}\ell_i(\alphahat, \betahat_1,\bgamma)\Big\}+\lambda^\dagger \|\bgamma\|_1$$
and the basic inequality $\mathcal{L}(\widehat\bgamma^{\dagger}) \leq \mathcal{L}(\bgamma_0)$  gives
	\begin{footnotesize}
\begin{equation}\nonumber
\begin{aligned}
&\frac{1}{N+n}\bigg\{
\sum_{m=1}^{M}N_m(\widehat\bgamma^{\dagger}-\bgamma_0)\trans
\bOmegahat_{x,x}\supm(\widehat\bgamma^{\dagger}-\bgamma_0) \\
&\qquad\qquad
+{\blue \frac12\sum_{i=1}^n(\widehat\bgamma^{\dagger}-\bgamma_0)\trans
\nabla^2_{\bgamma}\ell_i(\bthetahat,\bgamma_0)
(\widehat\bgamma^{\dagger}-\bgamma_0)}\bigg\}
+\lambda^\dagger\|\widehat\bgamma^{\dagger}\|_1 \\
\leq &\frac{1}{N+n}\bigg\{
2\sum_{m=1}^{M}N_m(\widehat\bgamma^{\dagger}-\bgamma_0)\trans
\Big\{-\widehat\Xi\supm+\bOmegahat_{x,x}\supm(\widehat\bgamma\supm-\bgamma_0)\Big\} \\
&\qquad
-\sum_{i=1}^n(\widehat\bgamma^{\dagger}-\bgamma_0)\trans
\nabla_{\bgamma}\ell_i(\bthetahat,\bgamma_0) \\
&\qquad {\blue
-\sum_{i=1}^n\int_0^1(1-t)(\widehat\bgamma^{\dagger}-\bgamma_0)\trans
\Big\{\nabla^2_{\bgamma}\ell_i\big(\bthetahat,\bgamma_0+t(\widehat\bgamma^{\dagger}-\bgamma_0)\big)
-\nabla^2_{\bgamma}\ell_i(\bthetahat,\bgamma_0)\Big\}
(\widehat\bgamma^{\dagger}-\bgamma_0)\,\mathrm{d}t}\bigg\}
+\lambda^\dagger\|\bgamma_0\|_1.
\end{aligned}
\end{equation}
\end{footnotesize}
{The same score and penalty argument as above gives
\[
\|(\widehat\bgamma^{\dagger}-\bgamma_0)_{S^c}\|_1\leq 3\|(\widehat\bgamma^{\dagger}-\bgamma_0)_S\|_1.
\]
Moreover, $(\widehat\bgamma^{\dagger}-\bgamma_0)_1=0$, so the normalized error belongs to $\mathbb{V}_0$ in Assumption \ref{ass:4}, and
$\|\widehat\bgamma^{\dagger}-\bgamma_0\|_1\leq4\sqrt{s}\|\widehat\bgamma^{\dagger}-\bgamma_0\|_2$.}
which gives 
\begin{small}
\begin{equation}\nonumber
\begin{aligned}
&  \|\widehat\bgamma^{\dagger}-\bgamma_0\|_2^2  
\lesssim  \frac{\sqrt{s}}{N+n} \|\widehat\bgamma^{\dagger}-\bgamma_0\|_2 \Big\{ \sqrt{N \log p} + \sqrt{n \log p} \Big\} + \lambda^\dagger \sqrt{s} \|\widehat{\bgamma}^{\dagger}-\bgamma_{0}\|_2 
\end{aligned}
\end{equation} 
\end{small}
by an argument similar to the one used above.
So we have
\begin{small}
\begin{equation}\label{eqn:r}
\|\widehat\bgamma^{\dagger}-\bgamma_0\|_2 
\lesssim \sqrt{s} \frac{\sqrt{N \log p} + \sqrt{n \log p} }{N+n} + \lambda^\dagger\sqrt{s} 
\lesssim \sqrt{\frac{s\log p}{N+n}} 
\lesssim \sqrt{\frac{s\log p}{N}} 
\end{equation}
\end{small}
by taking $\lambda^{\dagger} \asymp \sqrt{\frac{\log p}{N+n}} \asymp \sqrt{\frac{\log p}{N}} $.
Combining \eqref{eqn:a} and \eqref{eqn:r}, finally we have 
$$\small \|\bbetahat\subsash-\bbeta_0\|_2 \lesssim 
\sqrt{\frac{1}{n}} + \sqrt{\frac{s\log p}{N}}.$$
{\blue Moreover,
\[
\begin{aligned}
\|\bbetahat\subsash-\bbeta_0\|_1
&\leq |\widehat\beta_1-\beta_{01}|\,\|\bgamma_0\|_1
 +|\widehat\beta_1|\,\|\widehat\bgamma^\dagger-\bgamma_0\|_1\\
&\lesssim \sqrt{\frac{s}{n}}+s\sqrt{\frac{\log p}{N}},
\end{aligned}
\]
where $\|\bgamma_0\|_1\leq\sqrt{s}\|\bgamma_0\|_2$ and $\|\bgamma_0\|_2$ is bounded by Assumption \ref{ass:1}.}
Both bounds hold with probability at least $1-O(1/p)$.

\end{proof}

\subsection{Proof of Theorem \ref{thm:equivalence}}\label{sec:pf:equivalence}

We first give the complete IPD construction. Throughout this proof, the theoretical weighting function is one, as specified in Section~\ref{sec:theory}.

\paragraph{Step I.}
The first step is identical to SASH: fit the supervised penalized logistic regression on the labeled sample and set
\(
\widehat\bgamma_{{\rm IPD},[0]}
=
\widetilde\bbeta\subsup/\widetilde\beta_{{\scriptscriptstyle\sf sup},1}.
\)

\paragraph{Step II.}
For \(t=1,\ldots,T\), define
\begin{equation}\label{eqn:IPD-II}
\begin{aligned}
\widehat\bgamma_{{\rm IPD},[t]}
=
\argmin{\bgamma\in\mathbb R^p:\,\gamma_1=1}
\bigg[&
\sum_{m=1}^M\frac{N_m}{N}\Pbbhat\supm
\Big\{S-\fhat_m(\bX;\widehat\bgamma_{{\rm IPD},[t-1]})
-(\bgamma-\widehat\bgamma_{{\rm IPD},[t-1]})\trans
\partial_{\bgamma}\fhat_m(\bX;\widehat\bgamma_{{\rm IPD},[t-1]})\Big\}^2
\\&
+\lambda_{{\rm IPD},[t]}\|\bgamma\|_1
\bigg].
\end{aligned}
\end{equation}
Let \(\widehat\bgamma_{\subipd}=\widehat\bgamma_{{\rm IPD},[T]}\) and write
\(\lambda_{{\scriptscriptstyle\sf IPD},0}=\lambda_{{\rm IPD},[T]}\).

\paragraph{Step III.}
First obtain
\begin{equation}\label{eqn:IPD-III-1}
(\widehat\alpha_{\subipd},\widehat\beta_{\subipd,1})
=
\argmin{\alpha,\beta_1}
\Pbbhat_n\supone
\ell\{Y,g(\alpha+\beta_1\widehat\bgamma_{\subipd}\trans\bX)\}.
\end{equation}
Then define
\begin{equation}\label{eqn:IPD-III-2}
\begin{aligned}
\widehat\bgamma_{\subipd}^{\dagger}
=
\argmin{\bgamma\in\mathbb R^p:\,\gamma_1=1}
\bigg\{\frac{1}{N+n}\bigg[&
\sum_{m=1}^M N_m\Pbbhat\supm
\Big\{S-\fhat_m(\bX;\widehat\bgamma_{\subipd})
-(\bgamma-\widehat\bgamma_{\subipd})\trans
\partial_{\bgamma}\fhat_m(\bX;\widehat\bgamma_{\subipd})\Big\}^2
\\&
+n\Pbbhat_n\supone
\ell\{Y,g(\widehat\alpha_{\subipd}
+\widehat\beta_{\subipd,1}\bgamma\trans\bX)\}
\bigg]
+\lambda_{\subipd}\|\bgamma\|_1\bigg\}.
\end{aligned}
\end{equation}
The IPD estimator is
\(
\widehat\bbeta_{\subipd}
=
\widehat\beta_{\subipd,1}\widehat\bgamma_{\subipd}^{\dagger}.
\)

\begin{proof}
Put \(\lambda_N=(\log p/N)^{1/2}\),
\(r_{2,N}=\sqrt{s}\lambda_N\), and \(r_{1,N}=s\lambda_N\).
All statements below are made on the intersection of the high-probability events in Lemmas~\ref{lem:A1}--\ref{lem:re}; its probability tends to one.

\paragraph{Rates of the pooled IPD estimator.}
The proof of Theorem~\ref{thm:local} applies to the pooled update~\eqref{eqn:IPD-II}, with the weighted empirical measure
\(
\sum_m(N_m/N)\Pbbhat\supm
\)
and total sample size \(N\). The same initialization, score, cone, and derivative-weighted curvature arguments therefore give, for the tuning sequence specified in Theorem~\ref{thm:local},
\begin{equation}\label{eqn:ipd-rate-temp}
\|\widehat\bgamma_{\subipd}-\bgamma_0\|_2
=O_p(r_{2,N}),
\qquad
\|\widehat\bgamma_{\subipd}-\bgamma_0\|_1
=O_p(r_{1,N}).
\end{equation}
Applying the low-dimensional refitting argument in the proof of Theorem~\ref{thm:aggregate} to~\eqref{eqn:IPD-III-1} yields
\begin{equation}\label{eqn:ipd-rate-theta}
\|\widehat\btheta_{\subipd}-\btheta_0\|_2
=O_p(n^{-1/2}+r_{2,N}),
\qquad
\widehat\btheta_{\subipd}
=(\widehat\alpha_{\subipd},\widehat\beta_{\subipd,1})\trans.
\end{equation}
Finally, the proof for \(\widehat\bgamma^\dagger\) in Theorem~\ref{thm:aggregate}, with the summary quadratic loss replaced by its pooled individual-level counterpart, gives
\begin{equation}\label{eqn:ipd-rate-final}
\|\widehat\bgamma_{\subipd}^{\dagger}-\bgamma_0\|_2
=O_p(r_{2,N}),
\qquad
\|\widehat\bgamma_{\subipd}^{\dagger}-\bgamma_0\|_1
=O_p(r_{1,N}).
\end{equation}
Since \(\|\bgamma_0\|_1\le\sqrt{s}\|\bgamma_0\|_2\),
\eqref{eqn:ipd-rate-theta}--\eqref{eqn:ipd-rate-final} imply
\begin{equation}\label{eqn:ipd-beta-rates}
\|\widehat\bbeta_{\subipd}-\bbeta_0\|_2
=O_p(n^{-1/2}+r_{2,N}),
\qquad
\|\widehat\bbeta_{\subipd}-\bbeta_0\|_1
=O_p(\sqrt{s/n}+r_{1,N}).
\end{equation}
This proves the first assertion of the theorem.

\paragraph{A summary-to-IPD approximation bound.}
We record the comparison bound used below. Let
\(
\Delta_0=\widehat\bgamma-\widehat\bgamma_{\subipd}
\)
and
\(
\Delta_\dagger=\widehat\bgamma^\dagger-
\widehat\bgamma_{\subipd}^\dagger.
\)
Both have first coordinate zero. Taylor expansion of \(\fhat_m\), with the exact integral remainder
\begin{equation}\label{eqn:exact-taylor-ipd}
\fhat_m(\bX;u)-\fhat_m(\bX;v)
=
\partial_{\bgamma}\fhat_m(\bX;v)\trans(u-v)
+
\int_0^1(1-t)(u-v)\trans
\partial_{\bgamma}^2\fhat_m\{\bX;v+t(u-v)\}
(u-v)\,dt,
\end{equation}
combined with Lemmas~\ref{lem:A2}--\ref{lem:A4},
\ref{lem:fhat}, and~\ref{lem:partial-fhat}, gives
\begin{align}
|\mathcal R_{0,N}|
&=o_p(s\lambda_N^2),
\label{eqn:R0}\\
|\mathcal R_{\dagger,N}|
&=o_p(s\lambda_N^2),
\label{eqn:Rdagger}
\end{align}
where
\begin{align*}
\mathcal R_{0,N}
={}&
\frac1N\sum_{m=1}^MN_m\Delta_0\trans
\Big[
-\widehat\Xi\supm
+\widehat\Omega_{x,x}\supm
(\widehat\bgamma\supm-\widehat\bgamma_{\subipd})
+\widehat\Xi^{\subipd}
-\widehat\Omega_{x,x}^{\subipd}
(\widehat\bgamma_{{\rm IPD},[T-1]}-\widehat\bgamma_{\subipd})
\Big],
\\
\mathcal R_{\dagger,N}
={}&
\frac1{N+n}\sum_{m=1}^MN_m\Delta_\dagger\trans
\Big[
-\widehat\Xi\supm
+\widehat\Omega_{x,x}\supm
(\widehat\bgamma\supm-\widehat\bgamma_{\subipd}^{\dagger})
+\widehat\Xi_{\subipd}^{\dagger}
-\widehat H_{\subipd}^{\dagger}
(\widehat\bgamma_{\subipd}-\widehat\bgamma_{\subipd}^{\dagger})
\Big]
\\&
+
\frac1{N+n}\sum_{i=1}^n\Delta_\dagger\trans
\Big[
\nabla_{\bgamma}\ell_i(\widehat\btheta,\widehat\bgamma_{\subipd}^{\dagger})
-
\nabla_{\bgamma}\ell_i(\widehat\btheta_{\subipd},\widehat\bgamma_{\subipd}^{\dagger})
\Big]
\\&
+
\frac1{N+n}\sum_{i=1}^n
\int_0^1(1-t)\Delta_\dagger\trans
\Big[
\nabla_{\bgamma}^2\ell_i\{\widehat\btheta,
\widehat\bgamma_{\subipd}^{\dagger}+t\Delta_\dagger\}
-
\nabla_{\bgamma}^2\ell_i(\widehat\btheta,
\widehat\bgamma_{\subipd}^{\dagger})
\Big]
\Delta_\dagger\,dt.
\end{align*}
Here \(\widehat\Xi^{\subipd}\) and
\(\widehat\Omega_{x,x}^{\subipd}\) are the pooled score and quadratic matrix in the last update of~\eqref{eqn:IPD-II}, while
\(\widehat\Xi_{\subipd}^{\dagger}\) and
\(\widehat H_{\subipd}^{\dagger}\) are their analogues in~\eqref{eqn:IPD-III-2}.
For completeness, each term generated by~\eqref{eqn:exact-taylor-ipd} is bounded by a product of one comparison direction and at least one of
\[
h_m,
\quad
\|\widehat\bgamma\supm-\widehat\bgamma_{{\rm IPD},[T-1]}\|_2,
\quad
\|\widehat\bgamma\supm-\bgamma_0\|_2,
\quad
\|\widehat\bgamma_{{\rm IPD},[T-1]}-\bgamma_0\|_2,
\]
together with either a score fluctuation of order
\(\sqrt{\log p/N_m}\) or another direction error. Using Theorems~\ref{thm:local}--\ref{thm:aggregate},
\eqref{eqn:ipd-rate-temp}--\eqref{eqn:ipd-rate-final}, and
\(s=O\{N_m^{1/6}/\log(p\vee N_m)\}\), their weighted sum is
\(o_p(s\lambda_N^2)\). The labeled-score difference in
\eqref{eqn:Rdagger} is bounded by
\[
\frac n{N+n}\|\Delta_\dagger\|_2
\|\widehat\btheta-\widehat\btheta_{\subipd}\|_2,
\]
and the last integral is
\(o_p(\|\Delta_\dagger\|_2^2)\). Together with the preliminary rate bounds, these terms are all \(o_p(s\lambda_N^2)\), proving
\eqref{eqn:R0}--\eqref{eqn:Rdagger}.

\paragraph{Comparison of the temporary direction estimators.}
Let \(\rho(u)=\|u\|_1\). The KKT equation for the last pooled update is
\begin{equation}\label{eqn:ipd-kkt-temp}
2\widehat\Xi^{\subipd}
+2\widehat\Omega_{x,x}^{\subipd}
(\widehat\bgamma_{\subipd}-
\widehat\bgamma_{{\rm IPD},[T-1]})
+\lambda_{{\scriptscriptstyle\sf IPD},0}\rho'(\widehat\bgamma_{\subipd})=0.
\end{equation}
Choose
\(
\lambda=
\lambda_{{\scriptscriptstyle\sf IPD},0}+\lambda_{\delta,0}
\)
with
\(
\lambda_{\delta,0}=o(\lambda_N)
\)
and, along a deterministic sequence if necessary,
\(
|\mathcal R_{0,N}|=o_p(\lambda_{\delta,0}s\lambda_N).
\)
The basic inequality for the SASH summary objective, evaluated at
\(\widehat\bgamma_{\subipd}\), and~\eqref{eqn:ipd-kkt-temp} give
\begin{equation}\label{eqn:temp-basic-comp}
\Delta_0\trans
\left\{\frac1N\sum_{m=1}^MN_m
\widehat\Omega_{x,x}\supm\right\}\Delta_0
\le
2\mathcal R_{0,N}
+\lambda_{\delta,0}
\{\|\widehat\bgamma_{\subipd}\|_1-
\|\widehat\bgamma\|_1\}.
\end{equation}
By the rates of the two estimators,
\(\|\Delta_0\|_1=O_p(r_{1,N})\). Lemma~\ref{lem:re}, now applied on the full normalized tangent space in Assumption~\ref{ass:4}, and the Hessian approximation bounds used in its proof imply
\[
\Delta_0\trans
\left\{\frac1N\sum_mN_m\widehat\Omega_{x,x}\supm\right\}\Delta_0
\ge c\|\Delta_0\|_2^2-o_p(s\lambda_N^2).
\]
Combining this inequality with~\eqref{eqn:temp-basic-comp} yields
\begin{equation}\label{eqn:temp-l2-comp}
\|\widehat\bgamma-\widehat\bgamma_{\subipd}\|_2
=o_p(r_{2,N}).
\end{equation}
For later use, let \(S=\operatorname{supp}(\bgamma_0)\). Equation~\eqref{eqn:temp-l2-comp} gives
\(\|\Delta_{0,S}\|_1=o_p(r_{1,N})\). From
\eqref{eqn:temp-basic-comp}, decomposability of the \(\ell_1\) norm, and the choice of
\(\lambda_{\delta,0}\),
\[
\|\widehat\bgamma_{S^c}\|_1
\le
\|\widehat\bgamma_{\subipd,S^c}\|_1
+
\|\Delta_{0,S}\|_1
+
2|\mathcal R_{0,N}|/\lambda_{\delta,0}.
\]
Consequently,
\begin{equation}\label{eqn:temp-l1-comp}
\|\widehat\bgamma-\bgamma_0\|_1
\le
\|\widehat\bgamma_{\subipd}-\bgamma_0\|_1
+o_p(r_{1,N}).
\end{equation}

\paragraph{Comparison of the low-dimensional refits.}
Let
\(
U_n(\btheta,\bgamma)=n^{-1}\sum_{i=1}^n
\nabla_{\btheta}\ell_i(\btheta,\bgamma).
\)
The score equations give
\(U_n(\widehat\btheta,\widehat\bgamma)=0\) and
\(U_n(\widehat\btheta_{\subipd},\widehat\bgamma_{\subipd})=0\).
A mean-value expansion therefore yields
\[
0=A_n(\widehat\btheta-\widehat\btheta_{\subipd})
+B_n(\widehat\bgamma-\widehat\bgamma_{\subipd}),
\]
where \(A_n\) is an average Hessian in \(\btheta\) and \(B_n\) is an average cross derivative. Assumption~\ref{ass:1} and the consistency rates imply
\(\lambda_{\min}(A_n)\ge c>0\) and
\(\|B_n\|_2=O_p(1)\). Consequently,
\begin{equation}\label{eqn:theta-comp}
\|\widehat\btheta-\widehat\btheta_{\subipd}\|_2
\lesssim
\|\widehat\bgamma-\widehat\bgamma_{\subipd}\|_2
=o_p(r_{2,N}).
\end{equation}

\paragraph{Comparison of the final direction estimators.}
The KKT equation for~\eqref{eqn:IPD-III-2} is
\begin{equation}\label{eqn:ipd-kkt-final}
\frac1{N+n}\left[
2N\widehat\Xi_{\subipd}^{\dagger}
+2N\widehat H_{\subipd}^{\dagger}
(\widehat\bgamma_{\subipd}^{\dagger}-
\widehat\bgamma_{\subipd})
+
\sum_{i=1}^n
\nabla_{\bgamma}\ell_i(
\widehat\btheta_{\subipd},
\widehat\bgamma_{\subipd}^{\dagger})
\right]
+
\lambda_{\subipd}\rho'(\widehat\bgamma_{\subipd}^{\dagger})=0.
\end{equation}
Choose
\(
\lambda^\dagger=\lambda_{\subipd}+\lambda_\delta
\)
with \(\lambda_\delta=o(\lambda_N)\) and
\[
|\mathcal R_{\dagger,N}|
=o_p\left(\lambda_\delta r_{1,N}\right).
\]
The basic inequality for~\eqref{eqn:step3-3}, evaluated at
\(\widehat\bgamma_{\subipd}^{\dagger}\), together with
\eqref{eqn:ipd-kkt-final}, gives
\begin{align}
&\Delta_\dagger\trans
\left[
\frac1{N+n}
\left\{
\sum_{m=1}^MN_m\widehat\Omega_{x,x}\supm
+
\sum_{i=1}^n
\nabla_{\bgamma}^2\ell_i(
\widehat\btheta,
\widehat\bgamma_{\subipd}^{\dagger})
\right\}
\right]\Delta_\dagger
\nonumber\\
&\qquad\le
2\mathcal R_{\dagger,N}
+
\lambda_\delta
\{\|\widehat\bgamma_{\subipd}^{\dagger}\|_1-
\|\widehat\bgamma^\dagger\|_1\}.
\label{eqn:final-basic-comp}
\end{align}
The surrogate component of the matrix on the left is uniformly positive on
\(\mathbb V_0\) by Assumption~\ref{ass:4} and Lemma~\ref{lem:re}; the labeled Hessian is nonnegative. Hence
\eqref{eqn:final-basic-comp} and the rate bounds yield
\begin{equation}\label{eqn:final-direction-comp}
\|\widehat\bgamma^\dagger-
\widehat\bgamma_{\subipd}^{\dagger}\|_2
=o_p(r_{2,N}).
\end{equation}
Moreover, writing \(\Delta_\dagger=\widehat\bgamma^\dagger-
\widehat\bgamma_{\subipd}^{\dagger}\), we have
\(\|\Delta_{\dagger,S}\|_1=o_p(r_{1,N})\). Decomposability in
\eqref{eqn:final-basic-comp} further gives
\begin{equation}\label{eqn:final-l1-one-sided}
\|\widehat\bgamma^\dagger_{S^c}\|_1
\le
\|\widehat\bgamma_{\subipd,S^c}^{\dagger}\|_1
+o_p(r_{1,N}),
\end{equation}
and therefore
\[
\|\widehat\bgamma^\dagger-\bgamma_0\|_1
\le
\|\widehat\bgamma_{\subipd}^{\dagger}-\bgamma_0\|_1
+o_p(r_{1,N}).
\]

\paragraph{Conclusion.}
Using~\eqref{eqn:theta-comp} and~\eqref{eqn:final-direction-comp},
\[
\|\widehat\bbeta\subsash-
\widehat\bbeta\subipd\|_2
\le
|\widehat\beta_1-
\widehat\beta_{\subipd,1}|
\|\widehat\bgamma^\dagger\|_2
+
|\widehat\beta_{\subipd,1}|
\|\widehat\bgamma^\dagger-
\widehat\bgamma_{\subipd}^{\dagger}\|_2
=o_p(r_{2,N}).
\]
The \(\ell_2\) comparison in the theorem follows by the triangle inequality. For the \(\ell_1\) comparison, decompose the errors over \(S\) and \(S^c\). On \(S\),
\eqref{eqn:theta-comp} and~\eqref{eqn:final-direction-comp} give
\[
\|\{\widehat\bbeta\subsash-
\widehat\bbeta\subipd\}_S\|_1
\le \sqrt{s}\|\widehat\bbeta\subsash-
\widehat\bbeta\subipd\|_2
=o_p(r_{1,N}).
\]
On \(S^c\), where \(\bbeta_{0,S^c}=0\), equation~\eqref{eqn:final-l1-one-sided} and
\(|\widehat\beta_1-\widehat\beta_{\subipd,1}|=o_p(r_{2,N})\) imply
\[
\|(\widehat\bbeta\subsash)_{S^c}\|_1
\le
\|(\widehat\bbeta\subipd)_{S^c}\|_1
+o_p(r_{1,N}).
\]
Adding the two parts yields
\[
\|\widehat\bbeta\subsash-\bbeta_0\|_1
\le
\|\widehat\bbeta\subipd-\bbeta_0\|_1
+o_p(r_{1,N}),
\]
which is the stated \(\ell_1\) no-rate-loss comparison because
\(r_{1,N}\le\sqrt{s/n}+r_{1,N}\). This completes the proof.
\end{proof}

\section{Auxiliary Proofs}\label{sec:lem}

\subsection{Proof of Proposition \ref{prop:1}}\label{sec:proof:prop}

\begin{proof}
Under the outcome model (\ref{equ:asu:1}) the conditional independence condition (\ref{eqn:asu:indp}), for any $m\in\{1,\ldots,M\}$ and set $\mathcal{A}$ included by $S\supm_i$'s domain, we have
\begin{small}
\begin{align*}
\Pbb(S\supm_i \in \mathcal{A} \mid \bX\supm_i)
&= \sum_{y=0,1} \Pbb(S\supm_i\in  \mathcal{A}, Y\supm_i=y \mid \bX\supm_i) \\
&= \sum_{y=0,1} \Pbb(S\supm_i\in  \mathcal{A} \mid Y\supm_i=y, \bX\supm_i) 
\Pbb(Y\supm_i=y \mid \bX\supm_i) \\
&= \sum_{y=0,1} \Pbb(S\supm_i\in  \mathcal{A} \mid Y\supm_i=y) \Pbb(Y\supm_i=y \mid \bX\supm_i)\\
&= \sum_{y=0,1} \Pbb(S\supm_i\in  \mathcal{A} \mid Y\supm_i=y) g_y({ \alpha_{0}\supm}+\bbeta_0\trans\bX_i\supm),
\end{align*}
\end{small}
where $g_y(a)=yg(a)+(1-y)(1-g(a))$. This directly leads to the SIM condition:
\[ \small
\Ebb(S\supm_{i}\mid \bX\supm_{i} )=\breve f_m(\bbeta_0\trans \bX\supm_i)= f_{m}(\bgamma_0\trans \bX\supm_i ),
\]
for some $\breve f_m$ and $f_{m}$.

\end{proof}

\vspace{-30pt}

{
\subsection{Proof of Lemma \ref{lem:1}}\label{sec:proof:lem:1}

\begin{proof}
{\blue Let $\widetilde\btheta\subsup=(\widetilde\alpha\subsup,\widetilde\bbeta\subsup\trans)\trans$ and $\btheta_0=(\alpha_0,\bbeta_0\trans)\trans$. The joint-information condition in Assumption \ref{ass:1}, together with the sub-Gaussian design, gives restricted strong convexity for the logistic loss after accounting for the unpenalized intercept and the unpenalized first coefficient. Standard decomposable-regularizer arguments \citep{negahban2012unified} therefore give
\[
\|\widetilde\bbeta\subsup-\bbeta_0\|_2
=O_p\left(\sqrt{\frac{s\log p}{n}}\right),
\qquad
\|\widetilde\bbeta\subsup-\bbeta_0\|_1
=O_p\left(s\sqrt{\frac{\log p}{n}}\right),
\]
and $|\widetilde\alpha\subsup-\alpha_0|=O_p\{\sqrt{s\log p/n}\}$. Since $s\log p/(n\beta_{01}^2)=o(1)$,
\[
|\widetilde\beta_{{\scriptscriptstyle\sf sup},1}-\beta_{01}|
\le\|\widetilde\bbeta\subsup-\bbeta_0\|_2=o_p(|\beta_{01}|),
\]
so $|\widetilde\beta_{{\scriptscriptstyle\sf sup},1}|\ge|\beta_{01}|/2$ with probability approaching one. On this event,
\[
\widetilde\bgamma\subsup-\bgamma_0
=
\frac{\widetilde\bbeta\subsup-\bbeta_0}{\widetilde\beta_{{\scriptscriptstyle\sf sup},1}}
+
\bbeta_0\left(
\frac{1}{\widetilde\beta_{{\scriptscriptstyle\sf sup},1}}
-
\frac{1}{\beta_{01}}
\right).
\]
Using $\|\bbeta_0\|_2\le\kappa$, $\|\bbeta_0\|_1\le\sqrt{s}\|\bbeta_0\|_2$, and the two rates above yields
\[
\|\widetilde\bgamma\subsup-\bgamma_0\|_2
\lesssim_p\sqrt{\frac{s\log p}{n\beta_{01}^2}},
\qquad
\|\widetilde\bgamma\subsup-\bgamma_0\|_1
\lesssim_p s\sqrt{\frac{\log p}{n\beta_{01}^2}}.
\]
The first coordinate equals zero after normalization. Choosing $C_\Gamma$ sufficiently large gives $\widetilde\bgamma\subsup\in\Gamma$ with probability approaching one.}
\end{proof}

}

\vspace{-15pt}

\subsection{Auxiliary Lemmas Used in Section \ref{sec:A}}\label{sec:lem1}

\begin{lemma}\label{lem:A1}
Under the same conditions as Theorem \ref{thm:local},
\[
\sup_{\bgamma\in\Gamma}
\left\|
\frac{1}{N_m}\sum_{i=1}^{N_m}
 f'_m(\bgamma_0\trans\bX_i\supm)
 \left\{\bX_i\supm-\Ebb(\bX\mid\bgamma_0\trans\bX_i\supm)\right\}
 \epsilon_i\supm
\right\|_\infty
\le C\sqrt{\frac{\log p}{N_m}}
\]
with probability exceeding $1-\exp(-c_1\log p)$.
\end{lemma}

\begin{proof}
{\blue Conditional on $\bX_i\supm$, the summand has mean zero because
$\Ebb(\epsilon_i\supm\mid\bX_i\supm)=0$. The bounded derivative
$f'_m$, together with Assumption~\ref{ass:x}, gives a uniformly
sub-exponential coordinate envelope. Bernstein's inequality and a union bound
over the $p$ coordinates yield the stated rate.}
\end{proof}

\vspace{-25pt}

\begin{lemma}\label{lem:A2}
Under the same conditions as Theorem \ref{thm:local}, we have 
$$\small 
\sup_{\bgamma \in \Gamma} \bigg\| \frac{1}{N_m}\sum_{i=1}^{N_m} \Big\{\partial_{\bgamma}\fhat_m(\bX_i\supm;\bgamma)-f'_m(\bgamma_0\trans\bX\supm_i) \big(\bX\supm_i - \Ebb\big(\bX \biggiven \bgamma_0\trans\bX_i\supm \big) \big) \Big\} \epsilon\supm_{i} \bigg\|_{\infty}
\leq C \sqrt{\frac{\log p}{N_m}}$$
with probability exceeding $1-\exp(-c_1\log p)$. 
{\blue The uniformity over $\Gamma$ follows from the kernel-class entropy condition in Assumption \ref{ass:fm}(c).}

\begin{proof}
Define $$\footnotesize \widehat{G}^{(1)}(\bgamma\trans\bX \mid \bgamma):= \frac{\sum_{i=1}^{N_m} K_h'(\bgamma\trans\{\bX_i\supm-\bX\})S\supm_i}{\sum_{i=1}^{N_m} K_h(\bgamma\trans\{\bX_i\supm-\bX\})} 
- \frac{\sum_{i=1}^{N_m} K_h(\bgamma\trans\{\bX_i\supm-\bX\})S\supm_i\sum_{i=1}^{N_m} K_h'(\bgamma\trans\{\bX_i\supm-\bX\})}{\left[\sum_{i=1}^{N_m} K_h(\bgamma\trans\{\bX_i\supm-\bX\})\right]^2}.$$
First, we have 
\begin{scriptsize}
\begin{equation}\label{eqn:A4}
\sup_{\bgamma \in \Gamma} \bigg\| \frac{1}{N_m}\sum_{i=1}^{N_m} \Big\{ \Big( \widehat{G}^{(1)}(\bgamma\trans\bX\supm_i \mid \bgamma) - f'_m(\bgamma_0\trans\bX\supm_i) \Big) \Big(\bX\supm_i - \Ebb\big(\bX \biggiven \bgamma_0\trans\bX_i\supm \big)\Big) \Big\} \epsilon\supm_{i} \bigg\|_{\infty}
\lesssim  \frac{\left[h_m^{2} \log (p \vee N_m)\right]^{1 / 4}}{\sqrt{N_m}}    
\end{equation}\end{scriptsize} directly from Lemma A8 in \cite{wu2021model} under the assumption that $\small \Big(\frac{s\log p}{N_m}\Big)^{1/5} \lesssim h_m \lesssim  N_m^{-1/6}$.
Note that from Page 114 in \cite{wu2021model}, even though we have $\small \|\bgamma-\bgamma_{0}\|_{2} \lesssim \sqrt{\frac{s\log p}{n\beta_{01}^2}}$ instead, we still have the same bound for \eqref{eqn:A4} since $\small \sqrt{\frac{s\log p}{n}} / 2^{-l} h_m \lesssim \sqrt{\frac{s\log p}{n}} \sqrt{N_m} h_m \lesssim N_m$ still holds.

Then it suffices to bound
$$ \small \sup_{\bgamma \in \Gamma} \bigg\| \frac{1}{N_m}\sum_{i=1}^{N_m} \Big\{\partial_{\bgamma}\fhat_m(\bX_i\supm;\bgamma)-\widehat{G}^{(1)}(\bgamma\trans\bX\supm_i \mid \bgamma) \big(\bX\supm_i - \Ebb\big(\bX \biggiven \bgamma_0\trans\bX_i\supm \big) \big) \Big\} \epsilon\supm_{i} \bigg\|_{\infty}.
$$
We have 
\begin{scriptsize}
\begin{equation}\nonumber
\begin{aligned}
& \partial_{\bgamma}\fhat_m(\bX_i\supm;\bgamma)-\widehat{G}^{(1)}(\bgamma\trans\bX\supm_i \mid \bgamma) \big(\bX\supm_i - \Ebb\big(\bX \biggiven \bgamma_0\trans\bX_i\supm \big) \big) \\
= & \frac{\sum_{j=1}^{N_m} K_h'(\bgamma\trans\{\bX_j\supm-\bX_i\supm\})S\supm_j(\bX_j\supm-\bX_i\supm)}{\sum_{j=1}^{N_m} K_h(\bgamma\trans\{\bX_j\supm-\bX_i\supm\})}
-\frac{\sum_{j=1}^{N_m} K_h(\bgamma\trans\{\bX_j\supm-\bX_i\supm\})S\supm_j\sum_{j=1}^{N_m} K_h'(\bgamma\trans\{\bX_j\supm-\bX_i\supm\})(\bX_j\supm-\bX_i\supm)}{\left[\sum_{j=1}^{N_m} K_h(\bgamma\trans\{\bX_j\supm-\bX_i\supm\})\right]^2} \\
& - \bigg(\frac{\sum_{j=1}^{N_m} K_h'(\bgamma\trans\{\bX_j\supm-\bX_i\supm\})S\supm_j}{\sum_{j=1}^{N_m} K_h(\bgamma\trans\{\bX_j\supm-\bX_i\supm\})} 
- \frac{\sum_{j=1}^{N_m} K_h(\bgamma\trans\{\bX_j\supm-\bX_i\supm\})S\supm_j\sum_{j=1}^{N_m} K_h'(\bgamma\trans\{\bX_j\supm-\bX_i\supm\})}{\left[\sum_{j=1}^{N_m} K_h(\bgamma\trans\{\bX_j\supm-\bX_i\supm\})\right]^2}\bigg)
 \cdot \big(\bX\supm_i - \Ebb\big(\bX \biggiven \bgamma_0\trans\bX_i\supm \big) \big) \\
= & \frac{\sum_{j=1}^{N_m} K_h'(\bgamma\trans\{\bX_j\supm-\bX_i\supm\})S\supm_j\big(\bX\supm_j - \Ebb\big(\bX \biggiven \bgamma_0\trans\bX_i\supm \big)\big)}{\sum_{j=1}^{N_m} K_h(\bgamma\trans\{\bX_j\supm-\bX_i\supm\})}\\
&-\frac{\sum_{j=1}^{N_m} K_h(\bgamma\trans\{\bX_j\supm-\bX_i\supm\})S\supm_j\sum_{j=1}^{N_m} K_h'(\bgamma\trans\{\bX_j\supm-\bX_i\supm\})\big(\bX\supm_j - \Ebb\big(\bX \biggiven \bgamma_0\trans\bX_i\supm \big)\big)}{\left[\sum_{j=1}^{N_m} K_h(\bgamma\trans\{\bX_j\supm-\bX_i\supm\})\right]^2}
\end{aligned}    
\end{equation}
\end{scriptsize}
depends on $\bX\supm_i$ only through $\bgamma\trans\bX\supm_i$ and $\bgamma_0\trans\bX\supm_i$. 
%
%
%
Then, by arguments similar to those in Lemmas A11 and A8 of \cite{wu2021model}, we have
\begin{equation}\nonumber
\small
\max_{i} \sup_{\bgamma \in \Gamma} \Big\|  \partial_{\bgamma}\fhat_m(\bX_i\supm;\bgamma)-\widehat{G}^{(1)}(\bgamma\trans\bX\supm_i \mid \bgamma) \big(\bX\supm_i-\Ebb\big(\bX \biggiven \bgamma_0\trans\bX_i\supm \big) \big) \Big\|_{\infty} \leq C_0 h_m, 
\end{equation}
\begin{equation}\label{eqn:A9}
\small
\sup_{\bgamma \in \Gamma} \bigg\| \frac{1}{N_m}\sum_{i=1}^{N_m} \Big\{\partial_{\bgamma}\fhat_m(\bX_i\supm;\bgamma)-\widehat{G}^{(1)}(\bgamma\trans\bX\supm_i \mid \bgamma) \Big(\bX\supm_i - \Ebb\big(\bX \biggiven \bgamma_0\trans\bX_i\supm \big)\Big) \Big\} \epsilon\supm_{i} \bigg\|_{\infty}
\lesssim 
\frac{\left[h_m^{2} \log (p \vee N_m)\right]^{1 / 4}}{\sqrt{N_m}}.
\end{equation} 

Combining \eqref{eqn:A4}, \eqref{eqn:A9} and the assumption of Theorem \ref{thm:local},
finally we have 
$$\small \sup_{\bgamma \in \Gamma} \bigg\| \frac{1}{N_m}\sum_{i=1}^{N_m} \Big\{\partial_{\bgamma}\fhat_m(\bX_i\supm;\bgamma)-f'_m(\bgamma_0\trans\bX\supm_i) \Big(\bX\supm_i - \Ebb\big(\bX \biggiven \bgamma_0\trans\bX_i\supm \big)\Big) \Big\} \epsilon\supm_{i} \bigg\|_{\infty}
\lesssim \sqrt{\frac{\log p }{N_m}}$$
with probability exceeding $1-\exp(-c_1\log p)$.

\end{proof}

\end{lemma}

\begin{lemma}\label{lem:A3}
Under the same conditions as Theorem \ref{thm:local}, we have 
\begin{small}
\begin{equation}\nonumber
\begin{aligned}
& \Big| \frac{1}{N_m}\sum_{i=1}^{N_m} (\widehat\bgamma_{[t]}\supm-\bgamma_{0})\trans\partial_{\bgamma}\fhat_m(\bX_i\supm;\widehat\bgamma_{[t-1]}\supm)
\cdot \Big\{f_m(\bgamma_0\trans\bX_i\supm)- \fhat_m(\bX_i\supm;\bgamma_0) \Big\}
\Big| \\
\leq & 
C\Big\{ 
\|\Delta_{t}\supm\|_1\sqrt{\frac{1}{N_m}} + \|\Delta_{t}\supm\|_1 h_m^3 + \|\Delta_{t}\supm\|_2\|\Delta_{t-1}\supm\|_2 h_m^2
\Big\},
\end{aligned}    
\end{equation}
\end{small}
where $\Delta_{t}\supm = \widehat\bgamma_{[t]}\supm-\bgamma_{0}$.

\begin{proof}
By the bound for $I_{n4}$ in the proof of Lemma A1 
in \cite{wu2021model}, and $h_m[\log(p\vee N_m)]^{1/4} \leq c$ implied by the assumption of Theorem \ref{thm:local}, we immediately obtain
\begin{equation}\nonumber
\begin{aligned}
\footnotesize
\Big\| \frac{1}{N_m}\sum_{i=1}^{N_m} f'(\bgamma_0\trans\bX_i\supm)  
\Big(\bX_i\supm - \Ebb[\bX|\bgamma_0\trans\bX_i\supm]\Big)
\cdot \Big\{f_m(\bgamma_0\trans\bX_i\supm)- \fhat_m(\bX_i\supm;\bgamma_0) \Big\} \Big\|_{\infty}
\lesssim  \frac{h_m[\log(p\vee N_m)]^{1/4} }{\sqrt{N_m}}
\lesssim  \sqrt{\frac{1}{N_m}}
\end{aligned}    
\end{equation}
with probability at least $1-\exp(-c_1 \log p)$ for some constant $c_1>0$.
Then it suffices to bound
\begin{scriptsize}
\begin{equation}\nonumber
\begin{aligned}
&  \Big| \frac{1}{N_m}\sum_{i=1}^{N_m} (\widehat\bgamma_{[t]}\supm-\bgamma_{0})\trans \Big\{
\partial_{\bgamma}\fhat_m(\bX_i\supm;\widehat\bgamma_{[t-1]}\supm)  - f'_m(\bgamma_0\trans\bX\supm_i) \Big(\bX\supm_i - \Ebb\big(\bX \biggiven \bgamma_0\trans\bX_i\supm \big) \Big)
\Big\}
\cdot \Big\{f_m(\bgamma_0\trans\bX_i\supm)- \fhat_m(\bX_i\supm;\bgamma_0) \Big\} \Big|
\\
\leq &  \max_{i:\,>0} \sup_{\bgamma \in \Gamma} \Big|(\widehat\bgamma_{[t]}\supm-\bgamma_{0})\trans \Big(\partial_{\bgamma}\fhat_m(\bX_i\supm;\bgamma)  - f'_m(\bgamma\trans\bX\supm_i) \big(\bX\supm_i-\Ebb\big(\bX \biggiven \bgamma\trans\bX_i\supm \big) \big) \Big) \Big| 
\cdot \max_{i:\,>0} \Big|f_m(\bgamma_0\trans\bX_i\supm)- \fhat_m(\bX_i\supm;\bgamma_0) \Big| 
\\
& + 
 \Big| \frac{1}{N_m}\sum_{i=1}^{N_m} (\widehat\bgamma_{[t]}\supm-\bgamma_{0})\trans f'_m(\widehat\bgamma_{[t-1]}\supmtrans\bX\supm_i) \Big(\Ebb\big(\bX \biggiven {\widehat\bgamma_{[t-1]}\supm}{}\trans\bX_i\supm \big) - \Ebb\big(\bX \biggiven \bgamma_0\trans\bX_i\supm \big) \Big) 
\cdot \Big(f_m(\bgamma_0\trans\bX_i\supm)- \fhat_m(\bX_i\supm;\bgamma_0) \Big) \Big| \\
& + 
 \Big| \frac{1}{N_m}\sum_{i=1}^{N_m} (\widehat\bgamma_{[t]}\supm-\bgamma_{0})\trans \Big( f'_m(\widehat\bgamma_{[t-1]}\supmtrans\bX\supm_i)  - f'_m(\bgamma_0\trans\bX\supm_i) \Big)   \Big(\bX\supm_i - \Ebb\big(\bX \biggiven \bgamma_0\trans\bX_i\supm \big) \Big)
\cdot \Big(f_m(\bgamma_0\trans\bX_i\supm)- \fhat_m(\bX_i\supm;\bgamma_0) \Big)  \Big|  \\
:= & I_1 + I_2 + I_3.
\end{aligned}    
\end{equation}
\end{scriptsize}
We have 
\begin{scriptsize}
\begin{equation}\nonumber
\begin{aligned}
I_1 & = 
\max_{i:\,>0} \sup_{\bgamma \in \Gamma} \Big| (\widehat\bgamma_{[t]}\supm-\bgamma_{0})\trans \Big(\partial_{\bgamma}\fhat_m(\bX_i\supm;\bgamma)  - f'_m(\bgamma\trans\bX\supm_i) \big(\bX\supm_i-\Ebb\big(\bX \biggiven \bgamma\trans\bX_i\supm \big) \big) \Big) \Big|
\cdot \max_{i:\,>0} \Big|f_m(\bgamma_0\trans\bX_i\supm)- \fhat_m(\bX_i\supm;\bgamma_0) \Big|  \lesssim \|\Delta_{t}\supm\|_{1 } h_m^3,
\end{aligned}    
\end{equation}
\end{scriptsize}
holds by Lemma \ref{lem:partial-fhat}  and Lemma \ref{lem:fhat}.

\begin{scriptsize}
\begin{equation}\nonumber
\begin{aligned}
I_2 & =  
\Big| \frac{1}{N_m}\sum_{i=1}^{N_m} (\widehat\bgamma_{[t]}\supm-\bgamma_{0})\trans f'_m(\widehat\bgamma_{[t-1]}\supmtrans\bX\supm_i)
\Big(\Ebb\big(\bX \biggiven {\widehat\bgamma_{[t-1]}} \supmtrans\bX_i\supm \big) - \Ebb\big(\bX \biggiven \bgamma_0\trans\bX_i\supm \big) \Big) 
\cdot 
\Big(f_m(\bgamma_0\trans\bX_i\supm)- \fhat_m(\bX_i\supm;\bgamma_0) \Big) \Big|
\\
& =
\Big| \frac{1}{N_m}\sum_{i=1}^{N_m} (\widehat\bgamma_{[t]}\supm-\bgamma_{0})\trans  f'_m(\widehat\bgamma_{[t-1]}\supmtrans\bX\supm_i)  \Big(f_m(\bgamma_0\trans\bX_i\supm)- \fhat_m(\bX_i\supm;\bgamma_0)\Big)
\cdot 
\Big( \widehat\bgamma_{[t-1]}\supmtrans\bX_i\supm - \bgamma_0\trans\bX_i\supm  
 \Big)
\\
& \ \  \cdot
\int_0^1 \Ebb^{(1)}\big(\bX \biggiven \bgamma_0\trans\bX_i\supm  + t(\widehat\bgamma_{[t-1]}\supm-\bgamma_0)\trans\bX_i\supm  \big) \d t
 \Big| 
\\ 
& = 
\Big| \frac{1}{N_m}\sum_{i=1}^{N_m} f'_m(\widehat\bgamma_{[t-1]}\supmtrans\bX\supm_i)  \Big(f_m(\bgamma_0\trans\bX_i\supm)- \fhat_m(\bX_i\supm;\bgamma_0) \Big) (\widehat\bgamma_{[t-1]}\supm-\bgamma_{0})\trans \bX_i\supm 
\\
& \ \ \cdot
\int_0^1 \Ebb^{(1)}\big((\widehat\bgamma_{[t]}\supm-\bgamma_0)\trans\bX \biggiven \bgamma_0\trans\bX_i\supm  + t(\widehat\bgamma_{[t-1]}\supm-\bgamma_0)\trans\bX_i\supm  \big) \d t
\Big| \\
& \lesssim  \|\widehat\bgamma_{[t]}\supm-\bgamma_{0} \|_2 \|\widehat\bgamma_{[t-1]}\supm-\bgamma_0\|_2 h_m^2
\end{aligned}    
\end{equation}
\end{scriptsize}
where the last inequality follows by 
Assumption \ref{ass:fm}  that $ f_m'(\bgamma\trans\bX_i\supm)$ is bounded uniformly, Assumption \ref{ass:5} that 
\[ \small
\max_{i:\,>0} \sup_{\bgamma \in \Gamma} \big|  \Ebb^{(1)}(\bv \trans \bX \mid \bgamma\trans\bX_i\supm)\big| \leq C \|\bv \|_{2},
\]
Lemma \ref{lem:fhat} that  $\small
\max_{i:\,>0} \sup_{\bgamma \in \Gamma} \Big|f_m(\bgamma\trans\bX_i\supm)- \fhat_m(\bX_i\supm;\bgamma) \Big| 
\leq  c_0h_m^2$, 
and Lemma A3 of \cite{wu2021model} that 
$ \small
\Big| \frac{1}{N_m}\sum_{i=1}^{N_m} (\boldsymbol{v}\trans\bX_i\supm)^2 \Big| \lesssim \|\boldsymbol{v}\|_2^2$.
Similarly, we have
\begin{scriptsize}
\begin{equation}\nonumber
\begin{aligned}
 I_3 & =  
\Big| \frac{1}{N_m}\sum_{i=1}^{N_m} (\widehat\bgamma_{[t]}\supm-\bgamma_{0})\trans \Big( f'_m(\widehat\bgamma_{[t-1]}\supmtrans\bX\supm_i)  - f'_m(\bgamma_0\trans\bX\supm_i) \Big) 
\Big(\bX\supm_i - \Ebb\big(\bX \biggiven \bgamma_0\trans\bX_i\supm \big) \Big)\Big(f_m(\bgamma_0\trans\bX_i\supm)- \fhat_m(\bX_i\supm;\bgamma_0) \Big) \Big|
\\
 & \leq
\Big| \frac{1}{N_m}\sum_{i=1}^{N_m} (\widehat\bgamma_{[t-1]}\supm-\bgamma_0)\trans \bX_i\supm \int_0^1 f''\big(\bgamma_0\trans\bX_i\supm  + t(\widehat\bgamma_{[t-1]}\supm-\bgamma_0)\trans\bX_i\supm  \big) \d t \\
&\ \ \ \   \cdot  (\widehat\bgamma_{[t]}\supm-\bgamma_{0})\trans \Big(\bX\supm_i - \Ebb\big(\bX \biggiven \bgamma_0\trans\bX_i\supm \big) \Big)
\Big(f_m(\bgamma_0\trans\bX_i\supm)- \fhat_m(\bX_i\supm;\bgamma_0) \Big) \Big|
\\ 
& \leq 
\|\widehat\bgamma_{[t]}\supm-\bgamma_0\|_2 \|\widehat\bgamma_{[t-1]}\supm-\bgamma_0\|_2 h_m^2.
\end{aligned}    
\end{equation}
\end{scriptsize}
Thus, we achieve 
\begin{small}
\begin{equation}\nonumber
\begin{aligned}
& \Big| \frac{1}{N_m}\sum_{i=1}^{N_m} (\widehat\bgamma_{[t]}\supm-\bgamma_{0})\trans\partial_{\bgamma}\fhat_m(\bX_i\supm;\widehat\bgamma_{[t-1]}\supm)
\cdot \Big\{f_m(\bgamma_0\trans\bX_i\supm)- \fhat_m(\bX_i\supm;\bgamma_0) \Big\}
\Big|  \\
\leq &  C\Big\{ 
\|\Delta_{t}\supm\|_1\sqrt{\frac{1}{N_m}} + \|\Delta_{t}\supm\|_1 h_m^3 + \|\Delta_{t}\supm\|_2\|\Delta_{t-1}\supm\|_2 h_m^2
\Big\}
\end{aligned}    
\end{equation}
\end{small}
for some constant $C$,
where $\Delta_{t-1}\supm := \widehat\bgamma_{[t-1]}\supm-\bgamma_{0}$. 

\end{proof}

\end{lemma}

\vspace{-30pt}

\begin{lemma}\label{lem:A4}
Under the same conditions as Theorem \ref{thm:local}, we have
\begin{scriptsize}
\begin{equation}\nonumber
\begin{aligned}
    & \bigg| (\widehat\bgamma_{[t]}\supm-\bgamma_{0})\trans \frac{1}{N_m}\sum_{i=1}^{N_m} \partial_{\bgamma}\fhat_m(\bX_i\supm;\widehat\bgamma_{[t-1]}\supm) \bigg\{ (\bgamma_0-\widehat\bgamma_{[t-1]}\supm)\trans   
	\int_{0}^{1} \Big(\partial_{\bgamma}\fhat_m(\bX_i\supm; \widehat\bgamma_{[t-1]}\supm + \tau (\bgamma_0-\widehat\bgamma_{[t-1]}\supm)) - \partial_{\bgamma}\fhat_m(\bX_i\supm;\widehat\bgamma_{[t-1]}\supm) \Big) \mathrm{d} \tau  \bigg\}  \bigg| \\
	\leq & C\Big\{ 
	\|\Delta_{t}\supm\|_2 
  \|\Delta_{t-1}\supm\|_2 h_m
	+ 
	\|\Delta_{t}\supm\|_2
	 \|\Delta_{t-1}\supm\|_2^2
	 \Big\}.
\end{aligned}
\end{equation}
\end{scriptsize}

\begin{proof}
First, we aim to bound  

$$\small \bigg|\frac{1}{N_m}\sum_{i=1}^{N_m} \Big( (\widehat\bgamma_{[t]}\supm-\bgamma_{0})\trans \partial_{\bgamma}\fhat_m(\bX_i\supm;\widehat\bgamma_{[t-1]}\supm) \Big)^2 \bigg|.$$
We have
\begin{small}
\begin{equation}\nonumber
\begin{aligned}
\bigg| \frac{1}{N_m}\sum_{i=1}^{N_m} \Big( (\widehat\bgamma_{[t]}\supm-\bgamma_{0})\trans {\blue f'_m(\bgamma_0 \trans\bX\supm_i) \Big(\bX\supm_i-\Ebb\big(\bX \biggiven \bgamma_0 \trans\bX_i\supm \big) \Big)} \Big)^2 \bigg| 
	\leq 
C'\|\widehat\bgamma_{[t]}\supm-\bgamma_{0} \|_2^2
\end{aligned}
\end{equation}
\end{small}
by Lemma B2 of \cite{wu2021model}, 
\begin{footnotesize}
\begin{equation}\nonumber
\begin{aligned}
   & \Big|  \frac{1}{N_m}\sum_{i=1}^{N_m} \Big( (\widehat\bgamma_{[t]}\supm-\bgamma_{0})\trans\Big\{ {\blue f'_m(\widehat\bgamma_{[t-1]}\supmtrans\bX\supm_i) \Big(\bX\supm_i-\Ebb\big(\bX \biggiven \widehat\bgamma_{[t-1]}\supmtrans\bX_i\supm \big) \Big)} - {\blue f'_m(\bgamma_0 \trans\bX\supm_i) \Big(\bX\supm_i-\Ebb\big(\bX \biggiven \bgamma_0 \trans\bX_i\supm \big) \Big)} \Big\} \Big)^2 \Big| \\
   \leq & 
   C' \|\widehat\bgamma_{[t]}\supm-\bgamma_{0} \|_2^2\|\widehat\bgamma_{[t-1]}\supm-\bgamma_{0} \|_2^2
\end{aligned}
\end{equation}
\end{footnotesize}
holds by similar argument in the proof of Lemma \ref{lem:A3}, and 
\begin{footnotesize}
\begin{equation}\nonumber
\begin{aligned}
   &  \bigg| \frac{1}{N_m}\sum_{i=1}^{N_m} \Big(
   (\widehat\bgamma_{[t]}\supm-\bgamma_{0})\trans \Big\{ \partial_{\bgamma}\fhat_m(\bX_i\supm;\widehat\bgamma_{[t-1]}\supm)  - {\blue f'_m(\widehat\bgamma_{[t-1]}\supmtrans\bX\supm_i) \big(\bX\supm_i-\Ebb\big(\bX \biggiven \widehat\bgamma_{[t-1]}\supmtrans\bX_i\supm \big) \big)} \Big\} \Big)^2 \bigg| 
   	\leq 
   	C' \|\widehat\bgamma_{[t]}\supm-\bgamma_{0} \|_2^2 h_m^2  
\end{aligned}
\end{equation}
\end{footnotesize}
holds by Lemma \ref{lem:partial-fhat}.
By combining the above inequalities, we have

\begin{small}
\begin{equation}\label{eqn:A6}
\begin{aligned}
    & \bigg| \frac{1}{N_m}\sum_{i=1}^{N_m} \Big( (\widehat\bgamma_{[t]}\supm-\bgamma_{0})\trans \partial_{\bgamma}\fhat_m(\bX_i\supm;\widehat\bgamma_{[t-1]}\supm) \Big)^2 \bigg| \\
	\leq & 3C'
	\Big\{ \|\widehat\bgamma_{[t]}\supm-\bgamma_{0} \|_2^2
	+ \|\widehat\bgamma_{[t]}\supm-\bgamma_{0} \|_2^2\|\widehat\bgamma_{[t-1]}\supm-\bgamma_{0} \|_2^2 
	+ \|\widehat\bgamma_{[t]}\supm-\bgamma_{0} \|_2^2 h_m^2 \Big\}
	\leq C\|\widehat\bgamma_{[t]}\supm-\bgamma_{0} \|_2^2. 
\end{aligned}
\end{equation}
\end{small}

Then, we bound
\begin{footnotesize}
\begin{equation}\label{eqn:A7}
\begin{aligned}
    & \bigg| \frac{1}{N_m}\sum_{i=1}^{N_m} \bigg\{ (\bgamma_0-\widehat\bgamma_{[t-1]}\supm)\trans   
	\int_{0}^{1} \Big(\partial_{\bgamma}\fhat_m\big(\bX_i\supm; \widehat\bgamma_{[t-1]}\supm + \tau (\bgamma_0-\widehat\bgamma_{[t-1]}\supm) - \partial_{\bgamma}\fhat_m(\bX_i\supm;\widehat\bgamma_{[t-1]}\supm) \big) \Big) \mathrm{d} \tau  \bigg\}^2 \bigg| \\
	\lesssim &   
	\max_i \sup_{\bgamma \in \Gamma} \Big| (\bgamma_0-\widehat\bgamma_{[t-1]}\supm)\trans \Big( \partial_{\bgamma}\fhat_m (\bX_i\supm;\bgamma)  - {\blue f'_m(\bgamma\trans\bX\supm_i) \big(\bX\supm_i-\Ebb\big(\bX \biggiven \bgamma\trans\bX_i\supm \big) \big)} \Big) \Big|^2 \\
	& + \bigg| \frac{1}{N_m}\sum_{i=1}^{N_m} \bigg\{ (\bgamma_0-\widehat\bgamma_{[t-1]}\supm)\trans   
	\int_{0}^{1} \Big({\blue f'_m(\bgamma_\tau \trans\bX\supm_i) \Big(\bX\supm_i-\Ebb\big(\bX \biggiven \bgamma_\tau \trans\bX_i\supm \big) \Big)}  - {\blue f'_m(\widehat\bgamma_{[t-1]}\supmtrans\bX\supm_i) \Big(\bX\supm_i-\Ebb\big(\bX \biggiven \widehat\bgamma_{[t-1]}\supmtrans\bX_i\supm \big) \Big)} \Big) \mathrm{d} \tau  \bigg\}^2 \bigg| \\
	\lesssim & \|\bgamma_0-\widehat\bgamma_{[t-1]}\supm\|_2^2 h_m^2 + \|\bgamma_0-\widehat\bgamma_{[t-1]}\supm\|_2^4, 
\end{aligned}
\end{equation}
\end{footnotesize}
by an argument similar to that used in the proof of Lemma \ref{lem:A3}.
Here
$\bgamma_\tau := \widehat\bgamma_{[t-1]}\supm + \tau (\bgamma_0-\widehat\bgamma_{[t-1]}\supm)$.

In conclusion, by combining \eqref{eqn:A6} and \eqref{eqn:A7}, we have
\begin{footnotesize}
\begin{equation}\nonumber
\begin{aligned}
    & \bigg| (\widehat\bgamma_{[t]}\supm-\bgamma_{0})\trans \frac{1}{N_m}\sum_{i=1}^{N_m} \partial_{\bgamma}\fhat_m(\bX_i\supm;\widehat\bgamma_{t-1}\supm)   
	\cdot
 \bigg\{ (\bgamma_0-\widehat\bgamma_{t-1}\supm)\trans   
	\int_{0}^{1} \Big(\partial_{\bgamma}\fhat_m(\bX_i\supm; \widehat\bgamma_{[t-1]}\supm + \tau (\bgamma_0-\widehat\bgamma_{[t-1]}\supm) - \partial_{\bgamma}\fhat_m(\bX_i\supm;\widehat\bgamma_{t-1}\supm) \Big) \mathrm{d} \tau  \bigg\}  \bigg| \\
	\leq &
	C \sqrt{\|\widehat\bgamma_{[t]}\supm-\bgamma_{0} \|_2^2}
	\cdot 
	\sqrt{\|\bgamma_0-\widehat\bgamma_{[t-1]}\supm\|_2^2 h_m^2  + \|\bgamma_0-\widehat\bgamma_{[t-1]}\supm\|_2^4} \\
	\leq  &
	C \Big\{
	\|\widehat\bgamma_{[t]}\supm-\bgamma_{0} \|_2 
  \|\bgamma_0-\widehat\bgamma_{[t-1]}\supm\|_2 h_m
	+ 
	\|\widehat\bgamma_{[t]}\supm-\bgamma_{0} \|_2
	 \|\bgamma_0-\widehat\bgamma_{[t-1]}\supm\|_2^2
	 \Big\}
\end{aligned}
\end{equation}
\end{footnotesize}
for some constant $C$.

\end{proof}

\end{lemma}

\vspace{-30pt}

\begin{lemma}\label{lem:cone}
	Under the same conditions as Theorem \ref{thm:local} and assuming $   \sqrt{{\log p}/{N_m}}   + \|\Delta_{t-1}\supm\|_2 h_m  
+ \|\Delta_{t-1}\supm\|_2^2    < {\lambda_t\supm}/{2}$, we have $\|\Delta_{t,S^c}\supm\|_1 \leq 3\|\Delta_{t,S}\supm\|_1,$ where $S=\{j:\bgamma_{0,j} \neq 0\}$.
	
\begin{proof}
	By the assumption that $$  \sqrt{\frac{\log p}{N_m}}   + \|\Delta_{t-1}\supm\|_2 h_m 
+ \|\Delta_{t-1}\supm\|_2^2    < \frac{\lambda_t\supm}{2}$$ and \eqref{eqn:A2},
	 we have
     \begin{small}
	 \begin{equation}\label{eqn:A28}
	\begin{aligned}
	0 \leq  &  \frac{\lambda_t\supm}{2} \|\Delta\supm_t\|_1 + \lambda_t\supm\|\bgamma_{0}\|_1-\lambda_t\supm\|\widehat\bgamma\supm_{[t]}\|_1 \\
	= & \lambda_t\supm\ (\frac{1}{2}\|\Delta\supm_{t,S}\|_1 + \frac{1}{2} \|\Delta\supm_{t,S^c}\|_1 + \|\bgamma_{0,S}\|_1-\|\widehat\bgamma\supm_{t,S}\|_1-\|\widehat\bgamma\supm_{t,S^c}\|_1) \\
	\leq & \lambda_t\supm\ (\frac{1}{2}\|\Delta\supm_{t,S}\|_1 + \frac{1}{2}\|\Delta\supm_{t,S^c}\|_1 + \|\Delta\supm_{t,S}\|_1-\|\Delta\supm_{t,S^c}\|_1) \\
	\leq  & \lambda_t\supm\ (\frac{3}{2}\|\Delta\supm_{t,S}\|_1 - \frac{1}{2}\|\Delta\supm_{t,S^c}\|_1), 
	\end{aligned}
	 \end{equation}
     \end{small}
	where the equality holds by $\|\bgamma_{0}\|_1=\|\bgamma_{0,S}\|_1 + \|\bgamma_{0,S^c}\|_1=\|\bgamma_{0,S}\|_1$ and $\|\widehat\bgamma\supm_{[t]}\|_1 = \|\widehat\bgamma\supm_{t,S}\|_1+\|\widehat\bgamma\supm_{t,S^c}\|_1$. The second inequality holds by $\|\bgamma_{0,S}\|_1-\|\widehat\bgamma\supm_{t,S}\|_1 \leq  \|\Delta\supm_{t,S}\|_1$ and $ \|\widehat\bgamma\supm_{t,S^c}\|_1 =  \|\bgamma_{0,S^c}-\widehat\bgamma\supm_{t,S^c}\|_1 =  \|\Delta\supm_{t,S^c}\|_1$.
Then we immediately have $0 \leq 3\|\Delta\supm_{t,S}\|_1 - \|\Delta\supm_{t,S^c}\|_1$.
\end{proof}

\end{lemma}

\begin{lemma}\label{cor:1}
Under the same conditions in Lemma \ref{lem:cone}, we have
	$$\|\Delta\supm_{t}\|_1 \leq 4 \sqrt{s} \|\Delta\supm_{t}\|_2,$$
	where $s=|S|$ and $S=\{j:\bgamma_{0,j} \neq 0\}$.
\end{lemma}
\begin{proof}
From Lemma \ref{lem:cone}, we have $\|\Delta\supm_{t,S^c}\|_1 \leq 3\|\Delta\supm_{t,S}\|_1$, so
	\begin{align*}
	\|\Delta\supm_{t}\|_1
	= \|\Delta\supm_{t,S}\|_1 + \|\Delta\supm_{t,S^c}\|_1
	\leq 4\|\Delta\supm_{t,S}\|_1
	\leq 4\sqrt{s} \|\Delta\supm_{t,S}\|_2
	\leq 4\sqrt{s} \|\Delta\supm_{t}\|_2,
	\end{align*}
where $\|\Delta\supm_{t,S}\|_1
\leq \sqrt{s} \|\Delta\supm_{t,S}\|_2$ because vector $\Delta\supm_{t,S}$ has at most $s$ nonzero entries.
	
\end{proof}

\begin{lemma}\label{cor:2}
Under the same conditions in Lemma \ref{lem:cone}, we have
$$ \|\bgamma_{0}\|_1-\|\widehat\bgamma\supm_{[t]}\|_1 \leq 4 \sqrt{s} \|\Delta_{t}\supm\|_2,
$$
		where $s=|S|$ and $S=\{j:\bgamma_{0,j} \neq 0\}$.
		
\end{lemma}

\begin{proof}	
	We have
	\begin{align*}
	\|\bgamma_{0}\|_1-\|\widehat\bgamma\supm_{[t]}\|_1 
	= &  \|\bgamma_{0,S}\|_1 -\|\widehat\bgamma\supm_{t,S}\|_1 -\|\widehat\bgamma\supm_{t,S^c}\|_1  
	\leq    \|\Delta_{t,S}\supm\|_1 - \|\Delta_{t,S^c}\supm\|_1  
	\leq    \|\Delta_{t}\supm\|_1  
	\leq   4 \sqrt{s} \|\Delta_{t}\supm\|_2
	\end{align*}
	following the same argument in \eqref{eqn:A28}.
\end{proof}

\begin{lemma}[RE condition]\label{lem:re}
Under the conditions of Theorem \ref{thm:local}, there exists a positive constant $\kappa$ such that
\begin{small}
\begin{equation}\nonumber
\begin{aligned}
\frac{1}{N_m}\sum_{i=1}^{N_m}\Big\{(\widehat\bgamma\supm_{[t]}-\bgamma_{0})\trans\partial_{\bgamma}\fhat_m(\bX_i\supm;\widehat\bgamma_{[t-1]}\supm) \Big\}^2 \geq \kappa \|\widehat\bgamma\supm_{[t]}-\bgamma_{0}\|_2^2.
\end{aligned}    
\end{equation}
\end{small}

\begin{proof}
Since $a^2 \geq \frac{1}{2}b^2-(a-b)^2$, we have
\begin{footnotesize}
\begin{equation}\label{eqn:re}
\begin{aligned}
& \frac{1}{N_m}\sum_{i=1}^{N_m}\Big\{(\widehat\bgamma_{[t]}\supm-\bgamma_{0})\trans\partial_{\bgamma}\fhat_m(\bX_i\supm;\widehat\bgamma_{[t-1]}\supm) \Big\}^2 \\ 
\geq & \frac{1}{2N_m}\sum_{i=1}^{N_m}\Big\{(\widehat\bgamma_{[t]}\supm-\bgamma_{0})\trans f'_m(\bgamma_0\trans\bX\supm_i) \Big(\bX\supm_i - \Ebb\big(\bX \biggiven \bgamma_0\trans\bX_i\supm \big) \Big) \Big\}^2 \\
& - \frac{1}{N_m}\sum_{i=1}^{N_m}\bigg\{(\widehat\bgamma_{[t]}\supm-\bgamma_{0})\trans \Big(\partial_{\bgamma}\fhat_m(\bX_i\supm;\widehat\bgamma_{[t-1]}\supm) -  f'_m(\bgamma_0\trans\bX\supm_i) \Big(\bX\supm_i - \Ebb\big(\bX \biggiven \bgamma_0\trans\bX_i\supm \big) \Big) \Big) \bigg\}^2  \\
:= & \Pi_1 - \Pi_2.
\end{aligned}    
\end{equation}
\end{footnotesize}
{\blue Write $T_i=\bgamma_0\trans\bX_i\supm$ and $R_i=\bX_i\supm-\Ebb(\bX_i\supm\mid T_i)$. The leading term is
\[
\Delta_t\supmtrans
\left\{\frac1{N_m}\sum_{i=1}^{N_m}f'_m(T_i)^2R_iR_i\trans\right\}
\Delta_t\supm.
\]
By Lemma \ref{cor:1}, $\Delta_t\supm/\|\Delta_t\supm\|_2$ lies in the normalized tangent cone and has first coordinate zero. Assumption \ref{ass:4} and restricted covariance concentration give the required lower bound with probability at least $1-\exp(-c_1s\log p)$.}

For $\Pi_2$, we have
\begin{scriptsize}
\begin{equation}\label{eqn:Pi_2}
\begin{aligned}
|\Pi_2| & = \bigg| \frac{1}{N_m}\sum_{i=1}^{N_m}\Big\{(\widehat\bgamma_{[t]}\supm-\bgamma_{0})\trans \Big( \partial_{\bgamma}\fhat_m(\bX_i\supm;\widehat\bgamma_{[t-1]}\supm) -  f'_m(\bgamma_0\trans\bX\supm_i) \Big(\bX\supm_i - \Ebb\big(\bX \biggiven \bgamma_0\trans\bX_i\supm \big) \Big) \Big) \Big\}^2 \bigg|  \\
& \leq 2 \bigg| \frac{1}{N_m}\sum_{i=1}^{N_m}\bigg\{ (\widehat\bgamma_{[t]}\supm-\bgamma_{0})\trans \Big( \partial_{\bgamma}\fhat_m(\bX_i\supm;\widehat\bgamma_{[t-1]}\supm) -  {\blue f'_m({\widehat\bgamma_{[t-1]}\supm}{}\trans\bX\supm_i) \Big(\bX\supm_i-\Ebb\big(\bX \biggiven \widehat\bgamma_{[t-1]}\supmtrans \bX_i\supm \big) \Big)} \Big) \bigg\}^2 \bigg|  \\
& \ \ + 2 \bigg| \frac{1}{N_m}\sum_{i=1}^{N_m}\bigg\{ (\widehat\bgamma_{[t]}\supm-\bgamma_{0})\trans
\Big(  {\blue f'_m(\widehat\bgamma_{[t-1]}\supmtrans\bX\supm_i) \Big(\bX\supm_i-\Ebb\big(\bX \biggiven {\widehat\bgamma_{[t-1]}} \supmtrans \bX_i\supm \big) \Big)} -  f'_m(\bgamma_0\trans\bX\supm_i) \Big(\bX\supm_i - \Ebb\big(\bX \biggiven \bgamma_0\trans\bX_i\supm \big) \Big) \Big) \bigg\}^2 \bigg|  \\
& \lesssim \|\widehat\bgamma_{[t]}\supm-\bgamma_{0}\|^2_2 h_m^2 + \|\widehat\bgamma_{[t]}\supm-\bgamma_{0}\|_2^2\|\widehat\bgamma_{[t-1]}\supm-\bgamma_{0}\|_2^2 \\
& = o(\|\widehat\bgamma_{[t]}\supm-\bgamma_{0}\|_2^2)
\end{aligned}    
\end{equation}
\end{scriptsize}
by an argument similar to that used in the proof of Lemma \ref{lem:A3}. 

Combining \eqref{eqn:re}, the lower bound for $\Pi_1$, and \eqref{eqn:Pi_2}, we have
\begin{equation}\nonumber
\small
 \frac{1}{N_m}\sum_{i=1}^{N_m}\Big\{(\widehat\bgamma_{[t]}\supm-\bgamma_{0})\trans\partial_{\bgamma}\fhat_m(\bX_i\supm;\widehat\bgamma_{[t-1]}\supm) \Big\}^2  
\geq   \kappa \|\widehat\bgamma_{[t]}\supm-\bgamma_{0}\|_2^2
\end{equation}
with probability at least $1-\exp \left(-c_{1} s \log p\right)$  for some positive constant $c_{1}$.

\end{proof}

\end{lemma}

\vspace{-30pt}

\subsection{Auxiliary Lemmas in Section \ref{sec:lem1}}

\begin{lemma}\label{lem:fhat}
Suppose Assumptions \ref{ass:fm}--\ref{ass:x} hold and
$d_0s\log(p\vee N_m)\le N_mh_m^5$ for some constant $d_0$.
There exist constants $c_0,c_1>0$ such that
\[
\max_{1\le i\le N_m}\sup_{\bgamma\in\Gamma}
\left|f_m(\bgamma\trans\bX_i\supm)-
\fhat_m(\bX_i\supm;\bgamma)\right|
\le c_0h_m^2
\]
with probability at least $1-\exp\{-c_1\log(p\vee N_m)\}$.
The set $\Gamma$ is the local $\ell_1$--$\ell_2$ neighborhood in
Assumption~\ref{ass:fm}; uniformity over this set follows from
Assumption~\ref{ass:fm}(c).

\begin{proof}
Recall that the leave-one-out kernel estimator for $f_m(\bgamma\trans\bX_i\supm)$  is
\begin{small}
\begin{align*}
\fhat_m(\bX_i\supm;\bgamma)
=\frac{\sum_{j=1, j\neq i}^{N_m} K_{h_m}(\bgamma\trans\{\bX_j\supm-\bX_i\supm\})S\supm_j}{\sum_{j=1, j\neq i}^{N_m} K_{h_m}(\bgamma\trans\{\bX_j\supm-\bX_i\supm\})} 
= \frac{\sum_{j=1, j\neq i}^{N_m} K_{h_m}(\bgamma\trans\{\bX_j\supm-\bX_i\supm\})(f_m(\bgamma_0\trans\bX_j\supm)+\epsilon\supm_j)}{\sum_{j=1, j\neq i}^{N_m} K_{h_m}(\bgamma\trans\{\bX_j\supm-\bX_i\supm\})}, 
\end{align*}
\end{small}
which gives  
\begin{small}
\begin{equation}\label{eqn:B2}
\begin{aligned}
 & \fhat_m(\bX_i\supm;\bgamma) - f_m(\bgamma\trans\bX_i\supm) \\
 = & \frac{\sum_{j=1, j\neq i}^{N_m} K_{h_m}(\bgamma\trans\{\bX_j\supm-\bX_i\supm\})(f_m(\bgamma\trans\bX_j\supm)+\epsilon\supm_j)}{\sum_{j=1, j\neq i}^{N_m} K_{h_m}(\bgamma\trans\{\bX_j\supm-\bX_i\supm\})} - f_m(\bgamma\trans\bX_i\supm) \\
 = & \frac{\sum_{j=1, j\neq i}^{N_m} K_{h_m}(\bgamma\trans\{\bX_j\supm-\bX_i\supm\})\big(f_m(\bgamma\trans\bX_j\supm) - f_m(\bgamma\trans\bX_i\supm)\big) + \sum_{j=1, j\neq i}^{N_m} K_{h_m}(\bgamma\trans\{\bX_j\supm-\bX_i\supm\})\epsilon\supm_j}
{\sum_{j=1, j\neq i}^{N_m} K_{h_m}(\bgamma\trans\{\bX_j\supm-\bX_i\supm\})} \\
:= & \frac{A_{2}(\bgamma\trans\bX_i\supm \mid \bgamma)+A_{1}(\bgamma\trans\bX_i\supm \mid \bgamma)}{A_{3}(\bgamma\trans\bX_i\supm \mid \bgamma)},
\end{aligned}    
\end{equation}
\end{small}
where we define
\begin{small}
\begin{equation}\nonumber
\begin{aligned}
A_{1}(\bgamma\trans\bX_i\supm \mid \bgamma) & = \frac{1}{N_m-1} \sum_{j=1, j\neq i}^{N_m} K_{h_m}(\bgamma\trans\{\bX_j\supm-\bX_i\supm\})\epsilon\supm_j, \\
A_{2}(\bgamma\trans\bX_i\supm \mid \bgamma) & = \frac{1}{N_m-1} \sum_{j=1, j\neq i}^{N_m} K_{h_m}(\bgamma\trans\{\bX_j\supm-\bX_i\supm\})\big(f_m(\bgamma\trans\bX_j\supm) - f_m(\bgamma\trans\bX_i\supm)\big), \\
A_{3}(\bgamma\trans\bX_i\supm \mid \bgamma) & = \frac{1}{N_m-1} \sum_{j=1, j\neq i}^{N_m} K_{h_m}(\bgamma\trans\{\bX_j\supm-\bX_i\supm\}). \\
\end{aligned}    
\end{equation}
\end{small}
By Lemmas B4, B5 and B6 in \cite{wu2021model}, choosing $ d_0 s \log (p \vee N_m) \leq N_m h_m^5$ for some constant $d_0$, we have
\begin{small}
\begin{equation}\label{eqn:B4}
\begin{aligned}
& \max_{1\leq i \leq N_m} \sup_{\bgamma \in \Gamma} \Big|A_{1}(\bgamma\trans\bX_i\supm \mid \bgamma) \Big| 
\leq c_0 h_m^2 \\ %
& \max_{1\leq i \leq N_m} \sup_{\bgamma \in \Gamma} \Big|A_{2}(\bgamma\trans\bX_i\supm \mid \bgamma) \Big| 
\leq c_0 h_m^2 \\
& \max_{1\leq i \leq N_m} \sup_{\bgamma \in \Gamma} \Big|A_{3}(\bgamma\trans\bX_i\supm \mid \bgamma) - \Ebb A_{3}(\bgamma\trans\bX_i\supm \mid \bgamma) \Big| 
\leq c_0 h_m^2 \\
\end{aligned}    
\end{equation}
\end{small}
with probability exceeding $ 1-\exp(-c_1\log (p \vee N_m))$.
Then for $\Ebb A_{3}(\bgamma\trans\bX_i\supm \mid \bgamma)$, we have
\begin{small}
\begin{equation}\nonumber
\begin{aligned}
\Ebb A_{3}(\bgamma\trans\bX_i\supm \mid \bgamma)
& = \frac{1}{h_m} \int K\Big(\frac{y-\bgamma\trans\bX_i\supm}{h_m}\Big) g_{m,\bgamma}(y) d y \\
&=\int K(-z) g_{m,\bgamma}(\bgamma\trans\bX_i\supm + {h_m} z) \d z\\
&=\int K(z)\Big[g_{m,\bgamma}(\bgamma\trans\bX_i\supm)+{h_m} z g_{m,\bgamma}^{\prime}(\bgamma\trans\bX_i\supm)+\frac{{h_m}^{2} z^{2}}{2} g_{m,\bgamma}^{\prime \prime}(t)\Big] \d z\\
&=g_{m,\bgamma}(\bgamma\trans\bX_i\supm)+\frac{{h_m}^{2}}{2} \int z^{2} K(z) g_{m,\bgamma}^{\prime \prime}(t) d z
\end{aligned}
\end{equation}
\end{small}
where $g_{m,\bgamma}$ is the density function of $\bgamma\trans\bX$. Here the second equality holds by letting $z:=(y - \bgamma\trans\bX_i\supm) / {h_m}$, and the third equality holds by Taylor expansion and $t \in \mathbb{R}$ is between $\bgamma\trans\bX_i\supm$ and $\bgamma\trans\bX_i\supm+h_m z$. Also,  Assumptions \ref{ass:k} and \ref{ass:fm}(b) imply that
$$
\small
\max_i\sup_{\bgamma \in \Gamma}\bigg|\frac{{h_m}^{2}}{2} \int z^{2} K(z) g_{m,\bgamma}^{\prime \prime}({t}) d z\bigg| \leq c_{0} {h_m}^{2},
$$
which gives
\begin{small}
\begin{equation}\label{eqn:B6}
\begin{aligned}
\max_i \sup_{\bgamma \in \Gamma} \Big|\Ebb A_{3}(\bgamma\trans\bX_i\supm \mid \bgamma) -  g_{m,\bgamma}(\bgamma\trans\bX_i\supm) \Big| 
\leq c_0 {h_m}^2. 
\end{aligned}    
\end{equation}
\end{small}
Combining \eqref{eqn:B4}, \eqref{eqn:B6} and Assumption \ref{ass:fm}(b), under which $\min_i \inf_{\bgamma \in \Gamma} g_{m,\bgamma}(\bgamma\trans\bX\supm_i) \geq c$, we achieve that
\begin{small}
\begin{equation}\label{eqn:B7-0}
\begin{aligned}
\max_i \sup_{\bgamma \in \Gamma} \frac{1}{A_{3}(\bgamma\trans\bX_i\supm \mid \bgamma)} \leq c'
\end{aligned}    
\end{equation}
\end{small}
for some constant $c'>0$.

Finally, combining  \eqref{eqn:B2}, \eqref{eqn:B4} and \eqref{eqn:B7-0}, we have
$$\small \max_{1\leq i \leq N_m} \sup_{\bgamma \in \Gamma} \Big|f_m(\bgamma\trans\bX_i\supm)- \fhat_m(\bX_i\supm;\bgamma) \Big| \leq  c_0 {h_m}^2$$
 with probability exceeding $1-\exp(-c_1\log (p\vee N_m))$.

\end{proof}

\end{lemma}

\begin{lemma}\label{lem:partial-fhat}
Suppose Assumptions \ref{ass:fm}--\ref{ass:x} hold and
$d_0s\log(p\vee N_m)\le N_mh_m^5$ for some constant $d_0$.
Let
\[
\mathbb V_c=
\left\{\bv\in\mathbb R^p:
 v_1=0,\ \|\bv\|_2\le1,\
 \|\bv_{S^c}\|_1\le3\|\bv_S\|_1
\right\}.
\]
There exist constants $c_0,c_1>0$ such that
\begin{align*}
&\max_{1\le i\le N_m}\sup_{\bgamma\in\Gamma}
\left\|
\partial_{\bgamma}\fhat_m(\bX_i\supm;\bgamma)
-f_m'(\bgamma\trans\bX_i\supm)
\left\{\bX_i\supm-
\Ebb(\bX\mid\bgamma\trans\bX_i\supm)\right\}
\right\|_\infty
\le c_0h_m,\\
&\max_{1\le i\le N_m}
\sup_{\bgamma\in\Gamma,\,\bv\in\mathbb V_c}
\left|
\left[
\partial_{\bgamma}\fhat_m(\bX_i\supm;\bgamma)
-f_m'(\bgamma\trans\bX_i\supm)
\left\{\bX_i\supm-
\Ebb(\bX\mid\bgamma\trans\bX_i\supm)\right\}
\right]\trans\bv
\right|
\le c_0h_m
\end{align*}
with probability at least $1-\exp\{-c_1\log(p\vee N_m)\}$.

\begin{proof}
Recall that the leave-one-out kernel estimator for $f_m'(\bgamma\trans\bX_i\supm)$  is
 \begin{footnotesize}
 \begin{equation}\nonumber
\begin{aligned}
 \partial_{\bgamma}\fhat_m(\bX_i\supm;\bgamma) 
= & 
\frac{\sum_{j=1, j\neq i}^{N_m} K_{h_m}'(\bgamma\trans\{\bX_j\supm-\bX_i\supm\})\{f_m(\bgamma\trans\bX\supm_j)+\epsilon\supm_{j}\}(\bX_j\supm-\bX_i\supm)}{\sum_{j=1, j\neq i}^{N_m} K_{h_m}(\bgamma\trans\{\bX_j\supm-\bX_i\supm\})}\\
&-\frac{\sum_{j=1, j\neq i}^{N_m} K_{h_m}(\bgamma\trans\{\bX_j\supm-\bX_i\supm\})\{f_m(\bgamma\trans\bX\supm_j)+\epsilon\supm_{j}\} \sum_{j=1, j\neq i}^{N_m} K_{h_m}'(\bgamma\trans\{\bX_j\supm-\bX_i\supm\})(\bX_j\supm-\bX_i\supm)}{\left[\sum_{j=1, j\neq i}^{N_m} K_{h_m}(\bgamma\trans\{\bX_j\supm-\bX_i\supm\})\right]^2} \\
= & 
\frac{\sum_{j=1, j\neq i}^{N_m} K_{h_m}'(\bgamma\trans\{\bX_j\supm-\bX_i\supm\})\epsilon\supm_{j}(\bX_j\supm-\bX_i\supm)}{\sum_{j=1, j\neq i}^{N_m} K_{h_m}(\bgamma\trans\{\bX_j\supm-\bX_i\supm\})} \\
& + 
\frac{\sum_{j=1, j\neq i}^{N_m} K_{h_m}'(\bgamma\trans\{\bX_j\supm-\bX_i\supm\})f_m(\bgamma\trans\bX\supm_j)(\bX_j\supm-\bX_i\supm)}{\sum_{j=1, j\neq i}^{N_m} K_{h_m}(\bgamma\trans\{\bX_j\supm-\bX_i\supm\})}\\
& -  \frac{\sum_{j=1, j\neq i}^{N_m} K_{h_m}(\bgamma\trans\{\bX_j\supm-\bX_i\supm\})\epsilon\supm_j\sum_{j=1, j\neq i}^{N_m} K_{h_m}'(\bgamma\trans\{\bX_j\supm-\bX_i\supm\})(\bX_j\supm-\bX_i\supm)}{\left[\sum_{j=1, j\neq i}^{N_m} K_{h_m}(\bgamma\trans\{\bX_j\supm-\bX_i\supm\})\right]^2}  \\
& -  \frac{\sum_{j=1, j\neq i}^{N_m} K_{h_m}(\bgamma\trans\{\bX_j\supm-\bX_i\supm\})f_m(\bgamma\trans\bX\supm_j)\sum_{j=1, j\neq i}^{N_m} K_{h_m}'(\bgamma\trans\{\bX_j\supm-\bX_i\supm\})(\bX_j\supm-\bX_i\supm)}{\left[\sum_{j=1, j\neq i}^{N_m} K_{h_m}(\bgamma\trans\{\bX_j\supm-\bX_i\supm\})\right]^2} \\
:= & 
\frac{D_{i2}-B_{i2}}{D_{i1}}
 + 
\frac{D_{i3}-B_{i3}}{D_{i1}}
 -  \frac{D_{i4}-B_{i4}}{D_{i1}^2} 
 -  \frac{D_{i5}-B_{i5}}{D_{i1}^2} 
\end{aligned}    
\end{equation}
\end{footnotesize}
where
\begin{footnotesize}
\begin{equation}\nonumber
\begin{aligned}
& D_{i1}=\frac{1}{N_m-1} \sum_{j=1, j\neq i}^{N_m} K_{h_m}(\bgamma\trans\{\bX_j\supm-\bX_i\supm\}), \\
& D_{i2}=\frac{1}{N_m-1} \sum_{j=1, j\neq i}^{N_m} K_{h_m}'(\bgamma\trans\{\bX_j\supm-\bX_i\supm\})\epsilon\supm_{j}\bX_j\supm,  \\
& 
B_{i2}=\frac{1}{N_m-1} \sum_{j=1, j\neq i}^{N_m} K_{h_m}'(\bgamma\trans\{\bX_j\supm-\bX_i\supm\})\epsilon\supm_{j}\bX_i\supm, \\
& D_{i3}=\frac{1}{N_m-1}\sum_{j=1, j\neq i}^{N_m} K_{h_m}'(\bgamma\trans\{\bX_j\supm-\bX_i\supm\})f_m(\bgamma\trans\bX\supm_j)\bX_j\supm,  \\
& 
B_{i3}=\frac{1}{N_m-1}\sum_{j=1, j\neq i}^{N_m} K_{h_m}'(\bgamma\trans\{\bX_j\supm-\bX_i\supm\})f_m(\bgamma\trans\bX\supm_j)\bX_i\supm,
\end{aligned}    
\end{equation}
\end{footnotesize}
\begin{footnotesize}
\begin{equation}\nonumber
\begin{aligned}
& D_{i4}= \frac{1}{N_m-1}\sum_{j=1, j\neq i}^{N_m} K_{h_m}(\bgamma\trans\{\bX_j\supm-\bX_i\supm\})\epsilon\supm_j 
\cdot \frac{1}{N_m-1}\sum_{j=1, j\neq i}^{N_m} K_{h_m}'(\bgamma\trans\{\bX_j\supm-\bX_i\supm\})\bX_j\supm,  \\
& B_{i4}=\frac{1}{N_m-1}\sum_{j=1, j\neq i}^{N_m} K_{h_m}(\bgamma\trans\{\bX_j\supm-\bX_i\supm\})\epsilon\supm_j \cdot  \frac{1}{N_m-1}\sum_{j=1, j\neq i}^{N_m} K_{h_m}'(\bgamma\trans\{\bX_j\supm-\bX_i\supm\})\bX_i\supm, \\
& D_{i5}=\frac{1}{N_m-1}\sum_{j=1, j\neq i}^{N_m} K_{h_m}(\bgamma\trans\{\bX_j\supm-\bX_i\supm\})f_m(\bgamma\trans\bX\supm_j)
\cdot \frac{1}{N_m-1}\sum_{j=1, j\neq i}^{N_m} K_{h_m}'(\bgamma\trans\{\bX_j\supm-\bX_i\supm\})\bX_j\supm := D_{i51}D_{i52}, \\
& B_{i5}=\frac{1}{N_m-1}\sum_{j=1, j\neq i}^{N_m} K_{h_m}(\bgamma\trans\{\bX_j\supm-\bX_i\supm\})f_m(\bgamma\trans\bX\supm_j) 
\cdot \frac{1}{N_m-1}\sum_{j=1, j\neq i}^{N_m} K_{h_m}'(\bgamma\trans\{\bX_j\supm-\bX_i\supm\})\bX_i\supm:= B_{i51}B_{i52}.
\end{aligned}    
\end{equation}
\end{footnotesize}
We analyze the term $D_{i52}=\frac{1}{N_m-1}\sum_{j=1, j\neq i}^{N_m} K_{h_m}'(\bgamma\trans\{\bX_j\supm-\bX_i\supm\})\bX_j\supm$  in Lemma \ref{lem:B3} that
\begin{small}\begin{equation}\nonumber
    \max_{1\leq i \leq N_m} \sup_{\bgamma \in \Gamma} \Big\| D_{i52} +
    \Big\{ \Ebb^{(1)}\big(\bX \biggiven \bgamma\trans\bX=\bgamma\trans\bX_i\supm \big) g_{m,\bgamma}(\bgamma\trans\bX_i\supm) + \Ebb\big(\bX \biggiven \bgamma\trans\bX=\bgamma\trans\bX_i\supm \big) g'_{m,\bgamma}(\bgamma\trans\bX_i\supm) \Big\} \Big\|_\infty 
    \leq  c_0 h_m.
\end{equation}
\end{small}
Other terms can be analyzed similarly to obtain
\begin{footnotesize}
\begin{equation}\nonumber
\begin{aligned}
& \max_i\sup_{\bgamma \in \Gamma}\Big|D_{i1}-g_{m,\bgamma}(\bgamma\trans\bX\supm_i)\Big| \leq c_0 h_m
\\
& 
\max_i\sup_{\bgamma \in \Gamma} \Big\|D_{i2}\Big\|_{\infty} \leq c_0 h_m
\\
& 
\max_i\sup_{\bgamma \in \Gamma} \Big\|B_{i2}\Big\|_{\infty} \leq c_0 h_m
\\
& 
\max_i\sup_{\bgamma \in \Gamma} \Big\| D_{i3}
+ 
\Big\{  \Ebb\big(\bX \biggiven \bgamma\trans\bX_i\supm \big) g_{m,\bgamma}(\bgamma\trans\bX_i\supm) f'_m(\bgamma\trans\bX_i\supm) 
 + \Ebb\big(\bX \biggiven \bgamma\trans\bX_i\supm \big) g'_{m,\bgamma}(\bgamma\trans\bX_i\supm)f_m(\bgamma\trans\bX_i\supm) \\
& \ \ \ \quad \quad \quad \quad \quad + \Ebb^{(1)}\big(\bX \biggiven \bgamma\trans\bX_i\supm \big) g_{m,\bgamma}(\bgamma\trans\bX_i\supm) f_m(\bgamma\trans\bX_i\supm)\Big\} \Big\|_\infty 
\leq c_0 h_m
\\
& 
\max_i\sup_{\bgamma \in \Gamma} \Big\|B_{i3}
+
\Big(f_m(\bgamma\trans\bX\supm_i)g'_{m,\bgamma}(\bgamma\trans\bX\supm_i)+f'(\bgamma\trans\bX\supm_i)g_{m,\bgamma}(\bgamma\trans\bX\supm_i)\Big) \bX\supm_i \Big\|_\infty 
\leq c_0 h_m 
\\
& 
\max_i\sup_{\bgamma \in \Gamma} \Big\| D_{i4} \Big\|_\infty \leq c_0 h_m
\\
& 
\max_i\sup_{\bgamma \in \Gamma} \Big\| B_{i4} \Big\|_\infty \leq c_0 h_m
\\
& 
\max_i \sup_{\bgamma \in \Gamma} \Big\|D_{i5} + f_m(\bgamma\trans\bX\supm_i)g_{m,\bgamma}(\bgamma\trans\bX\supm_i) \Big\{ \Ebb^{(1)}\big(\bgamma\trans\bX_i\supm \big) g_{m,\bgamma}(\bgamma\trans\bX_i\supm) + \Ebb\big(\bgamma\trans\bX_i\supm \big) g'_{m,\bgamma}(\bgamma\trans\bX_i\supm) \Big\} \Big\|_\infty \leq c_0 h_m
\\
& 
\max_i\sup_{\bgamma \in \Gamma} \big\| B_{i5} + f_m(\bgamma\trans\bX\supm_i)g_{m,\bgamma}(\bgamma\trans\bX\supm_i)g'_{m,\bgamma}(\bgamma\trans\bX\supm_i)\bX\supm_i \big\|_\infty
\leq c_0 h_m
\end{aligned}
\end{equation}
\end{footnotesize}

%

Thus, by combining the above bounds, we have
\begin{small}\begin{equation}\nonumber
\max_{1\leq i \leq N_m} \sup_{\bgamma \in \Gamma} \Big\| \partial_{\bgamma}\fhat_m(\bX_i\supm;\bgamma)  - f'_m(\bgamma\trans\bX\supm_i) \big(\bX\supm_i -\Ebb (\bX \biggiven \bgamma\trans\bX_i\supm )\big) \Big\|_\infty 
\leq  c_0 h_m.
\end{equation}
\end{small}
Similarly, we can show
$$ \small
\max_{1\leq i \leq N_m} \sup_{\bgamma \in \Gamma, \boldsymbol{v} \in \mathbb V_c} \Big| \Big\{\partial_{\bgamma}\fhat_m(\bX_i\supm;\bgamma)  - f'_m(\bgamma\trans\bX\supm_i) \big(\bX\supm_i -\Ebb (\bX \biggiven \bgamma\trans\bX_i\supm )\big) \Big\}\trans \boldsymbol{v} \Big|
\leq  c_0 h_m.
$$
\end{proof}
\end{lemma}

\begin{lemma}\label{lem:B3}
Let
\[
D_{i52}=\frac1{N_m-1}\sum_{j\ne i}
K_{h_m}'\{\bgamma\trans(\bX_j\supm-\bX_i\supm)\}\bX_j\supm.
\]
Under the assumptions of Lemma~\ref{lem:partial-fhat},
\begin{align*}
&\max_i\sup_{\bgamma\in\Gamma,\,\bv\in\mathbb V_c}
\left|
\left[
D_{i52}
+\Ebb^{(1)}(\bX\mid\bgamma\trans\bX_i\supm)
 g_{m,\bgamma}(\bgamma\trans\bX_i\supm)
+\Ebb(\bX\mid\bgamma\trans\bX_i\supm)
 g'_{m,\bgamma}(\bgamma\trans\bX_i\supm)
\right]\trans\bv
\right|
\le c_0h_m,\\
&\max_i\sup_{\bgamma\in\Gamma}
\left\|
D_{i52}
+\Ebb^{(1)}(\bX\mid\bgamma\trans\bX_i\supm)
 g_{m,\bgamma}(\bgamma\trans\bX_i\supm)
+\Ebb(\bX\mid\bgamma\trans\bX_i\supm)
 g'_{m,\bgamma}(\bgamma\trans\bX_i\supm)
\right\|_\infty
\le c_0h_m
\end{align*}
with probability at least $1-\exp\{-c_1\log(p\vee N_m)\}$.

\begin{proof}   
We prove the second displayed bound below; the proof of the first is analogous.
We first calculate $\Ebb D_{i52}\trans \bv$ and then bound $\big(D_{i52}-\Ebb D_{i52}\big)\trans \bv$.
We have

\begin{scriptsize}
\begin{equation}\label{eqn:B7}
\begin{aligned}
\Ebb D_{i52} 
& =\Ebb\big(K_{h_m}'(\bgamma\trans\{\bX_j\supm-\bX_i\supm\})\bX_j\supm\big)\\ &=\Ebb_{\bgamma\trans\bX_j\supm} \left\{K_{h_m}'(\bgamma\trans\{\bX_j\supm-\bX_i\supm\}) \Ebb\big(\bX \biggiven  \bgamma\trans\bX_j\supm \big)\right\} \\
&= \int  K_{h_m}'(\bgamma\trans\{\bX_j\supm-\bX_i\supm\})  \Ebb\big(\bX \biggiven  \bgamma\trans\bX_j\supm \big) g_{m,\bgamma}(\bgamma\trans\bX_j\supm) \mathrm{d}\bgamma\trans\bX_j\supm  \\
&= \int  \frac{1}{h_m} K'(z)  \Ebb\big(\bX \biggiven \bgamma\trans\bX_i\supm +hz \big) g_{m,\bgamma}(\bgamma\trans\bX_i\supm+hz) \mathrm{d} z \\
&=\int  \frac{1}{h_m} K'(z) \bigg\{ \Ebb\big(\bX \biggiven \bgamma\trans\bX_i\supm \big)+ h_m z \Ebb^{(1)}\big(\bX \biggiven \bgamma\trans\bX_i\supm  \big)  + \frac{h_m^{2} z^{2}}{2} \Ebb^{(2)}\big(\bX \biggiven t' \big) \bigg\}
\cdot \bigg\{g_{m,\bgamma}(\bgamma\trans\bX_i\supm) + h_m z g'_{m,\bgamma}(\bgamma\trans\bX_i\supm) + \frac{h_m^{2} z^{2}}{2} g''_{m,\bgamma}(t'') \bigg\} \d z
\\ 
& = \int  \frac{1}{h_m} K'(z) \bigg\{ \Ebb\big(\bX \biggiven \bgamma\trans\bX_i\supm \big) g_{m,\bgamma}(\bgamma\trans\bX_i\supm) \\
& \quad\quad\quad\quad\quad + 
h_m z \Big[\Ebb^{(1)}\big(\bX \biggiven \bgamma\trans\bX_i\supm \big) g_{m,\bgamma}(\bgamma\trans\bX_i\supm) + \Ebb\big(\bX \biggiven \bgamma\trans\bX_i\supm \big) g'_{m,\bgamma}(\bgamma\trans\bX_i\supm) \big)\Big] 
\\
& \quad\quad\quad\quad\quad  + h_m^2  z^2 
\Big[
\frac{1}{2}\Ebb\big(\bX \biggiven \bgamma\trans\bX_i\supm \big)g''_{m,\bgamma}(t'') + \Ebb^{(1)}\big(\bX \biggiven \bgamma\trans\bX_i\supm  \big)g'_{m,\bgamma}(\bgamma\trans\bX_i\supm)
+ \frac{1}{2}\Ebb^{(2)}\big(\bX \biggiven t' \big)g_{m,\bgamma}(\bgamma\trans\bX_i\supm)
\Big] 
\\
& \quad\quad\quad\quad\quad  + \frac{1}{2} h_m^3  z^3 
\Big[
\Ebb^{(1)}\big(\bX \biggiven \bgamma\trans\bX_i\supm \big)g''_{m,\bgamma}(t'')
+ \Ebb^{(2)}\big(\bX \biggiven t' \big)g'_{m,\bgamma}(\bgamma\trans\bX_i\supm)
\Big] 
 \\
& \quad\quad\quad\quad\quad  + \frac{1}{4} h_m^4  z^4 
\Big[
\Ebb^{(2)}\big(\bX \biggiven t' \big)g''_{m,\bgamma}(t'') 
\Big] 
\bigg\} \d z
\end{aligned}
\end{equation}
\end{scriptsize}
The fourth equality holds by a change of variable $z:= \big(\bgamma\trans\bX_j\supm-\bgamma\trans\bX_i\supm\big)/h_m$. The fifth equality holds by Taylor expansion where $ \Ebb^{(1)}\big(\bX \biggiven t)$  and $ \Ebb^{(2)}\big(\bX \biggiven t)$ are the first and second derivatives of $\Ebb\big(\bX \biggiven t)$ with respect to $t$, and $t'$ and $t''$ are both between $\bgamma\trans\bX_i\supm$ and $\bgamma\trans\bX_i\supm+ h_m z$. The sixth equality follows by re-organizing terms. 
Then Assumptions \ref{ass:k}, \ref{ass:fm}, and \ref{ass:5} 
imply that 
\begin{small}
\begin{equation}\label{eqn:EA}
\Big\| \Ebb D_{i52}\trans \boldsymbol{v} + \Big\{ \Ebb^{(1)}\big(\bX \biggiven \bgamma\trans\bX_i\supm \big) g_{m,\bgamma}(\bgamma\trans\bX_i\supm) + \Ebb\big(\bX \biggiven \bgamma\trans\bX_i\supm \big) g'_{m,\bgamma}(\bgamma\trans\bX_i\supm) \Big\} \trans \boldsymbol{v}\Big\|_\infty \leq c_0 h_m .
\end{equation}
\end{small}

Next we analyze $\big(D_{i52}-\Ebb D_{i52}\big)\trans \bv$.  Following similar argument in \eqref{eqn:B7}, we have

\begin{footnotesize}
\begin{equation}\nonumber
\begin{aligned}
& \ \ \ \, \Ebb ( D_{i52}\trans  \boldsymbol{v} )^2 
\\
& = \Ebb \big( K_{h_m}'^2(\bgamma\trans \bX_j\supm-\bgamma\trans\bX_i\supm) ({\bX_j\supm} \trans \boldsymbol{v} )^2 \big) 
\\
& = \Ebb_{\bgamma\trans\bX_j\supm} \left\{K_{h_m}'^2(\bgamma\trans \bX_j\supm-\bgamma\trans\bX_i\supm ) \Ebb\big((\bX \trans \boldsymbol{v}  )^2  \biggiven  \bgamma\trans\bX_j\supm \big)\right\} \\
& = \int  K_{h_m}'^2(\bgamma\trans \bX_j\supm-\bgamma\trans\bX_i\supm )  \Ebb\big((\bX \trans \boldsymbol{v}  )^2  \biggiven  \bgamma\trans\bX_j\supm \big) g_{m,\bgamma}(\bgamma\trans\bX_j\supm) \mathrm{d}\bgamma\trans\bX_j\supm  \\
& = \frac{1}{h_m^3} \int K'^2(z)  \Ebb\big((\bX \trans \boldsymbol{v} )^2  \biggiven \bgamma\trans\bX_i\supm +hz \big) g_{m,\bgamma}(\bgamma\trans\bX_i\supm+hz) \mathrm{d} z 
\\
& = \frac{1}{h_m^3} \int  K'^2(z) \bigg\{ \Ebb\big((\bX \trans \boldsymbol{v}  )^2 \biggiven \bgamma\trans\bX_i\supm \big)+ h_m z \Ebb^{(1)}\big((\bX \trans \boldsymbol{v}  )^2 \biggiven t' \big) \bigg\}
\cdot \bigg\{g_{m,\bgamma}(\bgamma\trans\bX_i\supm) + h_m z g'_{m,\bgamma}(t'')  \bigg\} \d z 
\\ 
& = \frac{1}{h_m^3} \int  K'^2(z) \bigg\{\Big[
\Ebb\big((\bX \trans \boldsymbol{v}  )^2 \biggiven \bgamma\trans\bX_i\supm \big) g_{m,\bgamma}(\bgamma\trans\bX_i\supm)\Big] \\
& \quad\quad\quad + h_m z \Big[\Ebb^{(1)}\big((\bX \trans \boldsymbol{v} )^2 \biggiven t' \big) g_{m,\bgamma}(\bgamma\trans\bX_i\supm) + \Ebb\big((\bX \trans \boldsymbol{v}  )^2 \biggiven \bgamma\trans\bX_i\supm \big) g'_{m,\bgamma}(t'' \big)\Big] \\
& \quad\quad\quad + h_m^2  z^2 \Big[  \Ebb^{(1)}\big((\bX \trans \boldsymbol{v}  )^2 \biggiven t')g'_{m,\bgamma}(t'') \Big] 
\bigg\} \d z, 
\end{aligned}    
\end{equation}
\end{footnotesize}
which gives 
$\Ebb ( D_{i52}\trans  \boldsymbol{v} )^2 \leq c \frac{1}{h_m^3}\|\boldsymbol{v}\|_2^2 $
following similar arguments as \eqref{eqn:EA}.
Then by similar argument in the proof of Lemma A7 in \cite{wu2021model} and {Hoeffding inequality of sub-Gaussian}, we have
\begin{small}
\begin{equation}\nonumber
\Pbb \Big( \big|\big(D_{i52} -\Ebb D_{i52} \big)\trans \boldsymbol{v}\big| \geq t \Big) \leq 2 \exp\Big(-\frac{(N_m-1)t^2}{2c^2\|\boldsymbol{v}\|_2^2/h_m^3} \Big),
\end{equation}
\end{small}
which further gives
\begin{equation}\label{eqn:B9}
\small
\big|\big(D_{i52} -\Ebb D_{i52} \big)\trans \boldsymbol{v}\big| \leq c_0 h_m 
\end{equation}
with probability exceeding $ 1-\exp(-c_2 \log (p \vee N_m))$ by taking $t=c_0 h_m$ and applying the assumption  $d_0s\log (p\vee N_m) \leq N_m h_m^5$.
Combining \eqref{eqn:EA} and \eqref{eqn:B9} together, we achieve
\begin{small}
\begin{equation}\nonumber
    \Big| D_{i52}\trans \boldsymbol{v} + \Big\{ \Ebb^{(1)}\big(\bX \biggiven \bgamma\trans\bX_i\supm \big) g_{m,\bgamma}(\bgamma\trans\bX_i\supm) + \Ebb\big(\bX \biggiven \bgamma\trans\bX_i\supm \big) g'_{m,\bgamma}(\bgamma\trans\bX_i\supm) \Big\} \trans \boldsymbol{v} \Big|
    \leq c_0 h_m.
\end{equation}
\end{small}

To obtain a uniform bound for $\bgamma \in \Gamma$ and $\bv \in \mathbb V_c$, we cover $\Gamma$ and $\mathbb V_c$ by $L_2$-balls with radius $\delta$. Following similar argument in the proof of Lemma A7 in \cite{wu2021model} (Page 102), we have
\begin{small}
\begin{equation}\nonumber
\sup_{\bgamma \in \Gamma, \bv \in \mathbb V_c} \Big| \Big\{ D_{i52} 
+ \Big( \Ebb^{(1)}\big(\bX \biggiven \bgamma\trans\bX_i\supm \big) g_{m,\bgamma}(\bgamma\trans\bX_i\supm) + \Ebb\big(\bX \biggiven \bgamma\trans\bX_i\supm \big) g'_{m,\bgamma}(\bgamma\trans\bX_i\supm) \Big) \Big\}\trans \bv \Big|
\leq  c_0 h_m
\end{equation}
\end{small}
with probability exceeding $1-\exp(-d_1 \log (p \vee N_m))$ for some constant $d_1$.
Finally to obtain a uniform bound for $1 \leq i \leq N_m$, we have 
\begin{footnotesize}
\begin{equation}\nonumber
\begin{aligned}
& \, \Pbb\Big( \max_{1 \leq i \leq N_m}\sup_{\bgamma \in \Gamma, \bv \in \mathbb V_c } \Big| \Big\{ D_{i52} + \Big( \Ebb^{(1)}\big(\bX \biggiven \bgamma\trans\bX_i\supm \big) g_{m,\bgamma}(\bgamma\trans\bX_i\supm) + \Ebb\big(\bX \biggiven \bgamma\trans\bX_i\supm \big) g'_{m,\bgamma}(\bgamma\trans\bX_i\supm) \Big) \Big\}\trans \bv \Big|
\geq   c_0 h_m \Big) \\
\leq & \, \sum_{i=1}^{N_m} \Pbb\Big( \sup_{\bgamma \in \Gamma, \bv \in \mathbb V_c } \Big| \Big\{ D_{i52} + \Big( \Ebb^{(1)}\big(\bX \biggiven \bgamma\trans\bX_i\supm \big) g_{m,\bgamma}(\bgamma\trans\bX_i\supm) + \Ebb\big(\bX \biggiven \bgamma\trans\bX_i\supm \big) g'_{m,\bgamma}(\bgamma\trans\bX_i\supm) \Big) \Big\}\trans \bv \Big| \geq   c_0 h_m \Big) \\
\leq & \, N_m \exp(-d_1 \log (p \vee N_m))
\leq \exp(-c_1 \log (p \vee N_m))
\end{aligned}    
\end{equation}
\end{footnotesize}
for some constant $c_1$.

Similarly, we can also prove \begin{equation}\nonumber
\small
    \max_{1 \leq i \leq N_m} \sup_{\bgamma \in \Gamma} \Big\| D_{i52} + \Big\{ \Ebb^{(1)}\big(\bX \biggiven \bgamma\trans\bX_i\supm \big) g_{m,\bgamma}(\bgamma\trans\bX_i\supm) + \Ebb\big(\bX \biggiven \bgamma\trans\bX_i\supm \big) g'_{m,\bgamma}(\bgamma\trans\bX_i\supm) \Big\} \Big\|_\infty 
    \leq  c_0 h_m
\end{equation}
with probability exceeding $ 1-\exp(-c_1 \log (p \vee N_m))$.

\end{proof}

\end{lemma}




{
\section{Justification of Assumption \ref{ass:4}}\label{sec:ass}

We verify Assumption \ref{ass:4} under a Gaussian design. Suppose that $\bX\sim N(0,I_p)$. Then
\[
\operatorname{Cov}(\bX\mid\bgamma_0\trans\bX)
=I_p-\frac{\bgamma_0\bgamma_0\trans}{\|\bgamma_0\|_2^2}.
\]
For every $\bv\in\mathbb V_0$, because $v_1=0$, $\|\bv\|_2=1$, and $\gamma_{01}=1$, define
\[
a_m^2=\mathbb E\left[\{f'_m(\bgamma_0\trans\bX)\}^2\right].
\]
Assumption \ref{ass:fm}(a) requires $a_m^2>0$. Since the Gaussian conditional covariance is constant in the index,
{\blue
\begin{small}
\begin{equation}\label{eqn:eigen-value-lb}
\begin{aligned}
&\bv\trans\Ebb\left[\{f'_m(\bgamma_0\trans\bX)\}^2
\operatorname{Cov}(\bX\mid\bgamma_0\trans\bX)\right]\bv\\
&\qquad=a_m^2\left\{1-\frac{(\bv\trans\bgamma_0)^2}{\|\bgamma_0\|_2^2}\right\}
\ge \frac{a_m^2}{\|\bgamma_0\|_2^2}
\ge \frac{a_m^2}{\kappa_\gamma^2}.
\end{aligned}
\end{equation}
\end{small}
Thus the derivative-weighted restricted lower bound holds on the normalized tangent space with $c=\min_m a_m^2\kappa_\gamma^{-2}>0$. The corresponding upper bound is at most $\max_m a_m^2$. }

Furthermore, we evaluate the squared maximum eigenvalues:
\begin{footnotesize}
\begin{equation}
\begin{aligned}
&\sup_{\bgamma \in \Gamma} \frac{1}{N_m} \sum_{i=1}^{N_m}\Big[\lambda_{\max }\big(\Ebb(\bX \bX\trans \mid \bX_i\trans \bgamma)\big)\Big]^2 
 = \sup_{\bgamma \in \Gamma} \frac{1}{N_m} \sum_{i=1}^{N_m} \sup_{\|\bv\|_2=1}\Big[ \Ebb((\bX\trans \bv)^2   \mid \bX_i\trans \bgamma) \Big]^2 
\\
= & \sup_{\bgamma \in \Gamma} \frac{1}{N_m} \sum_{i=1}^{N_m} \sup_{\|\bv\|_2=1}\Big[ \frac{(\bX_i\trans\bgamma)^2(\bv\trans\bgamma)^2}{\|\bgamma\|_2^4} + \|\bv\|_2^2 - \frac{(\bv\trans\bgamma)^2}{\|\bgamma\|_2^2}\Big]^2
 \leq \sup_{\bgamma \in \Gamma} \frac{1}{N_m} \sum_{i=1}^{N_m} \sup_{\|\bv\|_2=1}\Big[ \frac{2(\bX_i\trans\bgamma)^4(\bv\trans\bgamma)^4}{\|\bgamma\|_2^8}  + 2\Big]
\\
 \leq & \sup_{\bgamma \in \Gamma} \frac{1}{N_m} \sum_{i=1}^{N_m} \Big[ \frac{2(\bX_i\trans\bgamma)^4}{\|\bgamma\|_2^4}  + 2\Big] \leq \sup_{\bgamma \in \Gamma} \Big[ \frac{3\E(\bX_i\trans\bgamma)^4}{\|\bgamma\|_2^4}  + 2\Big] = 11.
\end{aligned}
\end{equation}
\end{footnotesize}
The last inequality follows from Lemma B3 in \cite{wu2021model}.

Lastly,
\vspace{-10pt}
\begin{small}
\begin{equation}
\begin{aligned}
& \max_{1 \leq i \leq N_m} \sup_{\bgamma \in \Gamma} \lambda_{\max }\Big(\Ebb(\bX \bX\trans \mid \bX_i\trans \bgamma)\Big)
 = \max_{1 \leq i \leq N_m} \sup_{\bgamma \in \Gamma} \sup_{\|\bv\|_2=1} \Ebb((\bX\trans \bv)^2   \mid \bX_i\trans \bgamma)  
\\
 = & \max_{1 \leq i \leq N_m} \sup _{\bgamma \in \Gamma} \sup_{\|\bv\|_2=1}\Big[ \frac{(\bX_i\trans\bgamma)^2(\bv\trans\bgamma)^2}{\|\bgamma\|_2^4} + \|\bv\|_2^2 - \frac{(\bv\trans\bgamma)^2}{\|\bgamma\|_2^2}\Big]
 \leq \max_{1 \leq i \leq N_m} \sup_{\bgamma \in \Gamma} \Big[ \frac{(\bX_i\trans\bgamma)^2}{\|\bgamma\|_2^2} + 1 \Big] \leq C \log (p \vee N_m),
\end{aligned}
\end{equation}
\end{small}
with high probability.
}

\section{Additional implementation details and results}\label{sec:app:num:detail}

\subsection{Additional details of the method}\label{sec:app:method}

First, we suggest and comment on how to specify the weighting function $\phi_m$ for different types of $S$, including binary, count, and continuous surrogates. In general, our strategy is to estimate ${\rm Var}(S\mid\bX)$ and use inverse-variance weighting. 
\begin{itemize}
    \item[(a)] Take $\widehat\phi_m(\bX;\bgamma)=\fhat_m^{-1}(\bX;\bgamma)$ for Poisson (count) $S$ satisfying ${\rm Var}(S\mid\bX)=\Ebb(S\mid\bX)$;
    
	\item[(b)] Take $\widehat\phi_m(\bX;\bgamma)=\left[\fhat_m(\bX;\bgamma)\{1-\fhat_m(\bX;\bgamma)\}\right]^{-1}$ for Bernoulli (binary) $S$ satisfying ${\rm Var}(S\mid\bX)=\Ebb(S\mid\bX)[1-\Ebb(S\mid\bX)]$;
	
	\item[(c)] Take $\widehat\phi_m(\bX;\bgamma)=1$ for Gaussian (continuous) $S$ satisfying that ${\rm Var}(S\mid\bX)$ is some constant. For $S$ believed to have heteroscedastic ${\rm Var}(S\mid\bX)$, we could use kernel smoothing on $\{S-\fhat_m(\bX;\bgamma)\}^2$ against $\bX\trans\bgamma$ to estimate ${\rm Var}(S\mid\bX)$, and take $\widehat\phi_m$ as the inverse of this estimator.
 
\end{itemize}

Second, we present the specific form of $\partial_{\bgamma}\fhat_m(\bX_i\supm;\bgamma)$ used to construct the main Step II introduced in Section \ref{sec:method:main} 
\begin{footnotesize}
\begin{equation}
\begin{aligned}
\partial_{\bgamma}\fhat_m(\bX_i\supm;\bgamma) 
= & 
\frac{\sum_{j=1, j\neq i}^{N_m} K_{h_m}'(\bgamma\trans\{\bX_j\supm-\bX_i\supm\})S\supm_j(\bX_j\supm-\bX_i\supm)}{\sum_{j=1, j\neq i}^{N_m} K_{h_m}(\bgamma\trans\{\bX_j\supm-\bX_i\supm\})}\\
&-\frac{\sum_{j=1, j\neq i}^{N_m} K_{h_m}(\bgamma\trans\{\bX_j\supm-\bX_i\supm\})S\supm_j \sum_{j=1, j\neq i}^{N_m} K_{h_m}'(\bgamma\trans\{\bX_j\supm-\bX_i\supm\})(\bX_j\supm-\bX_i\supm)}{\left[\sum_{j=1, j\neq i}^{N_m} K_{h_m}(\bgamma\trans\{\bX_j\supm-\bX_i\supm\})\right]^2}. 
\end{aligned}   
\label{equ:form:pf}
\end{equation}
\end{footnotesize}

\subsection{Tuning strategies}\label{sec:app:tune:method}
{

\begin{itemize}

\item For SL, the regularization parameter $\lambda\subsup$ is tuned on the labeled data set $\Lscr\supone$ using $10$-fold cross-validation (CV), specifically using the default \texttt{cv.glmnet} function in R-package \texttt{glmnet}. The candidate range for $\lambda\subsup$ is set as the default range of \texttt{cv.glmnet}.

\item For SASH and IPD, we select $\lambda\supm_t$ at each iteration $t$ by minimizing the BIC:
\[ \small
{\rm BIC}\supm_t(\lambda)=\frac{Q\supm(\widehat\bgamma\supm_{[t]}(\lambda);\widehat\bgamma_{[t-1]}\supm)}{\{\widehat{\sigma}_{[t-1]}\supm\}^2} +\df\{\widehat\bgamma_{[t]}\supm(\lambda)\}\log(N_m),
\]
with respect to $\lambda$, where $\{\widehat\sigma_{[t-1]}\supm\}^2=L\supm(\widehat\bgamma_{[t-1]}\supm)$ and $\widehat\bgamma\supm_{[t]}(\lambda)$ is the estimator obtained by (\ref{eqn:step2}) with the penalty parameter $\lambda\supm_t$ set as $\lambda$. The candidate range for $\lambda\supm_t$ is $0.005\times \{1,2,\ldots,600\}\times(\log p/N_m)^{1/2}$. As suggested in Section \ref{sec:tuning}, we fix $h_m= \{{\log (p \vee N_m)}/{N_m} \}^{1/5}$, where the unknown $s^{1/5}$ is set directly as $1$. Noting that $30^{1/5}<2$, this choice $s^{1/5}=1$ is reasonable. As described in Section \ref{sec:tuning}, $\lambda$ and $\lambda^\dagger$ used in Step III are also selected using the BICs. Their candidate range is set as $0.003\times \{1,2,\ldots,2000\}\times(\log p/N)^{1/2}$. For IPD, the tuning strategies and candidate sets are the same as those of SASH, except for the BIC in its Step III, which is naturally set as the IPD loss of the SIM without the privacy constraint.

\item For pLasso and PASS, the parameters are tuned in the same way as \cite{zhang2022prior}. For the adaptive Lasso version of L1LS implemented in these two methods, they use the BIC estimated by 10-fold CV to select the penalty coefficient, with the candidate range set $[\zeta_m/N_m,1.5 \zeta_m]$, where $\zeta_m$ is the maximum absolute correlation with $S\supm$ among all predictors. For the penalty parameters $\lambda_1$ and $\lambda_2$ used in the prior-adaptive regression of PASS \citep{zhang2022prior}, they use 10-fold CV with the range of $\lambda_1$ set as $[\nu/n,1.5 \nu]$, where $n$ is the labeled size and $\nu$ is the correlation between $S\supone$ and $Y\supone$, and the candidate range of $\kappa=\lambda_1/\lambda_2$ set as \{1,2,\ldots,8\}.

\item For SS-uLasso, the $\ell_1$-penalty parameter of uLasso is selected via CV using \texttt{cv.glmnet} in R-package \texttt{glmnet}, with candidate ranges determined by its default settings.

\item  For SS-RMRCE, we use the R-package \texttt{RMRCE}, which performs cross-validation to select the optimal tuning parameters for RMRCE. For the candidate sets of the tuning parameters in the function \texttt{RMRCE\_cv}, their $\ell_1$-penalty coefficient $\lambda$ is chosen from $\{0.001, 0.01, 0.1\}$ and the smoothness parameter $\alpha$ is chosen from $\{1, 2, 3, 4, 5\}$. This choice is the same as their running example provided in \url{https://rdrr.io/github/zji90/RMRCE/}.

\end{itemize}

}

\subsection{Simulation on SASH+}\label{sec:app:sashplus}

See Figure \ref{fig:error:sash} for the results.

\begin{figure}[htb!]
\centering
\begin{tabular}{cccc}
{\small (A) Weak Surrogates - $\bbeta$} &
{\small (B) Weak Surrogates - $\bgamma$} &
{\small (C) Strong Surrogates - $\bbeta$} &
{\small (D) Strong Surrogates - $\bgamma$} \\
[2pt]
\includegraphics[width=.22\textwidth]{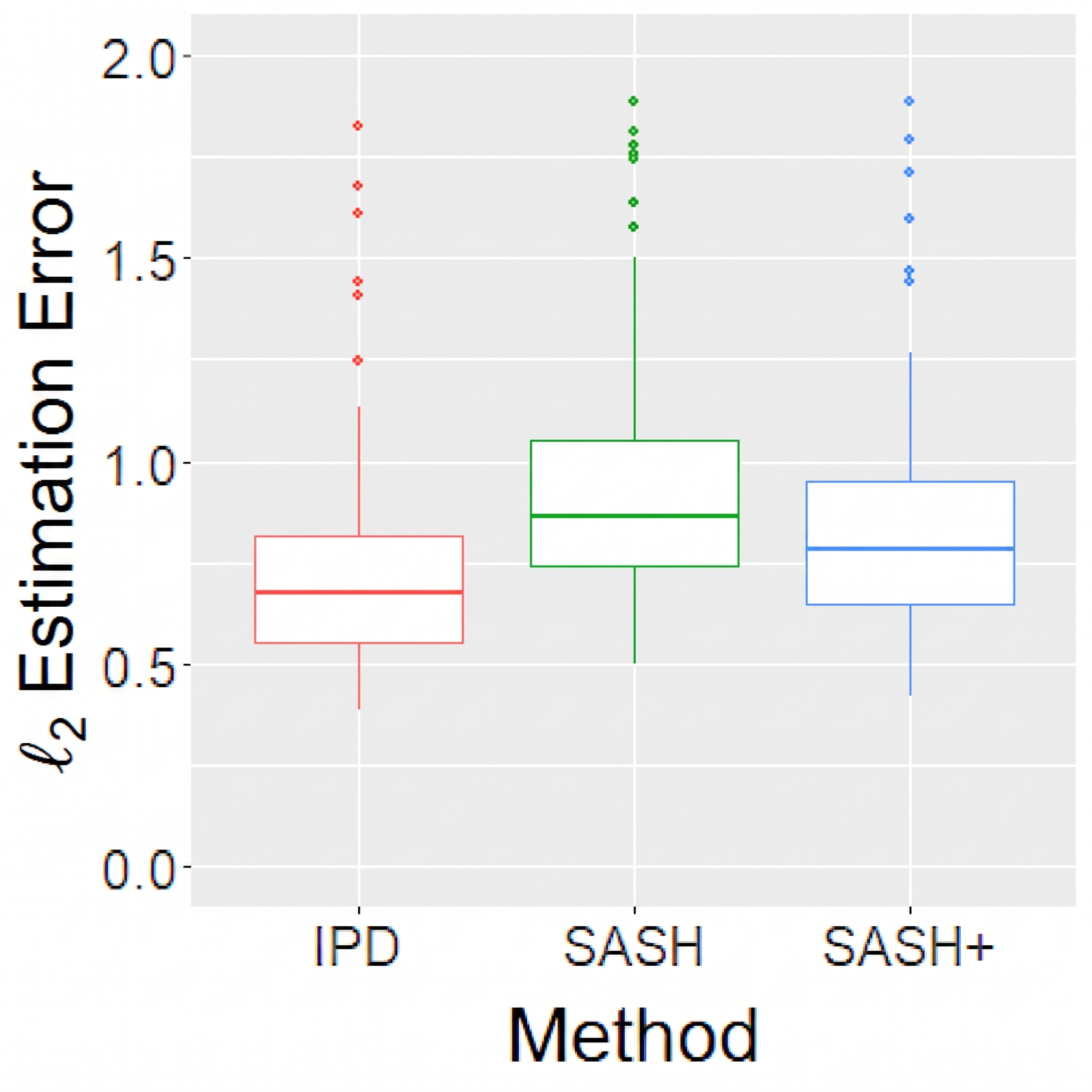} &
\includegraphics[width=.22\textwidth]{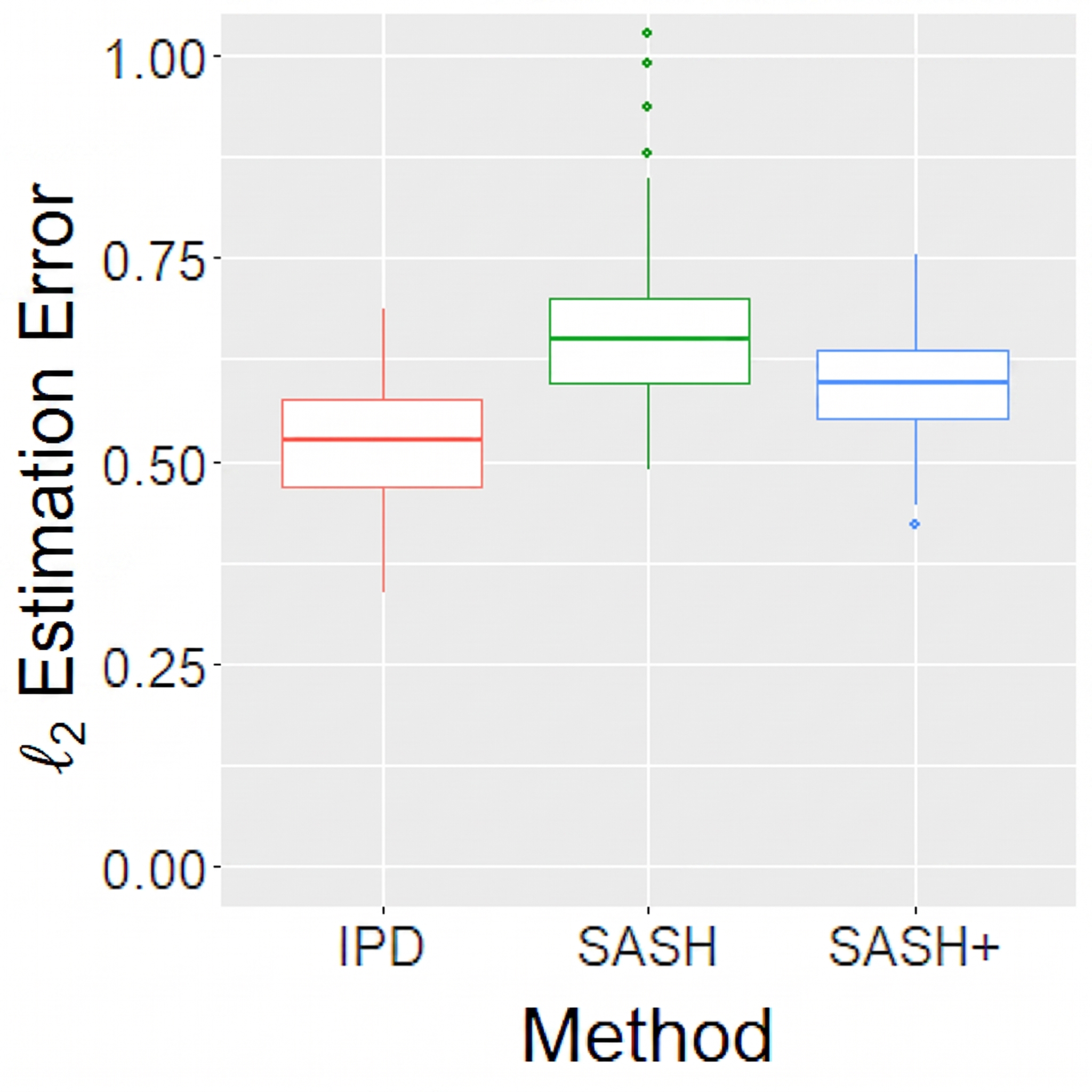} &
\includegraphics[width=.22\textwidth]{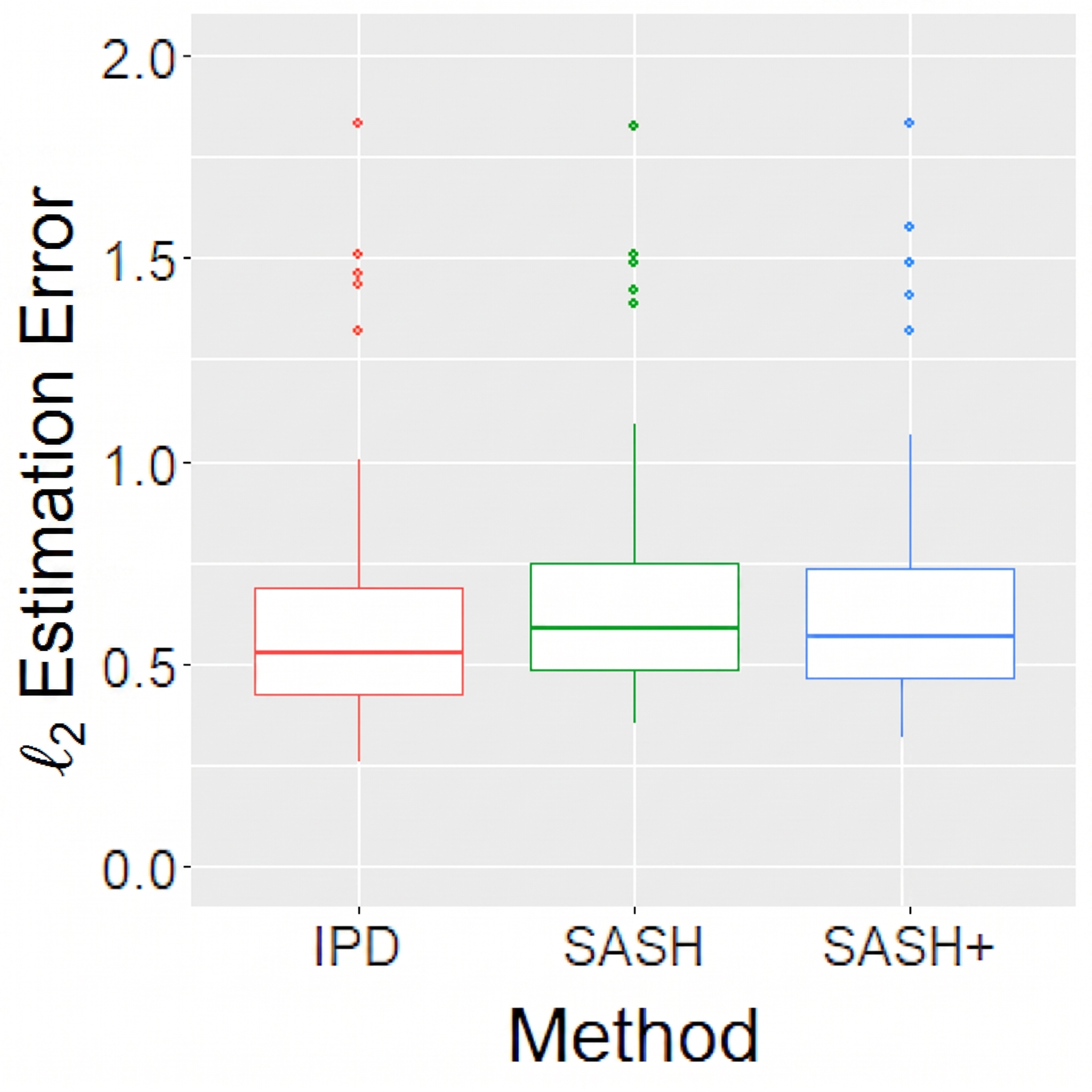} &
\includegraphics[width=.22\textwidth]{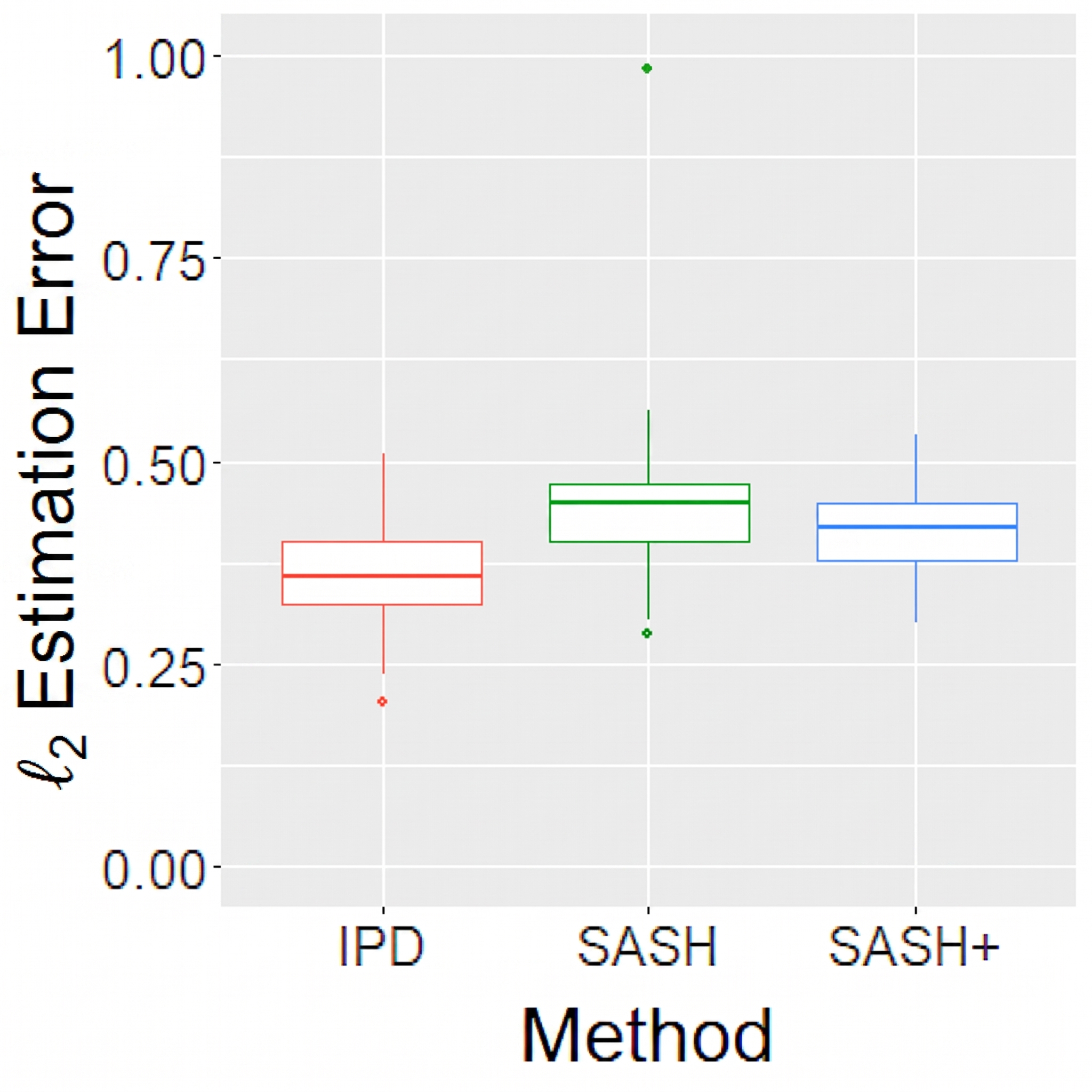}
\end{tabular}
 \vspace{-10pt}
\caption{Boxplots of the $\ell_2$ estimation errors of $\bbeta$ and $\bgamma$ under the strong surrogate and weak surrogate settings with $N=8000$, $M=4$, $n=200$, and $p=300$. The results are based on 200 simulation replications.}
\label{fig:error:sash}
\end{figure}

{

\subsection{Sensitivity analysis on tuning parameter $h$}\label{sec:analysis:h}

In this section, we perform a sensitivity analysis on the tuning parameter $h$. This tuning parameter plays a crucial role in our method and its selection is essential for obtaining reliable results. We compare the results obtained using three different values of $h$ taken $0.5h_0$, $h_0$, and $2h_0$, where $h_0:= \{{s \log (p \vee N_m)}/{N_m} \}^{1/5}$ is a fixed value of the tuning parameter chosen based on Theorem \ref{thm:local}. This is a wide enough range as the largest value is four times the smallest.

We generate the data following the strong surrogate setting introduced in Section \ref{sec:sim}, and implement our method with $h = 0.5h_0, h_0$ and $ 2h_0$. The resulting estimation errors are presented in Figure \ref{fig:h}. It indicates that there is no significant difference in the overall performance of the method under these different values of $h$. For example, the mean estimation error with $h=2h_0$ is $4.1\%$ smaller than that with $h=h_0$. In comparison, the mean error of IPD is $9.2\%$ smaller than that of SASH with $h=h_0$. Recall that the estimation error of all other methods (SS-uLasso, PASS, SS-RMRCE, pLasso, SL) is over $60\%$ larger than the estimation error of SASH as shown in Section \ref{sec:sim}. Therefore, the performance discrepancy due to different choices of $h$ (among a wide range) is much smaller compared to the discrepancy between our methods and other benchmark methods.

In conclusion, this sensitivity analysis shows that SASH is not sensitive to the choice of $h$. Although the choice of $h$ could affect the variation of estimation errors across different repetitions, it does not change the mean (overall) performance substantially. Therefore, for the main body of this paper, we use $h = h_0$ for simplicity and to ensure consistency across our analyses.


\begin{figure}[htb!]
	\begin{center}
		\begin{tabular}{c}
\includegraphics[width=0.42\textwidth,angle=0]{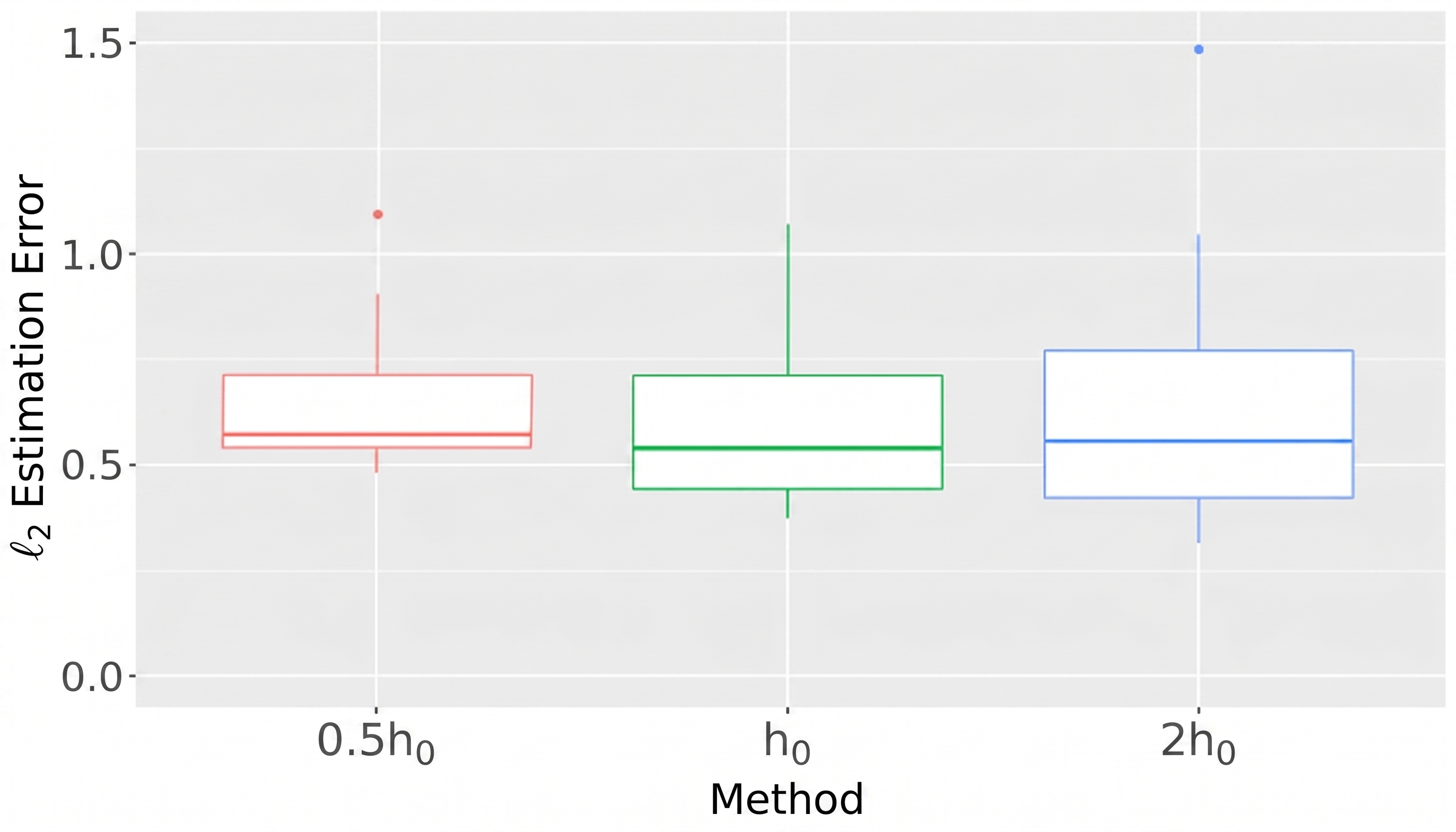} 
		\end{tabular}
	\end{center}
    \vspace{-20pt}
	\caption{Estimation error of SASH with different choices of $h$ under the strong surrogate setting with $N=8000$, $M=4$, $n=200$, and $p=300$, where $h_0$ is as defined in Appendix \ref{sec:analysis:h}.  It indicates that there is no significant difference in the overall performance of the method under these different values of $h$.} 
	\label{fig:h}
\end{figure}



}

\newpage

\subsection{Summary table for the selected features in our real-world study}\label{sec:table}
\renewcommand{\arraystretch}{0.95}

\begin{table}[htb!]
\centering
\resizebox{0.85\textwidth}{!}
{
\begin{tabular}{llllllll}
  \hline
 & SASH & SASH-CI & PASS & pLasso & SL & SS-RMRCE & SS-uLasso \\ 
  \hline
Ethnicity (EUR) & -1.26 &  (-2.13, -0.39) & -1.72 & 0.00 & -1.00 & 0.21 & -1.25 \\ 
Gender (Female) & -0.50 &  (-0.85, -0.15) & -0.99 & -0.05 & 0.00 & 0.00 & 0.00 \\ 
rs35011184\_A & 0.29 &  (0.09, 0.49) & 0.67 & 0.00 & 0.00 & 0.00 & 0.50 \\ 
rs10811662\_A & -0.20 &  (-0.33, -0.06) & -2.58 & -0.11 & 0.00 & 0.00 & -0.30 \\ 
rs1564348\_C & 0.11 &  (0.03, 0.19) & -0.52 & -0.08 & 0.00 & 0.00 & 0.00 \\ 
rs7756992\_G & 0.11 &  (0.03, 0.19) & 0.34 & 0.00 & 0.00 & 0.00 & 0.00 \\ 
rs11649653\_G & -0.11 &  (-0.18, -0.03) & -0.23 & 0.00 & 0.00 & 0.00 & -0.12 \\ 
rs2902940\_G & 0.09 &  (0.03, 0.16) & 0.00 & -0.01 & 0.00 & 0.00 & 0.03 \\ 
rs1077514\_T & -0.09 &  (-0.15, -0.03) & -0.25 & -0.02 & 0.00 & 0.00 & 0.00 \\ 
rs10128711\_C & -0.08 &  (-0.13, -0.02) & 0.00 & 0.00 & 0.00 & -0.19 & 0.00 \\ 
rs3136441\_C & -0.07 &  (-0.12, -0.02) & 0.00 & 0.00 & 0.00 & 0.00 & 0.00 \\ 
rs2972146\_T & 0.07 &  (0.02, 0.12) & 0.29 & 0.00 & 0.00 & 0.00 & 0.00 \\ 
rs11869286\_C & -0.06 &  (-0.11, -0.02) & 0.00 & 0.00 & 0.00 & 0.00 & 0.00 \\ 
rs2131925\_T & -0.06 &  (-0.10, -0.02) & 0.00 & 0.02 & 0.00 & 0.00 & 0.00 \\ 
rs12328675\_C & -0.06 &  (-0.10, -0.02) & 0.00 & 0.01 & 0.00 & 0.00 & 0.00 \\ 
rs38855\_G & -0.06 &  (-0.09, -0.02) & 0.00 & 0.00 & 0.00 & 0.00 & 0.00 \\ 
rs2328223\_C & 0.05 &  (0.02, 0.09) & 0.00 & 0.00 & 0.00 & 0.00 & 0.07 \\ 
rs11151789\_C & -0.05 &  (-0.09, -0.02) & 0.00 & 0.03 & 0.00 & 0.00 & 0.00 \\ 
rs17299838\_G & -0.05 &  (-0.09, -0.02) & 0.00 & -0.07 & 0.00 & 0.00 & -0.18 \\ 
rs4846914\_A & -0.04 &  (-0.07, -0.01) & -0.25 & -0.04 & 0.00 & 0.00 & 0.00 \\ 
rs2871865\_G & 0.04 &  (0.01, 0.07) & 0.20 & 0.00 & 0.00 & 0.00 & 0.22 \\ 
rs10150332\_C & 0.04 &  (0.01, 0.07) & 0.00 & 0.02 & 0.00 & 0.00 & 0.00 \\ 
rs571312\_A & 0.04 &  (0.01, 0.07) & 0.00 & 0.00 & 0.00 & 0.00 & 0.19 \\ 
rs4771122\_A & 0.04 &  (0.01, 0.06) & 0.00 & 0.00 & 0.00 & 0.00 & -0.05 \\ 
rs10187654\_T & 0.04 &  (0.01, 0.06) & 0.07 & 0.00 & 0.00 & 0.00 & 0.26 \\ 
rs622418\_A & -0.03 &  (-0.06, -0.01) & 0.00 & 0.00 & 0.00 & 0.00 & 0.00 \\ 
rs2606736\_T & -0.03 &  (-0.06, -0.01) & 0.00 & 0.00 & 0.00 & 0.00 & 0.00 \\ 
rs2890652\_C & 0.03 &  (0.01, 0.06) & 0.00 & 0.00 & 0.00 & 0.00 & 0.00 \\ 
rs4142995\_T & 0.03 &  (0.01, 0.05) & 0.00 & 0.00 & 0.00 & 0.00 & 0.00 \\ 
rs604109\_T & 0.03 &  (0.01, 0.05) & 0.00 & 0.04 & 0.00 & 0.00 & 0.00 \\ 
rs1260326\_C & 0.03 &  (0.01, 0.05) & -1.68 & -0.07 & 0.00 & 0.00 & 0.05 \\ 
rs2867125\_C & 0.03 &  (0.01, 0.05) & 0.00 & 0.00 & 0.00 & 0.00 & 0.05 \\ 
rs2922236\_C & -0.03 &  (-0.05, -0.01) & 0.00 & -0.04 & 0.00 & 0.00 & 0.00 \\ 
rs364585\_G & 0.03 &  (0.01, 0.05) & 0.00 & 0.00 & 0.00 & 0.00 & 0.00 \\ 
rs174546\_T & -0.02 &  (-0.04, -0.01) & -0.22 & 0.00 & 0.00 & 0.00 & 0.00 \\ 
rs2081687\_C & 0.02 &  (0, 0.03) & -0.13 & -0.02 & 0.00 & 0.00 & 0.00 \\ 
rs4650994\_A & -0.02 &  (-0.03, 0) & 0.00 & 0.02 & 0.00 & 0.00 & 0.00 \\ 
rs12444979\_T & -0.02 &  (-0.03, 0) & 0.00 & 0.00 & 0.00 & 0.00 & -0.08 \\ 
rs10901371\_A & -0.02 &  (-0.03, 0) & 0.00 & -0.02 & 0.00 & 0.00 & -0.15 \\ 
rs3184504\_C & 0.01 &  (0, 0.02) & 0.00 & 0.00 & 0.00 & 0.00 & 0.10 \\ 
rs698842\_T & 0.01 &  (0, 0.02) & 0.00 & 0.03 & 0.00 & 0.00 & 0.08 \\ 
rs7570971\_A & 0.01 &  (0, 0.02) & 0.09 & 0.00 & 0.00 & 0.00 & 0.00 \\ 
rs7134375\_A & -0.01 &  (-0.02, 0) & 0.00 & 0.00 & 0.00 & 0.00 & 0.00 \\ 
rs217386\_A & 0.01 &  (0, 0.02) & 0.00 & 0.00 & 0.00 & 0.00 & 0.00 \\ 
rs6864049\_G & 0.01 &  (0, 0.02) & -0.93 & -0.11 & 0.00 & 0.00 & 0.02 \\ 
rs10904908\_G & -0.01 &  (-0.01, 0) & 0.00 & 0.00 & 0.00 & 0.00 & -0.13 \\ 
  \hline
\end{tabular}
}
\caption{\label{tab:EHR2} Features assigned nonzero estimated coefficients by SASH, PASS, pLasso, SL, SS-RMRCE, or SS-uLasso. We also display the confidence intervals of SASH constructed using the procedure described in Section \ref{sec:method:ci}.}
\end{table}

\end{document}